\documentclass[amsmath,amssymb,pra,twocolumn,notitlepage,showkeys]{revtex4-1}
\usepackage{graphicx}
\usepackage{dcolumn}
\usepackage{bm}
\usepackage{etoolbox} 
\usepackage{booktabs}
\usepackage{color}
\newcounter{one}
\usepackage[hmargin=1in,vmargin=0.8in]{geometry}

\def\vect#1{\mbox{\boldmath $#1$}}
\newcommand{\ex}[1]{\langle #1 \rangle}
\newcommand{\bra}[1]{\langle #1 |}
\newcommand{\ket}[1]{| #1 \rangle}
\newcommand{\braket}[1]{\langle {#1} \rangle}

\newcommand{\Tr}[0]{ \mathrm{Tr}}

\newtheorem{theorem}{Theorem}

\newtheorem{lemma}{Lemma}

\newtheorem{corollary}{Corollary}

\def\QED{\mbox{\rule[0pt]{1.5ex}{1.5ex}}}

\def\endproof{\hspace*{\fill}~\QED\par\endtrivlist\unskip}

\newenvironment{proofof}[1]{\vspace*{5mm} \par \noindent
         {\bf Proof of #1:\hspace{2mm}}}{\endproof
}

\newcommand\calA{{\cal A}}

\newcommand\calC{{\cal C}}
\newcommand\calD{{\cal D}}
\newcommand\calE{{\cal E}}
\newcommand\calF{{\cal F}}

\newcommand\calH{{\cal H}}
\newcommand\calI{{\cal I}}

\newcommand\calK{{\cal K}}

\newcommand\calM{{\cal M}}
\newcommand\calN{{\cal N}}

\newcommand\calP{{\cal P}}
\newcommand\calQ{{\cal Q}}
\newcommand\calR{{\cal R}}

\newcommand\calU{{\cal U}}

\newcommand\calW{{\cal W}}

\newcommand{\tX}{\tilde{X}}
\newcommand{\tC}{\tilde{{\cal C}}}
\newcommand{\tD}{\tilde{\Delta}}
\newcommand{\tY}{\tilde{Y}}
\newcommand{\psikn}{\psi^{\perp}_k}
\newcommand{\psik}{\psi_k}
\newcommand{\dmul}{\delta_{\mathrm{multi}}}
\newcommand{\dmo}{\delta_{\mathrm{multi1}}}
\newcommand{\dmt}{\delta_{\mathrm{multi2}}}
\newcommand{\non}{\nonumber\\}
\newcommand{\Code}{\calE_{\mathrm{code}}}
\newcommand{\id}{\mathrm{id}}

\newcommand{\ew}{\epsilon_{\mathrm{worst}}}
\newcommand{\me}{\overline{\epsilon}}

\newcommand{\dm}[1]{\ket{ #1 }\bra{ #1 }}

\newcommand{\beq}{\begin{equation}}
\newcommand{\eeq}{\end{equation}}

\newcommand{\eq}[1]{\begin{align} #1 \end{align}}

\usepackage{hyperref}
\usepackage[dvipsnames]{xcolor}
\hypersetup{
    bookmarksnumbered=true, 
    unicode=false, 
    pdfstartview={FitH}, 
    pdftitle={}, 
    pdfauthor={}, 
    pdfsubject={}, 
    pdfcreator={}, 
    pdfproducer={}, 
    pdfkeywords={}, 
    pdfnewwindow=true, 
    colorlinks=true, 
    linkcolor=NavyBlue, 
    citecolor=NavyBlue, 
    filecolor=NavyBlue, 
    urlcolor=NavyBlue 
}

\newcommand{\HT}[1]{{\color{black} #1}}

\begin{document}
\title{Universal trade-off structure between symmetry, irreversibility, and quantum coherence in quantum processes}
\author	{Hiroyasu Tajima$^{1,2}$}
\email{hiroyasu.tajima@uec.ac.jp}
\author{Ryuji Takagi$^{3,4}$}
\author{Yui Kuramochi$^5$}
\affiliation{1. Department of Communication Engineering and Informatics, University of Electro-Communications, 1-5-1 Chofugaoka, Chofu, Tokyo, 182-8585, Japan}
\affiliation{2. JST, PRESTO, 4-1-8 Honcho, Kawaguchi, Saitama, 332-0012, Japan}
\affiliation{3. Department of Basic Science, The University of Tokyo, 3-8-1 Komaba, Meguro-ku, Tokyo 153-8902, Japan}
\affiliation{4. Nanyang Quantum Hub, School of Physical and Mathematical Sciences, Nanyang Technological University, 637371, Singapore}
\affiliation{5. Department of Physics, Kyushu University, 744 Motooka, Nishi-ku, Fukuoka, Japan}

\begin{abstract}
Symmetry, irreversibility, and quantum coherence are foundational concepts in physics. \HT{Here, we present a universal tradeoff relation between these three concepts. This particularly reveals that (1) under a global symmetry, any attempt to change the local conserved charge causes inevitable irreversibility, and (2) such irreversibility can be mitigated by quantum coherence. 
Our tradeoff relation follows solely from the unitarity and global symmetry of the total dynamics, allowing for general applicability.
For non-equilibrium physics, it relates the coherence cost and the entropy production---representing thermodynamic irreversibility---in arbitrary quantum processes.
It also provides fundamental limitations on the capability of a number of quantum information processing tasks---such as gate and measurement implementation and error correction---that involve symmetry restrictions.
Furthermore, it predicts how many bits of classical information thrown into a black hole become unreadable under energy conservation. 
Our tradeoff relation is based on quantum uncertainty relation, showcasing intimate connections between fundamental physical principles and ultimate operational capability of quantum processes.}
\end{abstract}

\maketitle

\HT{
\section{Introdcution}}
Symmetry, irreversibility, and quantum coherence  play central roles in physics.
Symmetry not only serves as a guiding principle in high-energy \cite{georgi2000lie} and condensed matter \cite{review_condensed} physics but places strong constraints on quantum information processing, such as measurement \cite{Wigner1952,Araki-Yanase1960,OzawaWAY,Korzekwa_thesis,TN}, gate implementation \cite{ozawaWAY_CNOT,Karasawa_2009,TSS,TSS2}, and error correction \cite{Eastin-Knill,e-EKFaist,e-EKKubica,e-EKZhou}.
Irreversibility generically appears when a large number of particles interact with each other, including the settings such as thermodynamics \cite{Carnot1824} and statistical mechanics \cite{Shiraishi_text}. It is also a central concept in quantum error correction \cite{QEC_text}, which aims to protect quantum data and resources from irreversible changes caused by external noise.
Quantum coherence, also known as superposition, is the core of quantum mechanics and is a necessary resource for quantum advantages in numerous tasks, such as computation \cite{Shor,Grover}, communication \cite{comm1,comm2}, sensing \cite{sensing1,sensing2}, and engines \cite{TF}.

At first sight, these fundamental concepts may appear to be independent notions.
In this paper, we show that they are actually intimately related---we provide a general quantitative relation that represents their universal trade-off structure.
Our trade-off relation reveals that dynamics changing the local conserved charge equipped with the given symmetry must be irreversible.
Furthermore, the required degree of irreversibility is inversely proportional to the amount of coherence between different charge sectors.

Our trade-off relation holds whenever the total dynamics obeys the unitarity and a continuous symmetry and is applicable to the standard irreversibility measures in various settings---ranging from quantum thermodynamics to quantum error correction---allowing for a number of applications.
In the context of quantum thermodynamics, our result gives a lower bound on the required coherence to realize an arbitrary quantum process, showing how quantum coherence suppresses thermodynamic irreversibility. 

Our trade-off bound also gives a grand unification of the restrictions on quantum information processing imposed by symmetry. 
It unifies the Wigner-Araki-Yanase (WAY) theorem for measurements \cite{Wigner1952,Araki-Yanase1960,OzawaWAY,Korzekwa_thesis,TN} and unitary gates \cite{ozawaWAY_CNOT,Karasawa_2009,TSS,TSS2}---which restricts the implementation of these processes under a symmetry constraint---and the quantitative Eastin-Knill theorem for error correcting codes \cite{Eastin-Knill,e-EKFaist,e-EKKubica}---which provides the unavoidable decoding error for covariant error-correcting codes. 
This unification is a reflection of the fact that errors in information processing can be regarded as a type of irreversibility.
Our result not only unifies the known results but adds further insights.
For instance, we (1) give the quantitative WAY theorem in terms of the gate fidelity error, (2) obtain quantitative restrictions in implementing quantum operations beyond unitary gates, and (3) provide fine-grained restrictions on covariant error-correcting codes with respect to arbitrary input states.
In addition, our results put constraints on the recovery operation by the Petz recovery map \cite{Petz_map_definition,wilde2015,junge2018}.


Our trade-off relation further provides insights into the problem of information scrambling, which has been a major topic in quantum chaos and black hole physics. 
In quantum many-body systems, it is known that information in a local system is scrambled and embedded into the global system, where the relationship between symmetry and information scrambling has been under active study \cite{JLiu,Yoshida-soft,Nakata,TS}.
Our trade-off bound provides a perspective to this problem by revealing how symmetry affects the scrambling of \textit{classical} information. 
In particular, we establish the limitation of recovering the classical information after information scrambling under conservation laws. For example, our result implies that when $m$ bits of classical information are thrown into an energy-preserving black hole, at least about $m/4$ bits will be unrecoverable until 99 percents of the black hole evaporates.

Notably, all of the above applications to quantum thermodynamics, quantum information processing, and black hole physics can be derived as direct consequences from a single trade-off theorem, showcasing its high generality and vast potential for future applications. 
Our trade-off relation is based on the quantum uncertainty relation \cite{Kennard_unc,Robertson_unc,Luo_unc,Frowis_unc,TN}, where our results provide explicit examples of which uncertainty principle imposes nontrivial constraints on the irreversibility of quantum dynamics. 

\HT{
\section{Results}}
\subsection{Framework}
This paper aims to clarify how the irreversibility of quantum processes is affected by symmetry and coherence.
To achieve this goal, we first introduce a framework for treating various types of the irreversibility of quantum processes simultaneously.
As discussed later, our formulation is directly applicable to various settings, including quantum thermodynamics, quantum error correction, quantum measurements and black hole physics.

We consider two quantum systems, $A$ and $B$, as represented in \HT{Fig.} \ref{setup}. The system $A$ is the system of interest, and its initial state is not fixed.
The system $B$ is another quantum system that works as an environment whose initial state is fixed to a quantum state $\rho_B$.
We perform a unitary operation $U$ on $AB$ and divide $AB$ into two systems, $A'$ and $B'$.
Then, the quantum process from $A$ to $A'$ is described as a completely positive trace preserving (CPTP) map $\calE(...):=\Tr_{B'}[U...\otimes\rho_BU^\dagger]$. 
When $U$ has a global symmetry described by a Lie group, the symmetry provides conserved quantities via Noether's theorem.
For simplicity, we focus on a single conserved quantity $X$ under the unitary operation. 
Namely, we assume that
\begin{align}
U^\dagger (X_{A'}+X_{B'})U=X_{A}+X_{B},\label{X-con}
\end{align}
where $X_{\alpha}$ is the local operator of the conserved quantity on the system $\alpha$ ($\alpha=A,B,A',B'$).
\HT{For more general Lie group symmetry case, see Supplementary \ref{subsecLie-extenxtion}.}

\begin{figure}[tb]
		\centering
		\includegraphics[width=.45\textwidth]{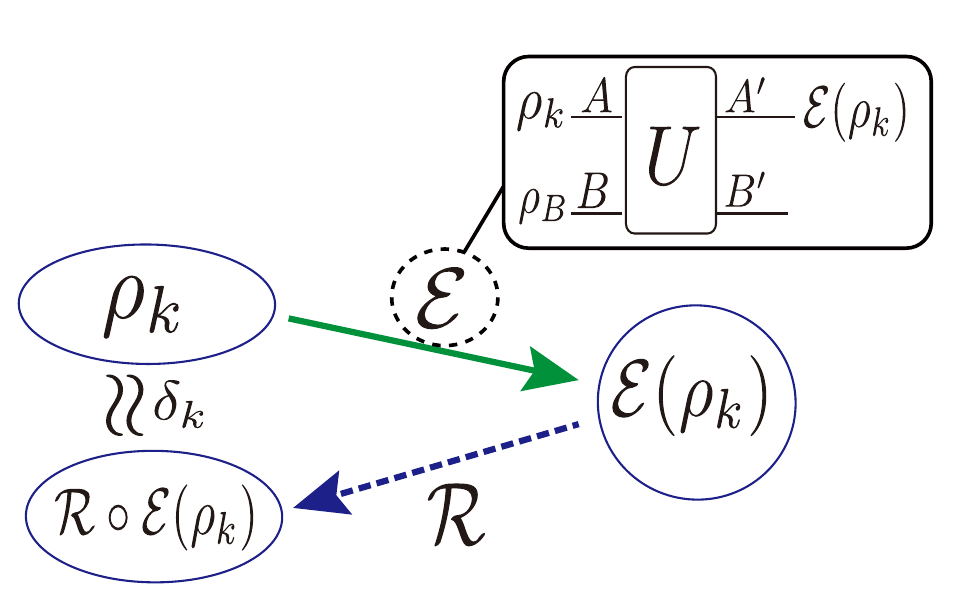}
		\caption{\HT{\textbf{Schematic diagram of the framework.}} We prepare the test states $\{\rho_k\}$ with probability $\{p_k\}$ and perform a CPTP map $\calE$ caused by a unitary interaction $U$. We try to recover the test states with a recovery CPTP map $\calR$ independent of $k$, and define the irreversibility of $\calE$ for the test ensemble $\{p_k,\rho_k\}$ as the average of recovery error for the optimal recovery map: $\delta:=\sqrt{\sum_kp_k\delta^2_k}$. We investigate the restriction on the irreversibility under the assumption that $U$ satisfies the conservation law \eqref{X-con}.}
		\label{setup}
\end{figure}

Now, let us define the irreversibility of the quantum process $\calE$. 
We prepare test states $\{\rho_k\}$ on $A$ with a probability $\{p_k\}$. We refer to the set $\{\rho_k,p_k\}$ as test ensemble.
The quantum process $\calE$ changes the test states.  After the process, we apply a CPTP map $\calR$ on $A'$ independent of $k$ and try to recover the test states of $A$ as accurately as possible.
We then define the irreversibility of $\calE$ for the test ensemble $\{\rho_k,p_k\}$ as the average of the recovery error of the best recovery map as follows:
\eq{
\delta&:=\min_{\calR:A'\rightarrow A}\sqrt{\sum_{k}p_k\delta^2_{k}}
\label{recovery error}\\
\delta_k&:=D_F(\rho_{k},\calR\circ\calE(\rho_k)).
}
Here $D_F$ is the purified distance defined as $D_F(\rho,\sigma):=\sqrt{1-F(\rho,\sigma)^2}$ and $F(\rho,\sigma):=\Tr[\sqrt{\sqrt{\sigma}\rho\sqrt{\sigma}}]$ is the Uhlmann fidelity.

\HT{The benefit of the irreversibility measure $\delta$ is that we can treat various concepts simultaneously via this measure.
First, it includes various types of measures of the irreversibility of $\calE$ as special cases.
For example, $\delta$ gives a lower bound for the entropy production, the standard measure of irreversibility in stochastic thermodynamics \cite{funo2018quantum}.
The irreversibility $\delta$ also includes the recovery error of the Petz recovery map \cite{Petz_map_definition,wilde2015,junge2018} and gives a lower bound for the entanglement fidelity error \cite{Watrous_text}, a standard measure of irreversibility in quantum error correction and quantum information scrambling. Furthermore, the irreversibility $\delta$ includes measures for various concepts other than irreversibility, e.g., (almost) arbitrary error and disturbance of quantum measurements and the Out-of-Time-Order Correlator (OTOC) \cite{ET2023}.
As we will see later, this property of $\delta$ gives our results a universality that can be applied to all of the various concepts mentioned above. For more details of the relations between $\delta$ and other various quantities, see Appendix \ref{prop_delta} and Supplementary Materials \ref{SIsubsecRelation}.}

\HT{As shown in the next section, the irreversibility $\delta$ is affected by quantum coherence with respect to the eigenbasis of the conserved quantity. To analyze the coherence effect quantitatively, we introduce the symmetric-logarithmic-derivative (SLD) quantum Fisher information (QFI) \cite{Helstrom,holevo2011} for the state family $\{e^{-iXt}\rho e^{iXt}\}_{t\in\mathbb{R}}$, which is a well-known measure of quantum coherence in the resource theory of asymmetry:
\begin{align}
\calF_{\rho}(X):=4\left.\frac{d^2 F(e^{-iX t}\rho e^{iXt},\rho)}{d t^2}\right|_{t=0}.\label{def_QFI}
\end{align}
The definition \eqref{def_QFI} says that the QFI $\calF_{\rho}(X)$ is a kind of speed of state change when the state $\rho$ evolves via the unitary time evolution $e^{-iX t}$.
Therefore, the QFI indicates how $\rho$ is non-commutative with $X$.
Because of this feature, the QFI is a good measure of quantum coherence (=asymmetry) on the eigenbasis of the conserved quantity $X$ \cite{skew_resource,Marvian_distillation,Takagi_skew,YT,Kudo_Tajima}. In other words, the QFI indicates the amount of superposition between eigenstates of $X$ whose eigenvalues are different from each other.
It is also a standard measure of the amount of the quantum fluctuation of $X$ \cite{min_V_Petz,min_V_Yu,Luo,Hansen,Marvian_distillation,Kudo_Tajima}.
This feature is also natural from \eqref{def_QFI}, because when $\rho$ changes fast with $e^{-iX t}$, the state $\rho$ has a lot of quantum fluctuation of $X$. 
For more detailed features of this quantity, see Appendix \ref{App_RTA}.

Next, we define a key quantity to describe the fundamental limitation of irreversibility, which is an indicator of change of local charge. We first introduce the ``work operator'' $Y:=X_{A}-\calE^\dagger(X_{A'})$ corresponding to the change of the local conserved charge caused by the quantum process $\calE$.
Here $\calE^\dagger$ is the dual map of $\calE$ that satisfies $\ex{\calE^\dagger(W)}_{\rho}=\ex{W}_{\calE(\rho)}$ for any $\rho$ and $W$, where $\ex{A}_\sigma:=\Tr(A\sigma)$ refers to the expectation value of an observable $A$ for a quantum state $\sigma$. By definition, the expectation value of the change of the local conserved quantity caused by $\calE$ is equal to the expectation value of $Y$. 
With this in mind, we introduce the following quantity 
\eq{
\calC:=\sqrt{\sum_{k\ne k'}p_kp_{k'}\Tr[(\rho_{k}-\rho_{k'})_{+}Y(\rho_{k}-\rho_{k'})_{-}Y]}.
}
Here, $(\rho_k-\rho_{k'})_{\pm}$ is the positive/negative part of $\rho_k-\rho_{k'}$.
We can interpret $\calC$ as the degree of the change of local charge, since whenever $Y\not\propto I_A$ holds, i.e., the change of the local conserved quantity caused by $\calE$ is not just a shift of its origin, $\calC>0$ holds at least one test ensemble.
In particular, when the set of the test states are pure states $\{\ket{\psi_k}\}$ orthogonal to each other,  $\calC$ becomes the sum of the absolute values of the non-diagonal elements of $Y$: $\calC=\sqrt{\sum_{k\ne k'}p_{k}p_{k'}|\bra{\psi_{k}}Y\ket{\psi_{k'}}|^2}$. 
Furthermore, when the test states $\{\rho_k\}$ are orthogonal to each other, we can give upper and lower bounds of $\calC$ written by the convexity of the QFI of the operator $Y$: $\mathrm{C}_\calF:=\sum_kp_k\calF_{\rho_k}(Y)-\calF_{\sum_kp_k\rho_k}(Y)$, which implies that $\calC$ describes the gain in quantum fluctuations of the operator $Y$ when we know $k$ compared to when we do not know $k$.
See Appendix \ref{App_Canddelta} for details of this feature and other properties of $\calC$.
}

\subsection{Main Results}
We are now in a position to establish a general structure between symmetry, irreversibility, and coherence.
To capture the essence, we first consider the case where the test states are orthogonal to each other, i.e., $F(\rho_k,\rho_{k'})=0$ for $k\ne k'$.
In this case, the following trade-off inequality holds:
\eq{
\frac{\calC}{\sqrt{\calF}+\Delta}\le\delta.\label{SIQ-Cini}
}
We show \eqref{SIQ-Cini} in Supplementary Materials \ref{SIsec2}.
When the test ensemble is in the form of $\{1/2,\rho_k\}_{k=1,2}$, we can make \eqref{SIQ-Cini} tighter by substituting $\sqrt{2}\calC$ for $\calC$.
Here, $\Delta$ is a positive quantity defined by
\eq{
\Delta:=\max_{\rho\in\cup_k\mathrm{supp}(\rho_k)}\sqrt{\calF_{\rho\otimes\rho_B}(\tilde{Y})},\label{def_Delta}
}
where $\tilde{Y}:=X_A\otimes 1_B-U^\dagger X_{A'}\otimes1_{B'}U$, and the maximization runs over the subspace spanned by the supports of the test states $\{\rho_k\}$. 
We remark that $\Delta$ has several upper bounds, e.g., $\Delta\le\Delta_1:=\Delta_{X_A}+\Delta_{X_{A'}}$, where $\Delta_{W}$ is the difference between the maximum and minimum eigenvalues of $W$. Therefore, we can substitute these upper bounds for $\Delta$ in \eqref{SIQ-Cini}. 
For details, see Appendix \ref{App_Canddelta}.

The inequality \eqref{SIQ-Cini} shows a close relationship between the global symmetry of dynamics $U$, the irreversibility of the process $\calE$, and the coherence in $B$.
This result implies the following two consequences.
First, it shows that when $\calC$ is finite, the CPTP map $\calE$ cannot be reversible.
We remark that when $Y\not\propto I_A$, i.e., the change of the local conserved quantity caused by $\calE$ is not just a shift of its origin, $\calC>0$ holds at least one test ensemble.
Therefore, when local dynamics $\calE$ changes the local charge nontrivially (i.e. when the change is not just a shift of origin), the local dynamics will be irreversible.

Second, the irreversibility of $\calE$ can be mitigated by the quantum coherence in $B$. For example, when $B$ has no quantum coherence, the irreversibility $\delta$ must be larger than $\calC/\Delta$.
On the other hand, when quantum coherence is present in the system $B$, the lower bound can be smaller than $\calC/\Delta$. We remark that there are concrete examples in which $\delta<\calC/\Delta$ holds when $\calF$ is large (see Supplementary Materials \ref{SIsubsec_suppression}). 
Thus, the inequality \eqref{SIQ-Cini} implies the suppression of irreversibility by coherence.
These findings are summarized as follows: \textit{Under a global symmetry, any attempt to change the local charge causes irreversibility. However, we can mitigate the irreversibility by using coherence of the conserved charge.}

The above trade-off structure also holds for the general test states $\{\rho_k\}$. In the general case, the following inequality holds:
\eq{
\frac{\calC}{\sqrt{\calF}+\Delta}\le\sqrt{\delta}.\label{SIQ-Gini}
}
We prove \eqref{SIQ-Gini} in Supplementary Materials \ref{SIsec2}.
Note the square root on the right-hand side in comparison to \eqref{SIQ-Cini}.
Again, when the test ensemble is in the form of $\{1/2,\rho_k\}_{k=1,2}$, we can make \eqref{SIQ-Gini} tighter by substituting $\sqrt{2}\calC$ for $\calC$.
Eq.~\eqref{SIQ-Gini} shows that the trade-off structure is still present even if the test states have no restriction.
When $\calC$ is finite, the quantum process cannot be reversible. And the irreversibility can be alleviated by quantum coherence.
A major difference between \eqref{SIQ-Cini} and  \eqref{SIQ-Gini} is in their scopes of application. 
Since \eqref{SIQ-Gini} does not impose any assumption on the test states, it is applicable to an even greater variety of irreversibility measures. 
For example, we will see in Section \ref{application} and Appendix \ref{prop_delta} that \eqref{SIQ-Gini} provides general bounds for the entropy production of thermal operations and the recovery error of the Petz recovery map.

We remark that the above results can be extended to the case where the conservation law \eqref{X-con} is violated. 
In this case, we define a Hermitian operator $Z$ that describes the degree of violation of the conservation as $Z:=U^\dagger (X_{A'}+X_{B'})U-(X_{A}+X_{B})$.
Then, we can extend the inequalites \eqref{SIQ-Cini} and \eqref{SIQ-Gini} by the following change:
\eq{
\calC\rightarrow\calC-\frac{\Delta_Z}{2}\enskip\mbox{and}\enskip \Delta\rightarrow\Delta+\Delta_Z.\label{correction}
}
The correction by \eqref{correction} shows that when the global symmetry is violated, our trade-off bound becomes weaker with the magnitude of the violation. In the extreme case where the global dynamics have no symmetry and $\Delta_Z$ gets large, the inequality is trivialized as $\calC-\Delta_Z/2$ becomes negative. For details, see Supplementary Materials \ref{SI_violated_SIQ}.

To close this subsection, we remark on an important implication of the main results. As we have remarked in the previous subsection, the irreversibility measure $\delta$ gives lower bounds for the measures of various concepts from the thermodynamic irreversibility to the OTOC.
This fact means that under the global continuous symmetry \eqref{X-con}, all of such measures obey the same-formed inequalities as \eqref{SIQ-Cini} and \eqref{SIQ-Gini}. 
Indeed, just by substituting such measures for $\delta$ in \eqref{SIQ-Cini} and \eqref{SIQ-Gini}, we can obtain applications to various fields. 
We see several examples of such applications in the next subsection. We also remark that many of these applications are experimentally verifiable. For this point, see the Supplementary Materials \ref{SI_violated_SIQ}.

\begin{figure*}[tb]
		\centering
		\includegraphics[width=0.9\textwidth]{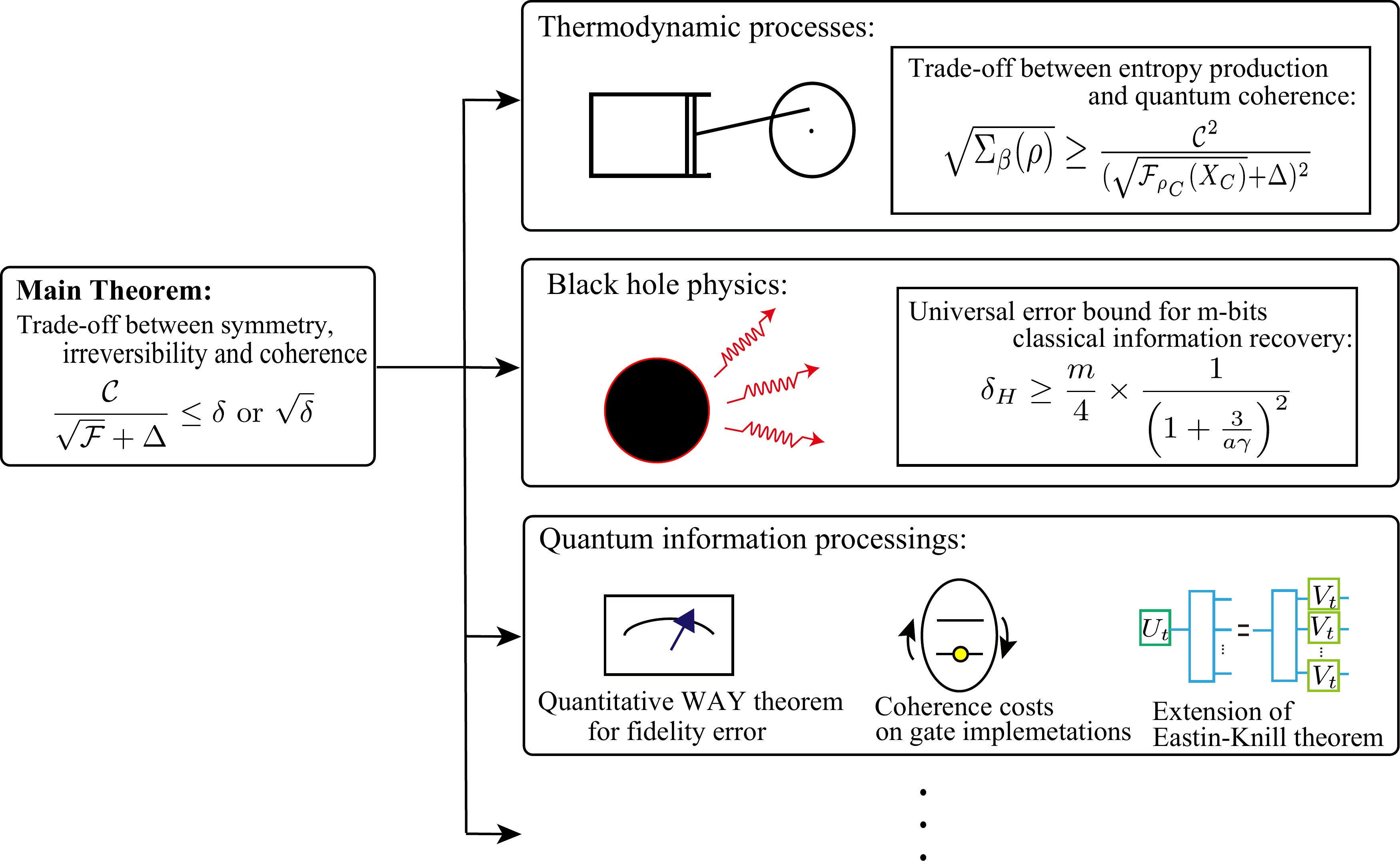}
	\caption{\HT{\textbf{Schematic diagram of the logical relationship between the main results and applications.}} The arrow indicates that the tip is a corollary of the root. As shown in the figure, our results are applicable to black hole physics, quantum error-correcting codes, quantum measurements, gates implementations, and quantum thermodynamics. We remark that there are still more applications besides those depicted in this diagram. For example, we give a restriction on Petz map recovery, and coherence cost for arbitrary channels under thermodynamic setups.}
		\label{diagram} 
	\end{figure*}

\subsection{Applications}\label{application}

\HT{As we remarked in the subsection of main result, the irreversibility measure $\delta$ recovers various measures of irreversibility and other concepts, and thus we can obtain insights into various fields from our main results.
In this subsection, we show applications to quantum thermodynamics, quantum information processing and black hole physics as examples (\HT{Fig.} \ref{diagram}).}
We remark that all of these applications can be obtained as direct consequences of \eqref{SIQ-Cini} and \eqref{SIQ-Gini} with suitable choices of $\calE$, $X$ and the test ensemble $\{p_k,\rho_k\}$.

\subsubsection{Non-equilibrium physics: universal tradeoff between coherence and entropy production}\label{thermo}

In thermodynamic settings, one often wants to interact heat reservoirs and batteries with a system to produce the desired dynamics $\calN$ in the system.
In such cases, the time evolution of the whole system is unitary and conserves energy. Therefore, our results can be used directly in this setup.
For example, consider a three-body system containing a heat reservoir $R$, a target system $S$, and some battery $C$ (\HT{Fig.} \ref{thermo_setup}). 
The battery can be a work battery, a coherence battery, a catalyst, or a combination of the three. At this point, by considering $X$ as energy, $S$ as $A$, and $RC$ as $B$, we can apply  \eqref{SIQ-Gini} to this setup. Then, $\calF$ in  \eqref{SIQ-Gini} describes the amount of energetic coherence in $RC$.
Since the heat reservoir $R$ is in the Gibbs state and has no energetic coherence, $\calF$ in \eqref{SIQ-Gini} is the amount of coherence in $C$. 
We remark that the discussion here is valid even for the case of multiple heat reservoirs.

We now derive two restrictions on thermodynamic processes from \eqref{SIQ-Gini}.
We first link the amount of coherence in $C$ to the thermodynamic irreversibility of the realized channel $\calN$, i.e., entropy production.
The entropy production is a standard measure of thermodynamic irreversibility. It is defined for Gibbs preserving maps, which are CPTP maps that do not change the Gibbs state at a specific inverse temperature $\beta$. The Gibbs preserving maps include all isothermal processes.
For a Gibbs preserving map $\calN$, the entropy production is defined as follows:
\eq{
\Sigma_{\beta}(\rho):=\Delta S(\rho)-\beta Q(\rho),\label{def_entropy}
}
where $\Delta S(\rho):=S(\calN(\rho))-S(\rho)$ and $Q(\rho):=\ex{\calN^\dagger(H)-H}_{\rho}$ are the increases of the von Neumann entropy and the energy of the target system.
Here, $H$ is the Hamiltonian of the target system.
The entropy production corresponds to the total entropy increase in the target system and the bath after the total system is thermalized.
As we see in Appendix \ref{prop_delta}, \eqref{SIQ-Gini} directly implies a universal trade-off relation between the entropy production and the coherence in $C$: 
\eq{
\sqrt{\Sigma_\beta(\rho)}\ge\frac{4\calC^2}{(\sqrt{\calF_{\rho_C}(X_C)}+\Delta)^2}\label{motomoto}
}
This inequality shows that the mitigation effect on irreversibility by quantum coherence is valid, even when the irreversibility is thermodynamic irreversibility.

When the process $\calN$ is an arbitrary CPTP map, the entropy production is not well defined in general. Even in that case, we can define the generalized entropy production, another standard measure of irreversibility: $\Sigma_{\calN,\rho,\sigma}:=D(\rho\|\sigma)-D(\calN(\rho)\|\calN(\sigma))$, where $D(\rho_1\|\rho_2):=\Tr(\rho_1\log\rho_1)-\Tr(\rho_1\log\rho_2)$ is the quantum relative entropy. When $\calN$ is Gibbs preserving and $\sigma$ is the Gibbs state, $\Sigma_{\calN,\rho,\sigma}$ becomes $\Sigma_\beta(\rho)$.
As we see in the Materials and Methods, we can also give another trade-off relation, in which $\Sigma_{\calN,\rho,\sigma}$ is substituted for $\Sigma_\beta(\rho)$ in \eqref{motomoto}.
In that case, \eqref{motomoto} gives a universal lower bound for the necessary coherence amount to realize the given arbitrary channel $\calN$ as follows:
\eq{
\calF_{\rho_C}(X_C)\ge\frac{4\calC^2}{\sqrt{\Sigma_{\calN,\rho,\sigma}}}-\Delta^2.
}

\begin{figure}[tb]
		\centering
		\includegraphics[width=.5\textwidth]{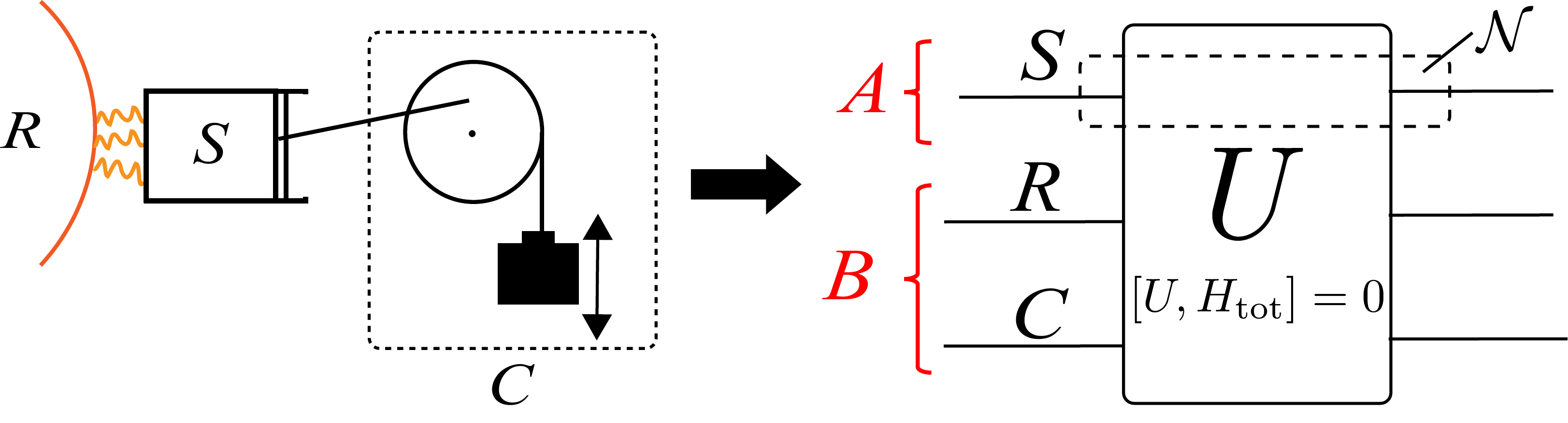}
		\caption{\HT{\textbf{Schematic diagram of the thermodynamic setup.}} We consider the situation where the target system $S$ (e.g. a cylinder of a heat engine) interacts with a heat reservoir $R$ and another system $C$. Here, $C$ can be a coherence battery, a catalyst or a work storage, etc. We assume the initial state of $R$ is a Gibbs state, and that the dynamics of the total system is a unitary channel $U$ satisfying the energy conservation law $[U,H_{\mathrm{tot}}]=0$, where $H_{\mathrm{tot}}=H_S+H_R+H_C$ and $H_{...}$ is the Hamiltonian of each system. Then, we can apply our trade-off relation \eqref{SIQ-Gini} to this situation by considering $S$ as $A$ and $RC$ as $B$.} 
		\label{thermo_setup}
	\end{figure}

\subsubsection{Quantum information processing: a grand unification of symmetry-induced limitations}
Our results provide a unified understanding of a \HT{topic with a long history:} the effect of symmetry on quantum information processing.
This investigation along this line has given fundamental symmetry-induced limitations on various information processing:
\vspace{-.5\baselineskip}
\begin{itemize}
\setlength{\parskip}{0cm} 
  \setlength{\itemsep}{0.1cm} 
\item{\textit{The WAY theorems for measurements \cite{Wigner1952,Araki-Yanase1960,OzawaWAY,Korzekwa_thesis,TN}:} When one tries to implement a projective measurement that does not commute with the conserved charge, the implementation error must be inversely proportional to the quantum fluctuation of the charge in the measurement apparatus.}
\item{\textit{The WAY theorems for unitary gates \cite{ozawaWAY_CNOT,Karasawa_2009,TSS,TSS2}:} The above restriction also holds when one tries to implement a unitary operation which does not commute with the conserved charge.}
\item{\textit{The Eastin-Knill theorems for error correcting codes \cite{Eastin-Knill,e-EKFaist,e-EKKubica}:} In covariant codes for continuous symmetry, the recovery error must be inversely proportional to the code size.}
\end{itemize}
As we show in the Supplementary Materials \ref{SIsecQIP}, the main result \eqref{SIQ-Cini} provides a universal lower bound for the coherence cost to implement an arbitrary channel $\calN$ in a standard setting in the resource theory of asymmetry \cite{Gour2008resource,Marvian_thesis,skew_resource,Takagi_skew,Marvian_distillation,YT,YT2,Kudo_Tajima,WAY_RToA1,WAY_RToA2,TN,TSS,TSS2,TS,e-EKZhou,e-EKYang,Liu1}, which recovers all of the above limitations:
\eq{
\sqrt{\calF^{\mathrm{cost}}_{\calN}}\ge\frac{\calC}{\delta}-\Delta.\label{cost-C}
}
Here, $\calF^{\mathrm{cost}}_{\calN}$ is the coherence cost to implement $\calN$, defined as follows:
\eq{
\calF^{\mathrm{cost}}_{\calN}:=&\min_{\mbox{$\calI$ realizes $\calN$}}\calF_{\rho_B}(X_B)
}
where $\calI:=(\rho_B,X_B,X_{B'},U)$ runs over implementations of $\calN$ which satisfy $\calN(...)=\Tr_{B'}[U(...\otimes \rho_B)U^\dagger]$ and $X_A+X_B=U^\dagger (X_{A'}+X_{B'})U$.
Since \eqref{cost-C} itself is derived from \eqref{SIQ-Cini}, the above limitations are special aspects of a single inequality \eqref{SIQ-Cini}  (\HT{Fig.} \ref{unification_theorem}).

\begin{figure}[tb]
		\centering
		\includegraphics[width=.48\textwidth]{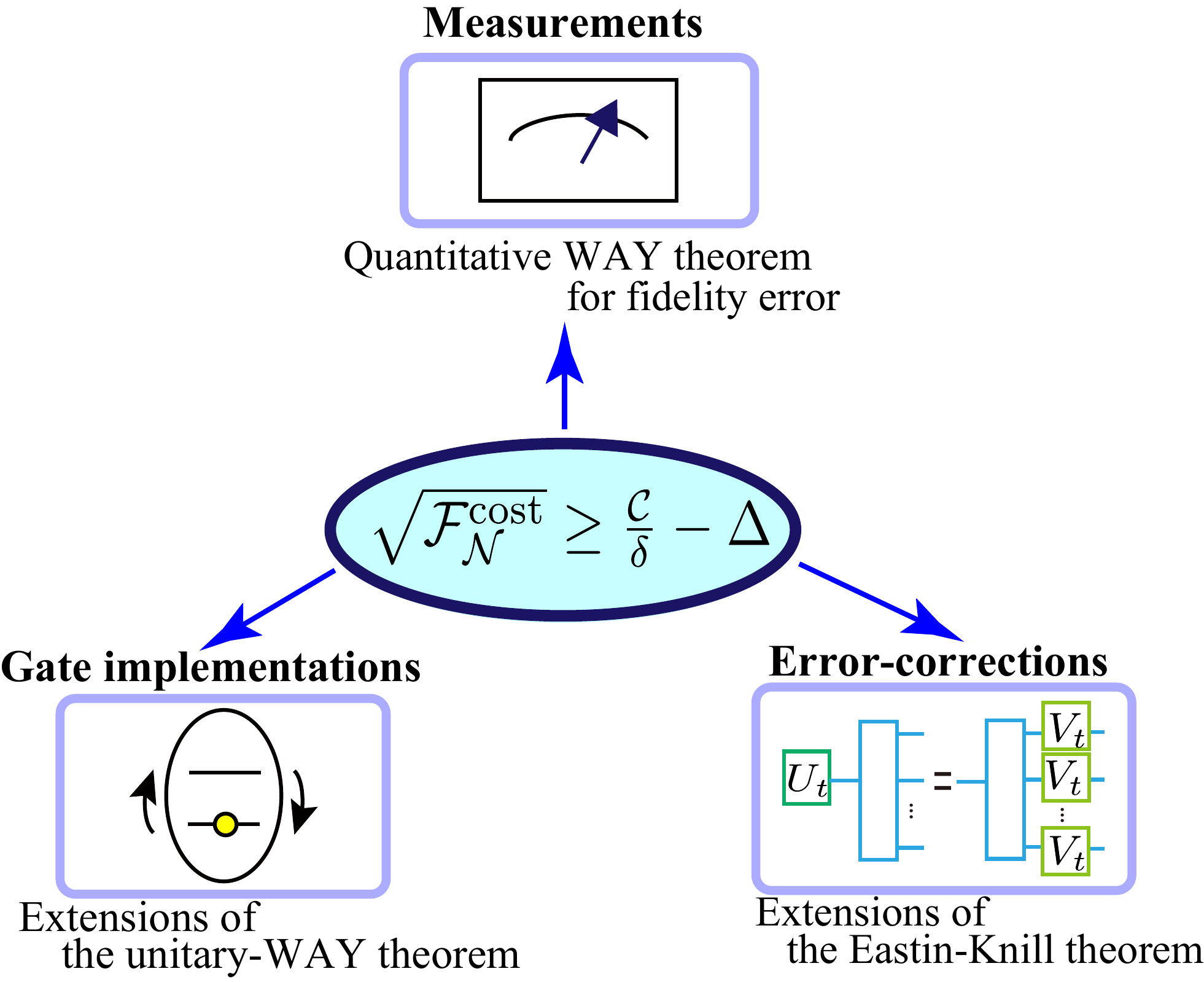}
		\caption{\HT{\textbf{The cost-irreversibility trade-off \eqref{cost-C} is a unification bound on fundamental restrictions of quantum information process imposed by symmetry.}}
		It unifies the WAY theorem for measurements, the coherence-error trade-off relations for unitary and other quantum gates, and the Eastin-Knill theorem for error correcting codes.
		These theorems can be interpreted as  special aspects of the unification bound \eqref{cost-C}.
		The unification bound also extends these restrictions on each quantum information processing.  
		It is remarkable that the unification bound \eqref{cost-C} itself is one of the corollaries of the main result \eqref{SIQ-Cini}.}
		\label{unification_theorem}
	\end{figure}

The inequality \eqref{cost-C} unifies and extends the limitations known in the settings mentioned above. 

\textit{Quantum measurement: a quantitative \HT{WAY} theorem based on fidelity error---}
The latest understanding of the WAY theorem on measurements is that there exist universal trade-off relations between coherence cost of measurement and measurement error  \cite{Korzekwa_thesis,TN}.
However, the known trade-off relations are restricted to so called Ozawa's error \cite{Ozawaerror}, which is defined as the expectation value of the square of the noise operator.
There are many other definitions for measurement errors, and the quantitative WAY theorems in terms of such error quantifiers still have been desired.

The inequality \eqref{cost-C} provides such a theorem: a quantitative WAY theorem for the gate fidelity error.
Let $\calQ$ be a projective measurement channel $\calQ(...):=\sum_{k\in\calK}\Tr[Q_k...]\ket{k}\bra{k}$ and $\calP$ be a measurement channel $\calP(...):=\sum_{k\in\calK}\Tr[P_k...]\ket{k}\bra{k}$. 
We assume that the projective measurement channel $\calQ$ is approximated by $\calP$, and define the fidelity-based approximation error as $\epsilon_{\mathrm{meas}}:=\max_{\rho}D_F(\calP(\rho),\calQ(\rho))$.
Then, as shown in the Supplementary Materials \ref{SIsubsec_WAY}, \eqref{cost-C} provides a lower bound for the implementation cost of $\calP$ under conservation law of $X$ as follows:
\eq{
\sqrt{\calF^{\mathrm{cost}}_{\calP}}\ge\max_k\frac{\sqrt{2}\|[X_A,Q_k]\|_{\infty}}{\epsilon}-\Delta'.
}
Here $\Delta':=\Delta_{X_A}+2\Delta_{X_{A'}}$.

\textit{Gate implementation: error-coherence tradeoffs beyond unitary gates---}
The inequality \eqref{cost-C} contributes to the studies of symmetry-induced limitations on gate implementations in two directions.
First, it extends the WAY theorems for unitary gates from the entanglement-gate-fidelity error to the gate-fidelity error (the Supplementary Materials \ref{SIsubsec_unitaryWAY}).
Second, the inequality \eqref{cost-C} also restricts the implementations of non-unitary gates (the Supplementary Materials \ref{SIsubsec_non-unitaryWAY}).
In fact, it provides the following no-go theorem:
\textit{ Let $U$ be a unitary and $\calN$ be a channel. If there exist two orthogonal eigenstates $\ket{x_{1,2}}$ of $X$ such that $\braket{x_1|U^\dagger X U|x_2}\neq 0$ and $\calN(\ket{x_{1,2}}\bra{x_{1,2}})=\ket{x_{1,2}}\bra{x_{1,2}}$, then $\calE=\calN\circ\calU$ cannot be exactly implemented by a finite coherence resource state.}

The above no-go theorem is NOT a direct consequence of the WAY theorems for unitary gates. This is because the implementation of $\calE=\calN\circ\calU$ is not unique, and thus there are many other ways of realizing $\calE$ other than sequentially implementing $\calU$ and $\calN$.
The above result prohibits any such implementation of $\calE$.

\textit{Quantum error correction: An extension of Eastin-Knill theorem to classical information---}
The inequality \eqref{cost-C} also extend the Eastin-Knill theorem, which is a universal restriction on covariant quantum error corrections, to a restriction on the recovery of specific states.
Let us consider a code channel $\calE_{\mathrm{code}}$ from the ``logical system'' $L$ to the ``physical system'' $P$. We assume that the code $\Code$ is an isometry and covariant with respect to $\{U^{L}_{\theta}\}_{\theta\in\mathbb{R}}$ and $\{U^{P}_{\theta}\}_{\theta\in\mathbb{R}}$, where $U^{L}_{\theta}:=e^{i\theta X_L}$ and  $U^{P}_{\theta}:=e^{i\theta X_P}$. 
The physical system $P$ is assumed to be a composite system of $N$ subsystems $\{P_{i}\}^{N}_{i=1}$, and the operator $X_{P}$ in $U^{P}_{\theta}$ is assumed to be written as $X_{P}=\sum_iX_{P_i}$. 
The noise $\calN$ that occurs after the code channel $\Code$ is assumed to be the erasure noise, and the location of the noise is assumed to be known. 
Under this setup, we define the error of the channel $\Code$ for the noise $\calN$ for a test ensemble $\{p_k,\rho_k\}$ as 
\eq{
\epsilon(\Code,\calN,\{p_k,\rho_k\}):=&\delta 
}
where $\delta$ is the recovery error defined in \eqref{recovery error} for $\mathcal{E}=\calN\circ\Code$.
We remark that $\epsilon(\Code,\calN,\{p_k,\rho_k\})$ is not given by the entanglement fidelity. It is defined as the fidelity error and can describe the recovery error for the given ensemble $\{p_k,\rho_k\}$ on $L$.
Then, \HT{for the covariant isometry maps}, in the Supplementary Materials \ref{SIsubsec_EK}, we derive a universal lower bound for $\epsilon(\Code,\calN,\{p_k,\rho_k\})$ from \eqref{cost-C}:
\eq{
\epsilon(\Code,\calN,\{p_k,\rho_k\})\ge\frac{\calC}{\Delta}
}
From this bound, we can see to what extent the classical information encoded by the given ensemble is hurt.
For example, we show that the following inequality holds for a specific $\{p_k,\psi_k\}$:
\begin{align}
\frac{\Delta_{X_L}}{\Delta_{X_L}+4\sqrt{2}N\max_{i}\Delta_{X_{P_i}}}\le \epsilon(\Code,\calN,\{p_k,\psi_k\}).\label{classicalFaist_main}
\end{align}

\subsubsection{Black hole physics: limitation on classical information recovery from scrambling with symmetry}

Our results also provide helpful insights into how the symmetry of black hole dynamics affects the recovery of information from black holes. 
In particular, we present a rigorous lower bound on how many of the $m$ bits of classical information string cannot be recovered in an information recovery protocol from a black hole with the energy conservation law.

We first review the Hayden-Preskill thought experiment \cite{HP}.
In the thought experiment, one considers the situation in which Alice throws a quantum system $A$ (her ``diary" in the original paper) into a quantum black hole $B$ (\HT{Fig.} \ref{HPmodel_main}). Another person, Bob, tries to recover the diary's contents from the Hawking radiation from the black hole.
Then, we assume the following two basic assumptions.
First, the black hole radiation is described by the random unitary model introduced in Ref. \cite{Page1993}, which is widely accepted based on pieces of evidence from string theory and Ads/CFT correspondence (for details, see a review \cite{BH_review}).
Namely, each system is described as qubits, and the dynamics of the black hole is described as a typical Haar random unitary $U$.
Second, the black hole is old enough, and thus there is a quantum system $R_B$ corresponding to the early Hawking radiation that is maximally entangled with the black hole. To decode Alice's diary contents, Bob can use not only the Hawking radiation $A'$ after Alice throws her diary but also the early radiation $R_B$. 
We refer to the numbers of qubits of $A$, $A'$, and $B$ as $k$, $l$, and $N$, respectively. 

\begin{figure}[tb]
    \begin{tabular}{cc}
      \begin{minipage}[t]{1\hsize}
        \centering
\includegraphics[width=1\textwidth]{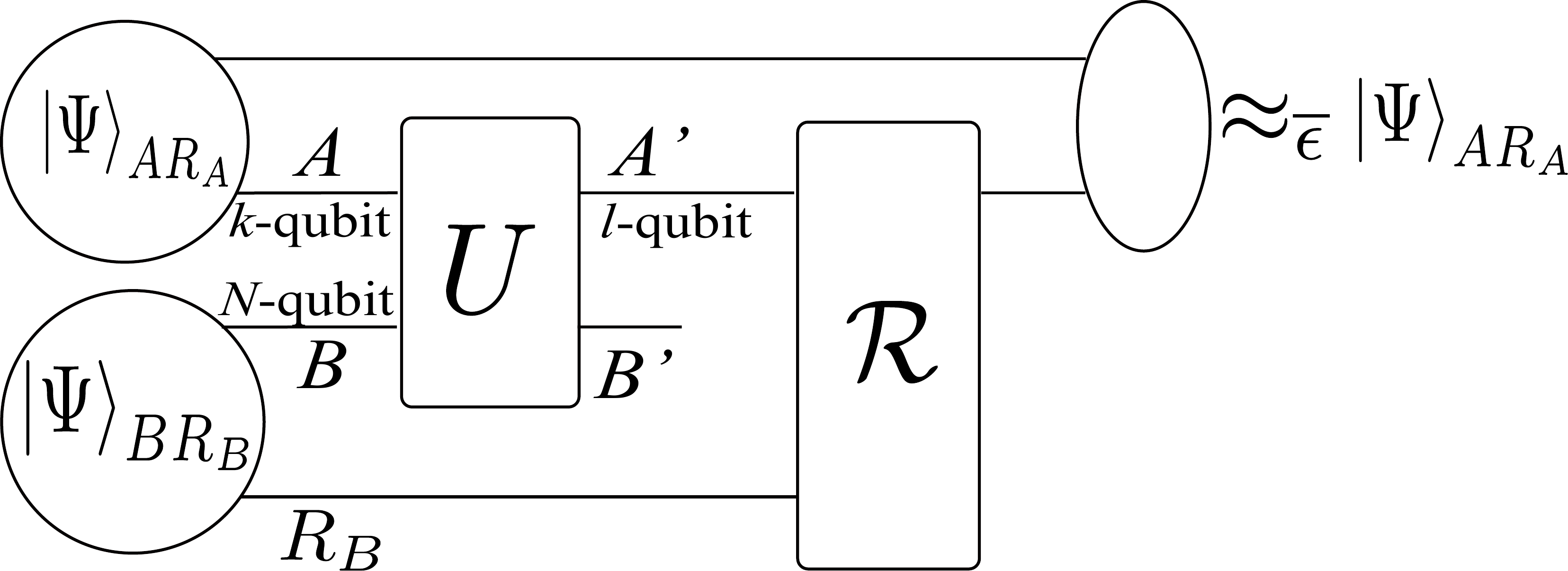}
		\caption{\HT{\textbf{Schematic diagram of the Hayden-Preskill black hole model for the quantum information recovery.}} \HT{In the quantum information recovery, we prepare another quantum reference system $R_A$ and try to recover the initial state $AR_A$ as close as possible.}}
		\label{HPmodel_main}
      \end{minipage} \\ \\
      \begin{minipage}[t]{1\hsize}
     \centering
        \includegraphics[width=1\textwidth]{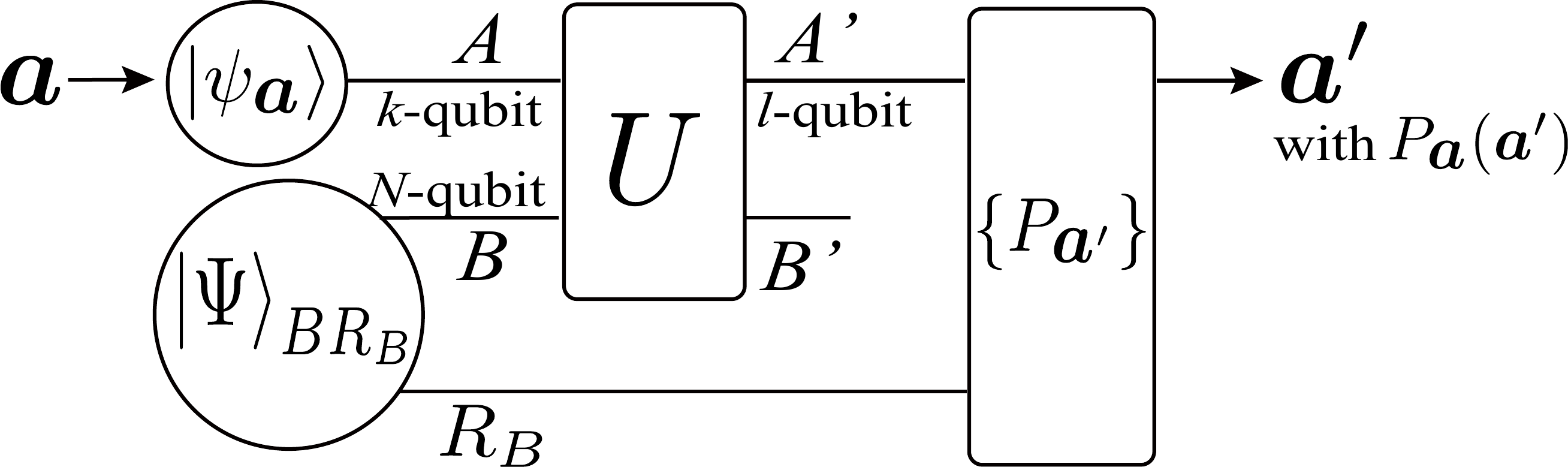}
       		\caption{\textbf{Schematic diagram of the classical information recovery in the Hayden-Preskill black hole model.} We encode classical $m$-bit string $\HT{\vect{a}}:=(a_1,...,a_m)$ into a quantum state $\ket{\psi_{\HT{\vect{a}}}}:=\otimes^{m}_{j=1}\ket{\psi^{(A_j)}_{a_j}}$ on $k=m\times n$ qubit system, and throw it into a black hole. The black hole scrambles the information with the dynamics $U$ which satisfies $U^\dagger (X_{A'}+X_{B'})U=X_A+X_B$. We try to recover $\HT{\vect{a}}$ with using arbitrary recovery protocol.}
 \label{HPmodel_classical_main}
      \end{minipage}
    \end{tabular}
  \end{figure}

Under the above settings, Hayden and Preskill considered how long Bob should wait to see the contents of Alice's diary. For the analysis, they considered the entanglement-fidelity based recovery error $\me$ defined as $\me:=\min_{\calR_{A'\rightarrow A}}D_F(\calR_{A'\rightarrow A}\circ\calE\otimes\mathrm{id}_{R}(\Psi),\Psi)$, where $\Psi$ is the maximally entangled state between $A$ and an external reference system $R_A$.
They then derived the following inequality:
\eq{
\me \le 2^{-(l-k)}.
}  
The implication of this inequality was surprising: Bob hardly has to wait to get almost complete contents when the number of qubits in Hawking radiation $A'$ was just a little more than the number of qubits in $A$.

The above result is derived via a rigorous argument once the setup is accepted. However, the above setup does not take conservation laws into account. 
Recently, the effect of conservation laws on this problem has been actively studied, and it has been shown that the energy conservation delays the speed of escape of information from a black hole \cite{Yoshida-soft,JLiu,Nakata,TS}. 

\begin{figure}[tb]
		\centering
		\includegraphics[width=.45\textwidth]{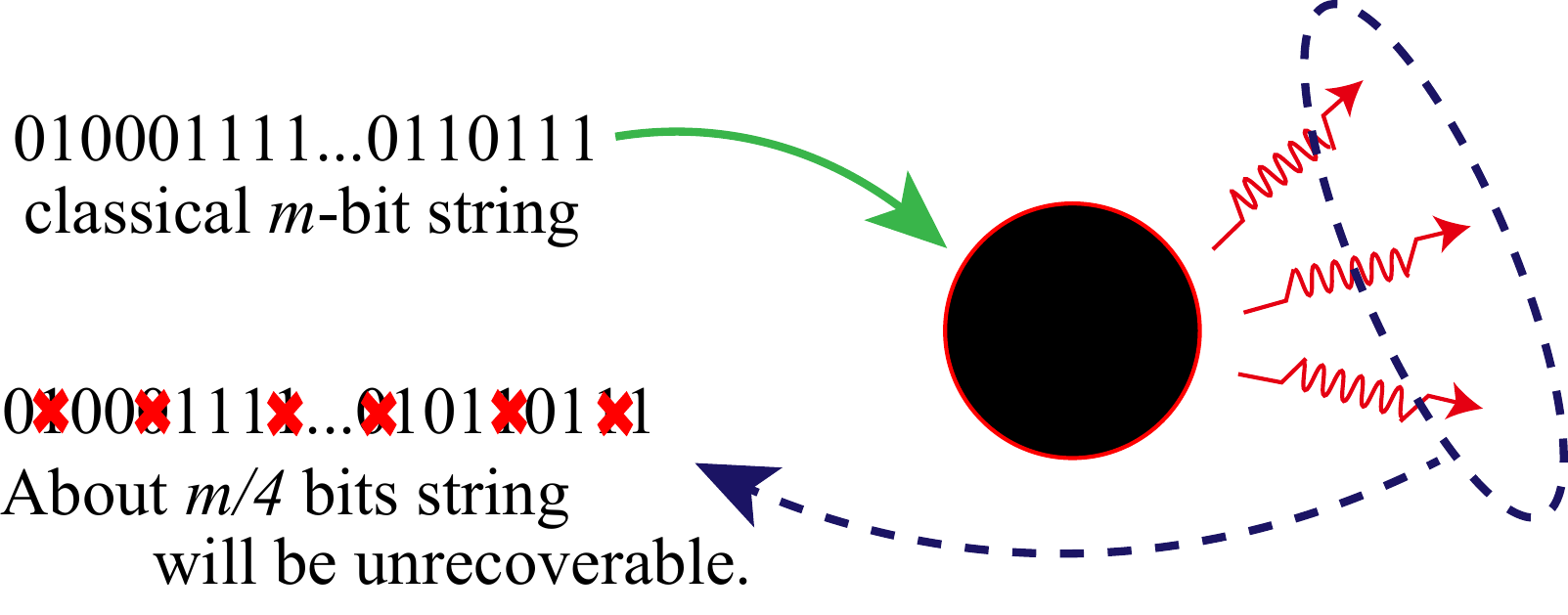}
		\caption{\HT{\textbf{The message of \eqref{BH_Hamming_main}.}} When Alice throws $m$-bit classical string within a proper quantum encoding, Bob cannot recover about $m/4$ bits until the black hole has almost evaporated.}
		\label{message_of_classical}
	\end{figure}

Our result \eqref{SIQ-Cini} brings a development to the above issue.
The highlight of our contribution is a universal lower bound on how many classical bits in Alice's diary will be unreadable under conservation laws (For details of background, see the Supplementary Materials \ref{SIsecBH}).
As a setup, we introduce a classical $m$-bit string $\HT{\vect{a}}:=(a_1,...,a_m)$. Here, each $a_j$ takes values 0 or 1 (\HT{Fig.} \ref{HPmodel_classical_main}).
To encode the classical string $\HT{\vect{a}}$, we prepare the diary $A$ as a composite system of $m$ subsystems $A=A_1...A_m$, where each $A_j$ consists of $n$ qubits. Namely, $k=mn$ holds.
We also prepare two pure states $\ket{\psi^{(A_j)}_{a_j}}$ ($a_j=0,1$) on each subsystem $A_j$ which are orthogonal to each other.
Using the pure states, we encode the string $\HT{\vect{a}}$ into a pure state $\ket{\psi_{\HT{\vect{a}}}}:=\otimes^{m}_{j=1}\ket{\psi^{(A_j)}_{a_j}}$ on $A$.
After the preparation, we throw the pure state $\ket{\psi_{\HT{\vect{a}}}}$ into the black hole $B$. In other words, we perform the energy-preserving Haar random unitary $U$ on $AB$ (Remark: our results also hold even when $U$ is not a Haar random unitary. See the Supplementary Materials \ref{SIsecBH}).
After the unitary dynamics $U$, we try to recover the classical information $\HT{\vect{a}}$. We perform a general measurement $\calM$ on $A'R_B$, and obtain a classical $m$-bit string $\HT{\vect{a}}':=(a'_1,...,a'_m)$ with probability $p'_{\HT{\vect{a}}}(\HT{\vect{a}}')$. We define the  recovery error $\delta_H$ by averaging the Hamming distance between $\HT{\vect{a}}$ and $\HT{\vect{a}}'$ for all possible input $\HT{\vect{a}}$ as follows:
\eq{
\delta_H:=\min_{\calM}\left(\sum_{\HT{\vect{a}},\HT{\vect{a}}'}\frac{p'_{\HT{\vect{a}}}(\HT{\vect{a}}')}{2^m}h(\HT{\vect{a}},\HT{\vect{a}}')\right).
}
Here $h(\HT{\vect{a}},\HT{\vect{a}}')$ is the Hamming distance, which represents how many bits in  $\HT{\vect{a}}$ differ from $\HT{\vect{a}}'$.

Under a proper encoding $\HT{\vect{a}}\rightarrow\{\ket{\psi^{A}_{\HT{\vect{a}}}}\}$, $\delta_H$ becomes proportional to $m$.
For example, when each qubit in $A$ has the same Hamiltonian, we can construct such an encoding by taking $n:=a\sqrt{N}$, where $a$ is an arbitrary constant satisfying $a\ge2$ (see the Materials and Methods). For the encoding,  
when $N\ge 10^3$ and $k\le N$ holds, \eqref{SIQ-Cini} provides the following inequality (\HT{Fig.} \ref{message_of_classical}): 
\eq{
\delta_H\ge \frac{m}{4}\times\frac{1}{\left(1+\frac{3}{a\gamma}\right)^2},\label{BH_Hamming_main}
}
where $\gamma:=1-\frac{l}{N+k}$ is the ratio between the number of qubits in the remaining black hole $B'$ and the total number of qubits $A'B'$.
We remark that \eqref{BH_Hamming_main} holds for an arbitrary decoding method $\calM$.

The inequality \eqref{BH_Hamming_main} is a lower bound on how many characters in Alice's diary will be lost, since it bounds the minimum average number of bit-flip errors in $m$-bit string.
It predicts that when $(1+3/a\gamma)^2$ is not so large, a non-negligible part of the classical bits cannot be read by Bob.

Let us see the above implication of \eqref{BH_Hamming_main} in a concrete example satisfying $k\ll N$.
We consider a black hole with the same size as the Bekenstein-Hawking entropy of Sagittarius A (the BH at the center of the Milky Way).
Since $N$ corresponds to the Bekenstein-Hawking entropy of the black hole \cite{BH_review}, $N=10^{85}$ holds in this case.
Then, $a\sqrt{N}$ can be much smaller than $N$, even if $a$ and $m$ are much larger than $1$.
Let us set $a=10^5$ and $m=10^7$ ($m=10^7$ corresponds to the case that Alice hides 1 megabyte classical information in her dirary). Then, $k=ma\sqrt{N}=10^{54.5}$ and the inequality $k\ll N$ still holds.
In this case, approximately $m/4$ bits in the classical data are unreadable until 99 percent of the black hole evaporates.

\HT{\section{Discussion}}
We have reported the existence of a universal trade-off structure between symmetry, irreversibility, and quantumness. Each of these three concepts plays a central role in various areas of physics. Therefore, the trade-off structures we have found have rich applications. In the main text, we specifically address applications to nonequilibrium thermodynamics, quantum information processing, and black hole physics. However, the range of applications of our results goes beyond what is discussed in the current manuscript. We conclude with a brief discussion of other fields where our results have potential applications.

\textbf{Fault-tolerant quantum computation:}
Our results put strong restrictions on quantum information processing. These restrictions apply not only to fundamental problems but also to practical ones considered for building fault-tolerant quantum computers. A good example is the bosonic codes: we expect to find a universal tradeoff between their decoding performance and the number of photons used.

\textbf{Quantum chaos:} 
Quantum chaos is closely related to fundamental problems such as the thermalization of isolated systems and arrows of time. 
Recently, these problems have actively been studied via information scrambling. Our results have already clarified how symmetry affects the information recovery from scrambling. Considering that our results actually restrict the OTOC \cite{ET2023}, it is natural to expect that our results have further applications to other measures of scrambling, e.g. the tripartite mutual information content.

\textbf{Condensed matter physics:}
Since symmetry is a fundamental concept in condensed matter physics, our results are expected to have useful applications in this field. In particular, we expect an application to the phenomenon called a measurement-induced phase transition, where the entanglement growth under probabilistic measurements is studied. Our trade-off relation is expected to evaluate the symmetry effect on the entanglement growth.

\HT{
\textbf{Quantum thermodynamics:}
For quantum thermodynamics, we have already given a trade-off between entropy production and coherence. However, our results are expected to give a variety of more powerful results. An interesting direction is evaluating the coherence cost of the Gibbs preserving maps. We leave this direction as a future work.
}

\HT{
Besides applications to various fields, we expect that our trade-off relations can be further extended.  
The most direct extension would be to infinite-dimensional systems, for which some preliminary results, such as an extension of the qualitative WAY theorem~\cite{Kuramochi-Tajima}, have just started being undiscovered.  
We leave the thorough investigation along this line as future work.
}

\section*{Acknowledgments} 
We are grateful to Keiji Saito, who we think almost as a co-author, for the fruitful discussion and many helpful comments. We also thank Haruki Emori, Satoshi Nakajima and Tomohiro Shitara for helpful comments.
The present work was supported by JSPS Grants-in-Aid for Scientific Research No. JP19K14610 \HT{(H.T.)}, No. JP22H05250 \HT{(H.T.)}, No. JP22K13977 \HT{(Y.K.)} and JST PRESTO No. JPMJPR2014 \HT{(H.T.)}, JST MOONSHOT No. JPMJMS2061 \HT{(H.T.)}, the Lee Kuan Yew Postdoctoral Fellowship at Nanyang Technological University Singapore \HT{(R.T.)}, JSPS KAKENHI Grant Number JP23K19028 \HT{(R.T.)}, and JST, CREST Grant Number JPMJCR23I3, Japan \HT{(R.T.)}. 

\appendix

\section{Resource theory of asymmetry}\label{App_RTA}

For the readers' convenience, we briefly introduce the minimal tips for the resource theory of asymmetry and the QFI.
The resource theory of asymmetry is a variant of resource theory \cite{Marvian_thesis,skew_resource,Takagi_skew,Marvian_distillation,YT,YT2,Kudo_Tajima,WAY_RToA1,WAY_RToA2,TN,TSS,TSS2,TS,e-EKZhou,e-EKYang,Liu1} that handles symmetries and conservation laws.
In the main text, we consider the case where the symmetry is described by the real number $\mathbb{R}$ or the unitary group $U(1)$.
It is the simplest case where the dynamics have a single conserved quantity.

Like other resource theories, the resource theory of asymmetry has free states and free operations, called symmetric states and covariant operations.
First, we define symmetric states.
Let $\rho$ and $X_S$ be a state and a Hermitian operator of the conserved quantity on $S$.
When $\rho$ satisfies the following relation, we call $\rho$ a symmetric state with respect to  $\{e^{iX_St}\}$.
\eq{
e^{iX_St}\rho e^{-iX_St}=\rho,\enskip\forall t.
}
By definition, $\rho$ is symmetric with respect to $\{e^{iX_St}\}$ if and only if $[\rho,X_S]=0$. In other words, a symmetric state is a quantum state with no coherence with respect to the eigenbasis of the conserved quantity.

Next, we define covariant operations.
Let $\calE_{S\rightarrow S'}$ be a CPTP map from $S$ to $S'$, and let $X_S$ and $X_{S'}$ be Hermitian operators on $S$ and $S'$.
When $\calE_{S\rightarrow S'}$ satisfies the following relation, we call $\calE_{S\rightarrow S'}$ a covariant operation with respect to $\{e^{iX_St}\}$ and $\{e^{iX_{S'}t}\}$:
\begin{align}
\calE_{S\rightarrow S'}(e^{iX_St}...e^{-iX_St})=e^{iX_{S'}t}\calE_{S\rightarrow S'}(...)e^{-iX_{S'}t},\enskip\forall t.\label{covariantcond}
\end{align}
An important property of covariant operations is that we can realize an arbitrary covariant operation by using a proper unitary operation satisfying a conservation law and a quantum state which commutes with the conserved quantity.
To be concrete, let $\calE_{S\rightarrow S'}$ be a covariant operation with respect to $\{e^{iX_St}\}$ and $\{e^{iX_{S'}t}\}$.
Then, we can take quantum systems $E$ and $E'$ satisfying $SE=S'E'$,  Hermitian operators $X_{E}$ and $X_{E'}$ on $E$ and $E'$, a unitary operation $U$ on $SE$ satisfying $U(X_S+X_{E})U^{\dagger}=X_{S'}+X_{E'}$, and a symmetric state $\mu_E$ on $E$ satisfying $[\mu_{E},X_E]=0$, and realize $\calE_{S\rightarrow S'}$ as follows \cite{Marvian_distillation}:
\begin{align}
\calE_{S\rightarrow S'}(...)=\Tr_{E'}[U(...\otimes\mu_E) U^{\dagger}].
\end{align}

\HT{The SLD-QFI for the family $\{e^{-iXt}\rho e^{iXt}\}_{t\in\mathbb{R}}$, described as $\calF_{\rho_S}(X_S)$, is known as a standard resource measure in the resource theory of asymmetry~\cite{Takagi_skew, Marvian_distillation,YT}.
In fact, in RTA, the QFI plays the same role as the entanglement entropy in the entanglement theory: it determines the state conversion ratio in the i.i.d. (independent and identically distributed) state conversion \cite{Marvian2022operational}, and its information-spectrum extension determines the state conversion ratio in the non-i.i.d. state conversion \cite{YT}.

The SLD-QFI} is also a quantifier of quantum fluctuation, since it is related to the variance $V_{\rho_S}(X_S):=\ex{X^2_S}_{\rho_S}-\ex{X_S}^2_{\rho_S}$ as follows \cite{min_V_Petz,min_V_Yu,Marvian_distillation}:
\begin{align}
\calF_{\rho_S}(X_S)&=4\min_{\{q_i,\phi_i\}}\sum_{i}q_iV_{\phi_i}(X_S)\label{F=4V}.
\end{align}
where $\{q_i,\phi_i\}$ runs over the ensembles satisfying $\rho=\sum_{i}q_i\phi_i$ and each $\phi_i$ is pure. We remark that when $\rho$ is pure, $\calF_{\rho}(X)=4V_{\rho}(X)$ holds.
The equality \eqref{F=4V} shows that $\calF_{\rho}(X)$ is the minimum average of the fluctuation caused by quantum superposition. Therefore, we can interpret $\calF_{\rho_S}(X_S)$ as a quantum fluctuation of $X_S$.

\section{properties of $\delta$}\label{prop_delta}
\subsection{Relation between $\delta$ and other irreversibility measures}
In this subsection, we show that our irreversibility measure $\delta$ lower bounds other well-used irreversibility measures.
This means that our inequalities \eqref{SIQ-Cini} and \eqref{SIQ-Gini} provide valid lower bounds for various other irreversibility measures. 
In the following, we list the irreversibility measures bounded by $\delta$.

\textit{Irreversibility measures defined by entanglement fidelity:} In quantum information theory, especially in the areas of quantum error correction and gate implementation, entanglement fidelity-based recovery errors are often used.
Three of the most commonly used recovery errors for a CPTP map $\calE$ from $A$ to $A'$ are as follows:
\eq{
\ew&:=\min_{\calR_{A'\rightarrow A}} \max_{\rho \mbox{ on } AR}D_F(\calR_{A'\rightarrow A}\circ\calE\otimes\mathrm{id}_{R}(\rho),\rho),\label{S78-main}\\
\me&:=\min_{\calR_{A'\rightarrow A}}D_F(\calR_{A'\rightarrow A}\circ\calE\otimes\mathrm{id}_{R}(\Psi),\Psi)\label{S79-main},\\
\epsilon(\psi)&:=\min_{\calR_{A'\rightarrow A}}D_F(\calR_{A'\rightarrow A}\circ\calE\otimes\mathrm{id}_{R}(\psi),\psi),\label{S80-main}
}
where $R$ is a reference system whose Hilbert space has the same dimension as that of $A$, $\Psi$ is the maximally entangled state on $AR$, and $\psi$ is an arbitrary pure state on $AR$. Note that $\me$ is a special case of $\epsilon(\psi)$.
As we show in the Supplementary Materials \ref{SIsubsecRelation}, the irreversibility measure $\delta$ can provide lower bounds for these three errors.

First, for an arbitrary test ensemble $\{p_k,\rho_k\}$, we obtain
\eq{
\delta\le\ew.\label{BB1_M}
}
Second, for an arbitrary test ensemble $\{p_k,\rho_k\}$ satisfying $\sum_kp_k\rho_k=I_A/d_A$ ($d_A$ is the dimension of $A$), we obtain
\eq{
\delta\le\me.\label{BB2_M}
} 
Third, for an arbitary pure state $\psi$ on $AR$ and for an arbitrary test ensemble $\{p_k,\rho_k\}$ satisfying $\sum_kp_k\rho_k=\Tr_{R}[\psi]$, we obtain 
\eq{
\delta\le\epsilon(\psi).\label{BB3_M}
} 

\textit{Petz recovery map:} Our irreversibility measure $\delta$ also bounds the recovery error of the Petz recovery map. 
For an arbitrary quantum channel $\calN$ and a ``reference state'' $\sigma$, the Petz recovery map is defined as follows \cite{Petz_map_definition}:
\eq{
\calR_{\calN,\sigma}(...):=\sqrt{\sigma}\calN^\dagger(\sqrt{\calN(\sigma)}^{-1}(...)\sqrt{\calN(\sigma)}^{-1})\sqrt{\sigma}.
}
The Petz recovery map introduced above has two important properties.
First, it recovers the reference state perfectly, i.e., 
$\sigma=\calR_{\calN,\sigma}(\sigma)$.
Second, the error of the Petz recovery map restricts the generalized entropy production $\Sigma_{\calN,\rho,\sigma}$.
Let us define the error of the Petz recovery map as $\delta_{P}:=D_F(\rho,\calR_{\sigma,\calN}\circ\calN(\rho))$. Then, the following inequality holds \cite{wilde2015,junge2018}:
\eq{
\Sigma\ge-\log (1-\delta^2_{P})\ge\delta^2_{P}.\label{Wgreat}
}
Due to these properties, the Petz recovery map is widely used in various fields of quantum information science \cite{Petz_to_QI}, statistical mechanics \cite{Petz_to_SM}, and black hole physics \cite{Petz_to_BH}.

Now let us apply our theorem to the Petz recovery map.
Due to $\sigma=\calR_{\calN,\sigma}(\sigma)$, when we choose the channel $\calE$ and the test ensemble $\{p_k,\rho_k\}$ as $\calN$ and $\{1/2,\rho_k\}_{k=1,2}$ where $\rho_1:=\rho$ and $\rho_2:=\sigma$, the irreversibility $\delta$ gives the following lower bound of the recovery error $\delta_{P}$ of the Petz map:
\begin{align}
\delta\le\min_{\calR:\delta_2=0}\sqrt{\frac{\delta^2_{1}+\delta^2_{2}}{2}}
\le\frac{\delta_{P}}{\sqrt{2}}.\label{P_d}
\end{align}
Therefore, \eqref{SIQ-Gini} limits the error $\delta_{P}$ of the Petz recovery map directly.

We remark that the irreversibility $\delta$ furthermore recovers indicators of various physical quantities other than irreversibility. For example,  in Ref. \cite{ET2023}, it is shown that $\delta$ recovers (almost) arbitrary definitions of error and disturbance of quantum measurements and OTOC.

\textit{Entropy production in thermodynamic processes, and its generalization:} 
Combining \eqref{Wgreat} and \eqref{P_d}, we obtain the following inequality for arbitrary $\rho$, $\sigma$ and $\calN$:
\eq{
2\delta^2\le\Sigma_{\calN,\rho,\sigma}.\label{delta-G-entropy}
}
When a quantum channel $\calN$ maps the Gibbs state with the temperature $\beta$ to the Gibbs state of the same temperature $\beta$, the generalized entropy production $\Sigma_{\calN,\rho,\sigma}$ becomes the entropy production $\Sigma_\beta$ defined in \eqref{def_entropy}.
Therefore, when the quantum channel $\calN$ is Gibbs preserving (i.e., when the entropy producition $\Sigma_\beta$ is well defined), we obtain
\eq{
2\delta^2\le\Sigma_\beta(\rho).\label{delta-entropy}
}
Substituting \eqref{delta-G-entropy} for $\delta$ in \eqref{SIQ-Gini}, and noting that when the test ensemble is in the form of $\{1/2,\rho_k\}_{k=1,2}$, we can make \eqref{SIQ-Gini} tighter by substituting $\sqrt{2}\calC$ for $\calC$ and obtain
\eq{
\sqrt{\Sigma_{\calN,\rho,\sigma}}\ge\frac{4\calC^2}{(\sqrt{\calF_{\rho_C}(X_C)}+\Delta)^2}.\label{GEPvscoherence}
}
When the quantum channel $\calN$ is Gibbs preserving, we obtain \eqref{motomoto}.

\HT{On might consider that \eqref{GEPvscoherence} may not give a tight estimate when $\Sigma_{\calN,\rho,\sigma}$ is large because the right-hand side of \eqref{GEPvscoherence} is less than 2 because of $\delta\le1$, while the generalized entropy production can be much larger than $2$.
We can remedy this by using the middle inequality $\Sigma\ge-\log (1-\delta^2_P)$ of \eqref{Wgreat} instead of $\Sigma\ge-\log (1-\delta^2_P)$ in the above derivation. Then, we obtain the following inequality:
\eq{
\Sigma_{\calN,\rho,\sigma}\ge-\log\left(1-\frac{4\calC^4}{(\sqrt{\calF}+\Delta)^4}\right).
}
This inequality particularly shows that when $\frac{4\calC^4}{(\sqrt{\calF}+\Delta)^4}=1$ (i.e. when the right-hand side of \eqref{GEPvscoherence} is equal to 2), the entropy production diverges.
}

\HT{
\subsection{A necessary condition for meaningful irreversibility measures: why $\delta$ requires a test ensemble with more than two states}
The irreversibility measure $\delta$ is a function of the local dynamics $\calE$ and the test ensemble $\{\rho_k,p_k\}$, and thus it requires at least two test states in the test ensemble to take non-zero value. To see the reason of this feature, note that the irreversibility $\delta$ of a given CPTP map cannot be defined for a single state, but only for a set of states (=test ensemble). This is because it is always possible to provide perfect recovery for any CPTP map when a single state $\rho$ is known to be the initial state (a map that simply discards the final state of the CPTP map and supplies $\rho$ serves as a perfect recovery map for the case). Therefore, a meaningful irreversibility measure can be defined only as a function of an ensemble that has multiple candidate states, unless using another reference system. For example, the entropy production $\Sigma_\beta(\rho)$ needs two states, the first one is $\rho$ and the second one is the Gibbs state $\rho_{\beta|H}$.
Reflecting this fact, our irreversibility measure $\delta$ always returns 0 for ensembles with only a single state.

We remark that the quantity $\calC$ also always returns 0 for ensembles with only a single state. This fact is a signal that the definition of $\calC$ is consistent with the definition of $\delta$, and also with another important property of $\calC$. See the subsection \ref{secCF_bound} in this Materials and Methods.
}

\HT{
\subsection{Existence of a recovery map the attains the minimum}
In the definition of the recovery error \eqref{recovery error} we implicitly assumed the existence of a CPTP map $\calR$ that attains the minimum on the right-hand side.
The existence of such a recovery map is proved as follows.

We first note that for given finite-dimensional quantum systems $A$ and $A'$, the set $\mathrm{CPTP}(A\to A')$ of CPTP maps from $A$ to $A'$ is a compact set.
Thus to establish the existence of $\calR$ attaining the minimum, it is sufficient to show that objective function $\sqrt{\sum_kp_k\delta_k^2}$ is a continuous function of $\calR$.
This reduces to showing the continuity of the function
\begin{equation}
\calR \mapsto D_F(\rho,\calR(\sigma))
\label{eq:RmapstoDF}    
\end{equation}
for fixed states $\rho$ and $\sigma$.
Since $D_F$ is a metric defined on the set of states and gives an the same topology as that induces from the trace distance (cf.\ \cite{761271}, Theorem~1), the continuity of \eqref{eq:RmapstoDF} follows from the continuity of 
\[\calR \mapsto \calR(\sigma).\]
Note that in general a metric $d(\cdot,\cdot)$ is continuous with respect to the topology induced by $d$ itself.
Thus the proof is done.
}

\section{Properties of $\calC$ and $\Delta$}\label{App_Canddelta}

\subsection{Shift invariance of $\calC$, $\Delta$ and $\Delta_\alpha$}
We also remark that $\calC$, $\Delta$ and $\Delta_\alpha$ are invariant with respect to the shift of $X_A$ and $X_{A'}$.
To be concrete, when we define $\tilde{X}_{A}:=X_{A}+aI_A$,  and $\tilde{X}_{A'}:=X_{A'}+bI_{A'}$ where $a$ and $b$ are arbitrary real numbers, and when we also define $\tilde{\calC}$, $\tilde{\Delta}$ and $\tilde{\Delta}_\alpha$ as $\calC$, $\Delta$ and $\Delta_\alpha$ for $\tX_A$ and $\tX_{A'}$, the following relations hold (for details, see the Supplementary Materials \ref{SIsubsec_prop_C}):
\eq{
\tC=\calC,\enskip\tD=\Delta,\enskip\tD_\alpha=\Delta_\alpha.\label{i-s-main}
}

\subsection{Conditions for $\calC>0$}
In this subsection, we give several conditions when $\calC>0$ holds. 
\HT{We first introduce an important condition:
\eq{
Y\not\propto I_A\Leftrightarrow \exists\{p_k,\rho_k\} \mbox{ s.t. } \calC>0.\label{nonshift-nonzero}
}
The proof is given in the Supplementary Materials \ref{SIsubsec_prop_C}. To see the meaning of \eqref{nonshift-nonzero}, let us recall that $Y=X_A-\calE^\dagger(X_{A'})$. In other words, $Y$ is the operator of the change of the local conserved charge by the local dynamics $\calE$. It is an extension of the ``work operator" in stochastic thermodynamics.
We also remark that $Y\propto I_A$ if and only if the change of the conserved charge is just a shift of the origin of the charge.
Therefore, unless the change of the local conserved quantity caused by $\calE$ is just a shift of its origin, $\calC>0$ holds at least one test ensemble. The relation \eqref{nonshift-nonzero} also guarantees its converse. In this sense, $\calC$ works as an indicator of the change of the local conserved charge.

Next, we give a necessary and sufficient condition for $\calC>0$ for a given test ensemble.
As a starting point, we recall the definition $\calC:=\sqrt{\sum_{k\ne k'}p_kp_{k'}\Tr[(\rho_{k}-\rho_{k'})_{+}Y(\rho_{k}-\rho_{k'})_{-}Y]}$, and show a necessary and sufficient condition of $\calC_{k,k'}:=\Tr[(\rho_{k}-\rho_{k'})_{+}Y(\rho_{k}-\rho_{k'})_{-}Y]>0$.
The term $\calC_{k,k'}$ is always non-negative and satisfies
\begin{align}
\calC_{k,k'}>0\Leftrightarrow[Y,\Pi^{(k,k')}_{\pm}]\ne0,
\end{align}
where $\Pi^{(k,k')}_{\pm}$ is the projection to the support of $(\rho_k-\rho_{k'})_{\pm}$.
Since the measurement $\{\Pi^{(k,k')}_+,\Pi^{(k,k')}_{-}\}$ (more precisely, $\{\Pi^{(k,k')}_+,1-\Pi^{(k,k')}_{+}\}$)} is the optimal measurement to distinguish $\rho_k$ and $\rho_{k'}$, we can interpret $\calC_{k,k'}$ as the sum of the non-diagonal elements on the optimal basis to distinguish $\rho_k$ and $\rho_{k'}$.
In fact, $\calC_{k,k'}=\sum_{l,l'}q_{l|+}q_{l'|-}|\bra{\phi_{l|+}}Y\ket{\phi_{l'|-}}|^2$ holds where $\{q_{l|\pm}\}$ and $\{\ket{\phi_{l|\pm}}\}$ are the eigenvalues and eigenbasis of $(\rho_k-\rho_{k'})_{\pm}$. Therefore, $\calC_{k,k'}$ is positive if and only if $Y$ has at least one non-diagonal element between an eigenvector of $(\rho_k-\rho_{k'})_{+}$ and another eigenvector of $(\rho_k-\rho_{k'})_{-}$.

\HT{From the above necessary and sufficient condition for $\calC_{k,k'}>0$, we can easily obtain the desired condition: 
\eq{
\calC>0  \mbox{ for a } \{p_k,\rho_k\} 
 \Leftrightarrow  [Y,\Pi^{(k,k')}_{\pm}]\ne0,\enskip \forall k\ne k'. \label{nonshift-nonzero_k_kprime}
}
}


\HT{
\subsection{Upper and lower bounds of $\calC$ in terms QFI of the work operator $Y$}\label{secCF_bound}
In this subsection, we give upper and lower bounds of $\calC$ with the QFI.
When the test states $\{\rho_k\}$ are orthogonal to each other, $\calC$ satisfies
\eq{
\left(\min_kp_k\lambda^{\min}_{>0}(\rho_k)\right)\frac{\mathrm{C}_{\calF}}{4}\le\calC^2\le\frac{\mathrm{C}_{\calF}}{8}.\label{convex_rep_C}
}
Here, $\mathrm{C}_{\calF}$ denotes the convexity of the QFI of the operator $Y$: 
\eq{
\mathrm{C}_{\calF}:=\sum_kp_k\calF_{\rho_k}(Y)-\calF_{\sum_kp_k\rho_k}(Y)
}
and $\lambda^{\min}_{>0}(\xi)$ is the minimum positive eigenvalue of state $\xi$.
The proof of \eqref{convex_rep_C} is given in the Supplementary Materials \ref{SIsubsec_prop_C}.

It is noteworthy that the convexity $\mathrm{C}_{\calF}$ of the QFI of the operator $Y$ gives upper and lower bounds of $\calC$. From the upper and lower bounds, we can see that $\calC$ reflects the gain in quantum fluctuations of the operator $Y$ when we know $k$ compared to when we do not know $k$. 
This property also well explains why $\calC$ always returns 0 for ensembles with only a single state, since the gain in the quantum fluctuations becomes 0 for such ensembles.

We stress that the constant $\min_kp_k\lambda^{\min}_{>0}(\rho_k)$ is always strictly positive by definition, and is not small in many cases.
For example, when $\xi$ is a pure state, $\lambda^{\min}_{>0}(\xi)=1$. Therefore, for example when the test ensemble is $\{(1/2,1/2),(\ket{\psi},\ket{\phi})\}$ for two orthogonal states $\ket{\psi}$ and $\ket{\phi}$, $\min_kp_k\lambda^{\min}_{>0}(\rho_k)=1/2$. 
Therefore, in such cases (the applications to the WAY theorem, unitary WAY theorem, the Eastin-Knill theorem and the black holes are included in this type of test ensembles), \eqref{convex_rep_C} reduces to
\eq{
\calC^2=\frac{\mathrm{C}_{\calF}}{8}.
}
Therefore, as long as we use the test ensemble $\{(1/2,1/2),(\ket{\psi},\ket{\phi})\}$, the following relation holds:
\eq{
\frac{\mathrm{C}_\calF}{4(\sqrt{\calF}+\Delta)^2}\le\delta^2.
}

}

\subsection{Upper bounds with $\Delta_\alpha$ for $\Delta$}
The quantity $\Delta$ defined in \eqref{def_Delta} has several upper bounds:
\eq{
\Delta&\le\Delta_1:=\Delta_{X_A}+\Delta_{X_{A'}},\label{D_1}\\
\Delta&\le\Delta_2:=\Delta_Y+2\sqrt{\|\calE^\dagger(X^2_{A'})-\calE^\dagger(X_{A'})^2\|_\infty}\label{D_2}\\
\Delta&\le\Delta_3:=\max_{\rho\in\cup_k\mathrm{supp}(\rho_k)}\left[\sqrt{\calF_{\rho}(Y)}+\sqrt{\calF_{\rho\otimes\rho_B}(\tilde{\tilde{Y}})}\right].\label{D_3}
}
Here $\tilde{\tilde{Y}}:=U^\dagger X_{A'}\otimes1_{B'}U-\calE^\dagger(X_{A'})\otimes1_{B'}).$
We show these inequalities in the Supplementary Materials \ref{SIsubsec_boundDelta}.
Due to the above three bounds, we can substitute $\Delta_1$, $\Delta_2$ and $\Delta_3$ for $\Delta$ in \eqref{SIQ-Cini} and \eqref{SIQ-Gini}. When a statement, equation, etc., are valid using either $\Delta_1$, $\Delta_2$ or $\Delta_3$, we use the symbol $\Delta_\alpha$ to denote them collectively.

\section{A proper encoding to hide classical bits in a black hole}
In this section, we give a concrete form of the encoding $\HT{\vect{a}}\rightarrow\ket{\psi_{\HT{\vect{a}}}}$ which satisfies \eqref{BH_Hamming_main}, under the  assumption that each qubit in $A$ has the same Hamiltonian $H:=\ket{1}\bra{1}$. 
Then, the energy eigenvalues of the Hamiltonian $H^{A_j}$ on $A_j$ become integers from $0$ to $n$.
We refer to the eigenvectors of $H^{A_j}$ with the eigenvalues $0$ and $n$ as $\ket{0}_{A_j}$ and $\ket{n}_{A_j}$, respectively, and define $\ket{\phi^{A_j}_{0}}:=(\ket{0}_{A_j}+\ket{n}_{A_j})/\sqrt{2}$ and $\ket{\phi^{A_j}_{1}}:=(\ket{0}_{A_j}-\ket{n}_{A_j})/\sqrt{2}$, respectively.
Let us take $n:=a\sqrt{N}$, where $a$ is an arbitrary constant satisfying $a\ge2$. 
Then, we can define $\ket{\psi_{\HT{\vect{a}}}}$ as $\otimes^{m}_{j=1}\ket{\phi^{A_j}_{a_j}}$.
As shown in the Supplementary Materials \ref{SIsecBH}, this encoding satisfies \eqref{BH_Hamming_main}, When $N\ge 10^3$ and $k\le N$ holds.

\section{Coherence cost of operator conversion}

In this section, we introduce the method we use to derive the main results.
The main results \eqref{SIQ-Cini} and \eqref{SIQ-Gini} are derived from a single lemma that rules the coherence cost of the operator conversion. 
\begin{lemma}\label{L1}
Let us consider two quantum systems \HT{$A$ and $A'$}, and Hermitian operators \HT{$X_A$ and $X_{A'}$} on them.
We also take a projective operator $Q$ on \HT{$A$} and a non-negative operator $0\le P\le I$ on \HT{$A'$}.
Let $\Lambda$ be a CPTP map from \HT{$A$ to $A'$}, and suppose its dual $\Lambda^\dagger$ approximately changes $P$ to $Q$ as follows:
\begin{align}
\ex{\Lambda^\dagger(P)}_{(1-Q)\rho_{\HT{A}}(1-Q)}+\ex{1-\Lambda^\dagger(P)}_{Q\rho_{\HT{A}} Q}\le\epsilon^2.\label{c-op}
\end{align}
Here $\epsilon$ is a real positive number.  
We also introduce another quantum system $\HT{B}$ and a tuple \HT{$(V,\rho_B,X_B,X_{B'})$} of a unitary $V$ on \HT{$AB$}, a state \HT{$\rho_B$ on $B$}, an Hermitian operator \HT{$X_B$ on $B$} and another Hermitian operator \HT{$X_{B'}$ on $B'$, where $B'$} is a quantum system satisfying \HT{$AB=A'B'$}.
We assume that \HT{$(V,\rho_B,X_B,X_{B'})$} is an implementation of $\Lambda$ and satisfies the conservation law of $X$, i.e., 
$\Lambda(...)=\Tr_{\HT{B'}}[V(...\otimes\rho_{\HT{B}})V^{\dagger}]$ and $X_{\HT{A}}+X_{\HT{B}}=V^{\dagger}(X_{\HT{A'}}+X_{\HT{B'}})V$.
Then, the following relation holds:
\begin{align}
\epsilon\ge\frac{|\ex{[Q,Y_{\HT{A}}]}_{\rho_{\HT{A}}}|}{\Delta_{\HT{A,A',\rho_A}}+\sqrt{\calF_{\rho_{\HT{B}}}(X_{B})}},\label{convert}
\end{align}
where \HT{$Y_A:=X_A-\Lambda^\dagger(X_{A'})$} and $\Delta_{\HT{A,A',\rho_A}}$ is a symbol corresponding to $\Delta$, which is defined as  
\eq{
\Delta_{\HT{A,A',\rho_A}}:=\sqrt{\calF_{\HT{\rho_A\otimes\rho_B}}(X_{\HT{A}}\otimes1_{\HT{B}}-V^\dagger X_{\HT{A'}}\otimes1_{\HT{B'}}V)}.
}
\end{lemma}
The condition \eqref{c-op} means that if we perform measurements $\{Q,1-Q\}$ and $\{\Lambda^\dagger(P),1-\Lambda^\dagger(P)\}$ on \HT{$\rho_A$} in succession, the probability of a discrepancy between the results of the first and second measurements is less than $\epsilon$.
In that sense, the number $\epsilon$ describes the error of the conversion from $P$ to $Q$ by $\Lambda^\dagger$ for the initial state \HT{$\rho_A$}. 
Then, Lemma \ref{L1} states that to convert $P$ $\epsilon$-close to $Q$, we need coherence $\calF_{\HT{\rho_B}}(\HT{X_B})$ inversely proportional to $\epsilon^2$.

We can derive the main results \eqref{SIQ-Cini} and \eqref{SIQ-Gini} from Lemma \ref{L1} by choosing proper $P$, $Q$, and $\rho_A$ (see Supplementary \HT{Materials} \ref{SIsec1} for details).
Lemma \ref{L1} is derived from the following improved version of the Kennard-Robertson uncertainty relation \cite{Frowis_unc,TN}.
\eq{
|\ex{[O_1,O_2]}_\rho|\le\sqrt{\calF_{\rho}(O_1)}\sqrt{V_\rho(O_2)}.
\label{KR inequality}
}
In other words, all the main results and applications in this paper are derived from the quantum uncertainty relation.

\bibliography{SIQ}
\bibliographystyle{unsrt}
\bibliographystyle{apsrmp4-2}

\clearpage

\begin{widetext}

\setcounter{tocdepth}{2} 

\begin{center}
{\large \bf Supplemental Material for \protect \\ 
``Universal trade-off structure between symmetry, irreversibility, and quantum coherence in quantum processes''}\\
\vspace*{0.3cm}
Hiroyasu Tajima$^{1,2}$, Ryuji Takagi$^{3}$ and Yui Kuramochi$^{1}$ \\
\vspace*{0.1cm}
$^{1}${\small \em Department of Communication Engineering and Informatics, University of Electro-Communications, 1-5-1 Chofugaoka, Chofu, Tokyo, 182-8585, Japan}
\\
$^{2}${\small \em JST, PRESTO, 4-1-8 Honcho, Kawaguchi, Saitama, 332-0012, Japan}
$^{3}${\small \em Department of Basic Science, The University of Tokyo, 3-8-1 Komaba, Meguro-ku, Tokyo 153-8902, Japan}
\\
$^{4}${\small \em Department of Physics, Kyushu University, 744 Motooka, Nishi-ku, Fukuoka, Japan}
\end{center}
\renewcommand{\appendixname}{}
\renewcommand{\thesection}{S.\arabic{section}}
\setcounter{section}{0}
\setcounter{equation}{0}
\setcounter{lemma}{0}
\setcounter{page}{1}
\renewcommand{\theequation}{S.\arabic{equation}}
\renewcommand{\thesubsection}{\Alph{subsection}}


\section{Coherence cost of the operator conversion}\label{SIsec1}

In this section, we derive Lemma \ref{L1} in the main text that rules the coherence cost of the operator conversion.
We also extend it to the case of violated conservation law.

\subsection{Derivation of Lemma \ref{L1} in the main text}
For readers' convenience, we repeat Lemma \ref{L1} in the main text here:
\begin{lemma}\label{L1S}
Let us consider two quantum systems $A$ and $A'$, and Hermitian operators $X_A$ and $X_{A'}$ on them.
We also take a projective operator $Q$ on $A$ and a non-negative operator $P$ satisfying $0\le P\le 1_{A'}$ on $A'$.
Let $\Lambda$ be a CPTP map from $A$ to $A'$, and let its dual $\Lambda^\dagger$ approximately change $P$ to $Q$ as follows:
\begin{align}
\ex{\Lambda^\dagger(P)}_{(1_A-Q)\rho_A(1_A-Q)}+\ex{1_A-\Lambda^\dagger(P)}_{Q\rho_A Q}\le\epsilon^2.\label{c-opS}
\end{align}
Here $\epsilon$ is a real positive number, and $\Lambda^\dagger$ is the dual of $\Lambda$.  
We also introduce another quantum system $B$ and a tuple $(V,\rho_B,X_B,X_{B'})$ of a unitary $V$ on $AB$, a state $\rho_B$ on $B$, an operator $X_B$ on $B$ and an operator $X_{B'}$ on $B'$, where $B'$ is a quantum system satisfying $AB=A'B'$.
We assume that $(V,\rho_B,X_B,X_{B'})$ is an implementation of $\Lambda$ and satisfies the conservation law of $X$, i.e., 
$\Lambda(...)=\Tr_{B'}[V(...\otimes\rho_B)V^{\dagger}]$ and $X_{A}+X_{B}=V^{\dagger}(X_{A'}+X_{B'})V$.
Then, the following relation holds:
\begin{align}
\epsilon\ge\frac{|\ex{[Q,Y_A]}_{\rho_A}|}{\Delta_{A,A',\rho_A}+\sqrt{\calF_{\rho_{B}}(X_{B})}}.\label{convertS}
\end{align}
where $Y_A:=X_A-\Lambda^\dagger(X_{A'})$ and $\Delta_{A,A',\rho_A}$ is a symbol corresponding to $\Delta$, which is defined as  
\eq{
\Delta_{A,A',\rho_A}:=\sqrt{\calF_{\rho_A\otimes\rho_B}(X_{A}\otimes1_{B}-V^\dagger X_{A'}\otimes1_{B'}V)}.
}

We furthermore give another lower bound of $\epsilon$ as follows:
\begin{align}
\epsilon\ge\frac{|\ex{[Q,Y_A]}_{\rho_A}|}{\sqrt{\Delta_{C:A,A',\rho_A}+\calF_{\rho_{B}}(X_{B})}}.\label{convertS_improved}
\end{align}
where
\eq{
\Delta_{C:A,A',\rho_A}:=\calF_{\rho_A\otimes\rho_B}(X_{A}\otimes1_{B}-V^\dagger X_{B'}\otimes1_{B'}V)+2\calF_{\rho_A\otimes\rho_B}(X_{A}\otimes1_{B}-V^\dagger X_{A'}\otimes1_{B'}V,1_{A}\otimes X_B),
}
where $\calF_{\xi}(O_{1},O_{2})$ is the element of the SLD quantum Fisher matrix that is defined as follows:
\eq{
\calF_{\xi}(O_{1},O_{2})&:=\ex{L_{O_{1}},L_{O_{2}}}^{\text{SLD}}_{\xi}.
}
Here $L_{O_x}$ $(x=1,2)$ are Hermitian operators defined as
\eq{
i[\rho,O_{x}]&=\frac{\{\rho,L_{O_{x}}\}}{2},\enskip(x=1,2),
}
and $\ex{O_{1},O_{2}}^{\text{SLD}}_{\xi}$ is the SLD-inner product which is defined as $\ex{O_{1},O_{2}}^{\text{SLD}}_{\xi}:=\ex{\{O_{1},O_{2}\}}_{\xi}/2$ and $\{O_{1},O_{2}\}:=O_{1}O_{2}+O_{2}O_{1}$.
\end{lemma}

\begin{proof}
We first define the following operator:
\begin{align}
N&:=V^\dagger P\otimes1_{B'}V-Q\otimes1_{B}.
\end{align}
Then, because of the improved Kennard-Robertson inequality \eqref{KR inequality}, we obtain
\begin{align}
|\ex{[N,V^\dagger1_{A'}\otimes X_{B'}V]}_{\rho_{A}\otimes\rho_{B}}|\le\sqrt{\calF_{\rho_{A}\otimes\rho_{B}}(V^\dagger1_{A'}\otimes X_{B'}V)}\sqrt{V_{\rho_{A}\otimes\rho_{B}}(N)}.\label{10}
\end{align}
We evaluate $\sqrt{\calF_{\rho_{A}\otimes\rho_{B}}(V^\dagger1_{A'}\otimes X_{B'}V)}$ as follows:
\eq{
\sqrt{\calF_{\rho_{A}\otimes\rho_{B}}(V^\dagger1_{A'}\otimes X_{B'}V)}
&\stackrel{(a)}{=}\sqrt{\calF_{\rho_A\otimes\rho_B}(-V^\dagger X_{A'}\otimes1_{B'}V+X_A\otimes1_B+1_A\otimes X_B)}\non
&\stackrel{(b)}{\le}\sqrt{\calF_{\rho_A\otimes\rho_B}(X_A\otimes1_B-V^\dagger X_{A'}\otimes1_{B'}V)}+\sqrt{\calF_{\rho_A\otimes\rho_B}(1_A\otimes X_B)}\non
&\stackrel{(c)}{=}\sqrt{\calF_{\rho_A\otimes\rho_B}(X_A\otimes1_B-V^\dagger X_{A'}\otimes1_{B'}V)}+\sqrt{\calF_{\rho_B}(X_B)}\non
&=\Delta_{A,A',\rho_A}+\sqrt{\calF_{\rho_B}(X_B)}\label{S6},
}
where we used the conservation law $X_{A}+X_{B}=V^{\dagger}(X_{A'}+X_{B'})V$ in (a), the relation $\calF_{\rho_A\otimes\rho_B}(X_A+X_B)=\calF_{\rho_A}(X_A)+\calF_{\rho_B}(X_B)$ \cite{Hansen} in (c), and the inequality $\sqrt{\calF_\rho(W+W')}\le\sqrt{\calF_\rho(W)}+\sqrt{\calF_\rho(W')}$ in (b), which is shown as follows:
\eq{
\sqrt{\calF_\rho(W+W')}&=\sqrt{\ex{L_W+L_{W'},L_{W}+L_{W'}}^{SLD}_\rho}\non
&=\sqrt{\ex{L_W,L_{W}}^{SLD}_\rho
+\ex{L_W,L_{W'}}^{SLD}_\rho
+\ex{L_{W'},L_{W}}^{SLD}_\rho
+\ex{L_{W'},L_{W'}}^{SLD}_\rho
}\non
&\le\sqrt{\ex{L_W,L_{W}}^{SLD}_\rho
+2\sqrt{\ex{L_W,L_{W}}^{SLD}_\rho\ex{L_{W'},L_{W'}}^{SLD}_\rho}
+\ex{L_{W'},L_{W'}}^{SLD}_\rho
}\non
&=\sqrt{\ex{L_W,L_W}^{SLD}_\rho}+\sqrt{\ex{L_{W'},L_{W'}}^{SLD}_\rho}\non
&=\sqrt{\calF_\rho(W)}+\sqrt{\calF_\rho(W')},\label{S7}
}
where $\ex{O_1,O_2}^{SLD}_\rho:=\Tr[\rho (O_1O_2+O_2O_1)/2]$ and $L_O$ is defined by $i[\rho,O]=(L_O\rho+\rho L_O)/2$.

We also derive
\begin{align}
V_{\rho_A\otimes\rho_B}(N)&\le\Tr[\rho_{A}\otimes\rho_B N^2]\nonumber\\
&=\Tr[\rho_{A}\otimes\rho_B (V^\dagger P^2\otimes1_{B'}V-V^\dagger P\otimes1_{B'}VQ\otimes1_{B}-Q\otimes1_{B}V^\dagger P\otimes1_{B'}V+Q\otimes1_{B})]\nonumber\\
&\le\Tr[\rho_{A}\otimes\rho_B (V^\dagger P\otimes1_{B'}V-V^\dagger P\otimes1_{B'}VQ\otimes1_{B}-Q\otimes1_{B}V^\dagger P\otimes1_{B'}V+Q\otimes1_{B})]\nonumber\\
&=\Tr[\rho_{A}\otimes\rho_B ((1_A-Q)\otimes1_BV^\dagger P\otimes1_{B'}V(1_A-Q)\otimes1_B+Q\otimes1_{B}(1_A\otimes1_B-V^\dagger P\otimes1_{B'}V)Q\otimes1_{B})]\nonumber\\
&=\ex{\Lambda^\dagger(P)}_{(1_A-Q)\rho_A(1_A-Q)}+\ex{\Lambda^\dagger(1_{A'}-P)}_{Q\rho_A Q}\nonumber\\
&\le\epsilon^2\label{S8}
\end{align}
where in the second line we used that $Q^2=Q$ becasue $Q$ is a projective operator, in the third line we used $P^2\leq P$, and in the fifth line we used that for arbitrary positive semidefinite operators $C$ and $D$,
\begin{align}
\ex{\Lambda^\dagger(C)}_{D\rho_A D} & =\Tr[\Lambda^\dagger(C)D\rho_A D]\\
&=\Tr[C\,\Lambda(D\rho_A D)]\\
&= \Tr[C\,\Tr_B \left\{V(D\rho_A D\otimes\rho_B)V^\dagger\right\}]\\
&=\Tr[C\otimes 1_B V (D\rho_A D\otimes\rho_B)V^\dagger]\\
&=\Tr[(\rho_A\otimes\rho_B) (D\otimes 1_B)(V^\dagger C\otimes 1_B V) (D\otimes 1_B)].
\end{align}
We also transform the left-hand side of \eqref{10} as follows:
\begin{align}
\ex{[N,V^\dagger1_{A'}\otimes X_{B'}V]}_{\rho_{A}\otimes\rho_{B}}&=\ex{[V^\dagger P\otimes1_{B'}V-Q\otimes1_{B},V^\dagger(1_{A'}\otimes X_{B'})V]}_{\rho_A\otimes\rho_B}\nonumber\\
&=
-\ex{[Q\otimes1_{B},V^\dagger(1_{A'}\otimes X_{B'})V]}_{\rho_A\otimes\rho_B}\nonumber\\
&=-\ex{[Q\otimes1_{B},X_A\otimes1_B+1_A\otimes X_B-V^\dagger(X_{A'}\otimes 1_{B'})V]}_{\rho_A\otimes\rho_B}\nonumber\\
&=-\ex{[Q,Y_A]}_{\rho_A}.\label{S9}
\end{align}
where in the third line we used the assumption $X_A+X_B=V^\dagger(X_{A'}+X_{B'})V$ and in the fourth line we used 
\begin{align}
 \ex{QY_A}_{\rho_A} &= \Tr[Q(X_A-\Lambda^\dagger(X_{A'}))\rho_A]\\    
 &=\Tr[Q\otimes 1_B (X_A\otimes 1_B-V^\dagger X_{A'}\otimes 1_B V)\rho_A\otimes \rho_B]
\end{align}
and that $[Q\otimes 1_B, 1_A\otimes X_B]=0$.
Hence, we obtain
\begin{align}
|\ex{[Q,Y_A]}_{\rho_A}|\le\epsilon\times(\Delta_{A,A',\rho_A}+\sqrt{\calF_{\rho_{B}}(X_{B})})\label{finalF}
\end{align}
that we seek.

To obtain \eqref{convertS_improved}, we use the following relation:
\eq{
\sqrt{\calF_{\rho_{A}\otimes\rho_{B}}(V^\dagger1_{A'}\otimes X_{B'}V)}
=\Delta_{C:A,A',\rho_A}+\calF_{\rho_B}(X_B)\label{S6_im}.
}
Therefore, substituting \eqref{S6_im} for \eqref{S6} in the above derivation of \eqref{convertS}, we obtain \eqref{convertS_improved}.
\end{proof}

\subsection{Extension to the case of violated conservation law}\label{SI_violated_operator}
In Lemma \ref{L1}, we assumed that the conservation law $X_A+X_B=V^\dagger (X_{A'}+X_{B'})V$ holds.
We can also treat the case where the conservation law is violated.
Let us define an Hermitian operator $Z$ that describes the degree of violation of the conservation as $Z:=V^\dagger (X_{A'}+X_{B'})V-(X_{A}+X_{B})$.
In this case, inequality \eqref{convertS} in Lemma \ref{L1} is generalized as follows:
\begin{align}
\epsilon\ge\frac{|\ex{[Q,Y_A)]}_{\rho_A}|-\frac{\Delta_{Z}}{2}}{\Delta_{A,A',\rho_A}+\Delta_{Z}+\sqrt{\calF_{\rho_{B}}(X_{B})}},\label{convertS_ex}
\end{align}
where $\Delta_Z$ is the difference between the maximum and minimum eigenvalues of $Z$.

{\bf\noindent Proof of \eqref{convertS_ex}:}
The proof is completely the same as the proof of \eqref{convertS} until \eqref{10}:
\begin{align}
|\ex{[N,V^\dagger1_{A'}\otimes X_{B'}V]}_{\rho_{A}\otimes\rho_{B}}|\le\sqrt{\calF_{\rho_{A}\otimes\rho_{B}}(V^\dagger1_{A'}\otimes X_{B'}V)}\sqrt{V_{\rho_{A}\otimes\rho_{B}}(N)}.\label{10'}
\end{align}
Since we do not use the conservation law in the derivation of \eqref{S8}, we can use it again and obtain
\begin{align}
|\ex{[N,V^\dagger1_{A'}\otimes X_{B'}V]}_{\rho_{A}\otimes\rho_{B}}|\le\epsilon\sqrt{\calF_{\rho_{A}\otimes\rho_{B}}(V^\dagger1_{A'}\otimes X_{B'}V)}.\label{10''}
\end{align}
We evaluate $\sqrt{\calF_{\rho_{A}\otimes\rho_{B}}(V^\dagger1_{A'}\otimes X_{B'}V)}$ in the same manner as \eqref{S6}, but use $X_{A}+X_{B}+Z=V^{\dagger}(X_{A'}+X_{B'})V$ instead of $X_{A}+X_{B}=V^{\dagger}(X_{A'}+X_{B'})V$:
\begin{align}
\sqrt{\calF_{\rho_{A}\otimes\rho_{B}}(V^\dagger1_{A'}\otimes X_{B'}V)}
&=\sqrt{\calF_{\rho_A\otimes\rho_B}(-V^\dagger X_{A'}\otimes1_{B'}V+X_A\otimes1_B+1_A\otimes X_B+Z)}\non
&\le\sqrt{\calF_{\rho_A\otimes\rho_B}(X_A\otimes1_B-V^\dagger X_{A'}\otimes1_{B'}V)}+\sqrt{\calF_{\rho_A\otimes\rho_B}(1_A\otimes X_B)}+\sqrt{\calF_{\rho_A\otimes\rho_B}(Z)}\non
&=\sqrt{\calF_{\rho_A\otimes\rho_B}(X_A\otimes1_B-V^\dagger X_{A'}\otimes1_{B'}V)}+\sqrt{\calF_{\rho_B}(X_B)}+\sqrt{\calF_{\rho_A\otimes\rho_B}(Z)}.\label{S6_ex}
\end{align}
Therefore, we obtain
\eq{
\sqrt{\calF_{\rho_A\otimes\rho_B}(V^\dagger1_{A'}\otimes X_{B'}V)}\le\sqrt{\calF_{\rho_B}(X_B)}+\Delta_Z+\Delta_{A,A',\rho_A}.
\label{eq:S26}
}

Similarly, we evaluate $|\ex{[N,V^\dagger1_{A'}\otimes X_{B'}V]}_{\rho_{A}\otimes\rho_{B}}|$ in the same manner as \eqref{S9} but use $X_{A}+X_{B}+Z=V^{\dagger}(X_{A'}+X_{B'})V$ instead of $X_{A}+X_{B}=V^{\dagger}(X_{A'}+X_{B'})V$:
\begin{align}
\ex{[N,V^\dagger1_{A'}\otimes X_{B'}V]}_{\rho_{A}\otimes\rho_{B}}
&=
-\ex{[Q\otimes1_{B},V^\dagger(1_{A'}\otimes X_{B'})V]}_{\rho_A\otimes\rho_B}\nonumber\\
&=-\ex{[Q\otimes1_{B},X_A\otimes1_B+1_A\otimes X_B-V^\dagger(X_{A'}\otimes 1_{B'})V+Z]}_{\rho_A\otimes\rho_B}\nonumber\\
&=-\ex{[Q,Y_A]}_{\rho_A}-\ex{[Q\otimes1_{B},Z]}_{\rho_A\otimes\rho_B}.\label{S9_ex}
\end{align}
Therefore, we obtain
\eq{
|\ex{[N,V^\dagger1_{A'}\otimes X_{B'}V]}_{\rho_{A}\otimes\rho_{B}}|&\ge|\ex{[Q,Y_A]}_{\rho_A}|-|\ex{[Q\otimes1_{B},Z]}_{\rho_A\otimes\rho_B}|\non
&\ge\ex{[Q,Y_A]}_{\rho_A}|-\sqrt{\calF_{\rho_A\otimes\rho_B}(Q\otimes1_B)}\sqrt{V_{\rho_A\otimes\rho_B}(Z)}\non
&\stackrel{(a)}{\ge}\ex{[Q,Y_A]}_{\rho_A}|-\frac{\Delta_Z}{2}.
\label{eq:S28}
}
Here we used $\calF_{\rho_A\otimes\rho_B}(Q\otimes1_B)\le1$, which follows from $0\le Q\le 1_A$, and $V_{\rho_A\otimes\rho_B}(Z)\le\frac{\Delta^2_{Z}}{4}$.
Combining \eqref{10''}, \eqref{eq:S26}, and \eqref{eq:S28}, we obtain \eqref{convertS_ex}.
\hspace*{\fill}~\QED\par\endtrivlist\unskip

\section{Trade-off structure between symmetry, irreversibility and Quantum coherence }\label{SIsec2}
In this section, we derive \eqref{SIQ-Cini} and \eqref{SIQ-Gini} in the main text from Lemma \ref{L1} that describe the trade-off structure between symmetry, irreversibility and quantum coherence. We also give improved versions of \eqref{SIQ-Cini} and \eqref{SIQ-Gini}.
In the subsection B, we extend the trade off to the case of violated conservation law.
In the subsection C--E, we also derive several properties of quantities in the inequalities, i.e. $\calC$, $\Delta$ and $\delta$, introduced in the Materials and Methods.
In the subsection F, we discuss examples of the suppression effect on the irreversibility imposed by quantum coherence.

\subsection{Derivation of (\ref{SIQ-Cini}) and (\ref{SIQ-Gini}) (and some extensions) in the main text}
To this end, we prove the following theorem that includes \eqref{SIQ-Cini} and \eqref{SIQ-Gini} as special cases:
\begin{theorem}\label{T1S}
Let us consider two quantum systems $A$ and $A'$, and Hermitian operators $X_A$ and $X_{A'}$ on them.
Let $\calE$ be a CPTP map from $A$ to $A'$ which is implemented by  unitary interaction with another system $B$ that satisfies the conservation law of $X$. To be concrete, we introduce a tuple $(U,\rho_B,X_B,X_{B'})$ of a unitary $U$ on $AB$, a state $\rho_B$ on $B$, an operator $X_B$ on $B$ and an operator $X_{B'}$ on $B'$, where $B'$ is a quantum system satisfying $AB=A'B'$, and assume that 
\eq{
\calE(...)=\Tr_{B'}[U(...\otimes\rho_B)U^{\dagger}],\enskip
X_{A}+X_{B}=U^{\dagger}(X_{A'}+X_{B'})U.
}
We also take a test ensemble $\{p_k,\rho_k\}$ where $\{\rho_k\}$ is a set of quantum states and $\{p_k\}$ is a probability distribution.
We define two measures of irreversibility of $\calE$ for the test ensemble $\{p_k,\rho_k\}$ as
\eq{
\delta:=\sqrt{\sum_kp_k\delta^2_k},\enskip
\delta_T:=\sum_kp_k\delta_{k,T},
}
where $\delta_k:=D_F(\rho_k,\calR\circ\calE(\rho_k))$ and $\delta_{k,T}:=T(\rho_k,\calR\circ\calE(\rho_k))$ where $T(\rho,\sigma):=\|\rho-\sigma\|_1/2$.
Then, for arbitrary $\{\rho_k,p_k\}$, the following relation holds:
\eq{
\frac{\calC}{\Delta+\sqrt{\calF_{\rho_{B}}(X_{B})}}\le\sqrt{\dmul}\label{SIQ-GiniS}
}
Here, we can substitute either the following $\dmo$ or $\dmt$ for $\dmul$:
\eq{
\dmo&:=\delta\times\overline{T}\\
\dmt&:=\delta_T\times\left(1-\min_kp_k\right),
}
where $\overline{T}:=\sqrt{\sum_{k,k'}p_{k}p_{k'}T(\rho_k,\rho_{k'})^2}\le 1$  and 
\eq{
\calC:=\sqrt{\sum_{k\ne k'}p_kp_{k'}\Tr[(\rho_{k}-\rho_{k'})_{+}Y(\rho_{k}-\rho_{k'})_{-}Y]}.
}
where $(\rho_k-\rho_{k'})_{\pm}$ is the positive/negative part of $\rho_k-\rho_{k'}$, and $Y:=X_A-\calE^\dagger(X_{A'})$.
And $\Delta$ is defined as
\eq{
\Delta:=\max_{\rho\in\cup_k\mathrm{supp}(\rho_k)}\sqrt{\calF_{\rho\otimes\rho_B}(X_A\otimes 1_B-U^\dagger X_{A'}\otimes1_{B'}U)},
}
where the minimum runs over the subspace which is the sum of the supports the test states $\{\rho_k\}$.
Furthermore, when $\{\rho_k\}$ are orthogonal to each other, i.e., when $F(\rho_k,\rho_{k'})=0$ for any $k\ne k'$,
\begin{align}
\frac{\calC}{\Delta+\sqrt{\calF_{\rho_{B}}(X_{B})}}\le\delta\times\sqrt{1-\min_kp_k}.\label{SIQ-CiniS}
\end{align}

Furthermore, we can also substitute $\Delta_C$ for $\Delta$ in \eqref{SIQ-GiniS} and \eqref{SIQ-CiniS}:
\eq{
\Delta_C:=\max_{\rho\in\cup_k\mathrm{supp}(\rho_k)}\calF_{\rho\otimes\rho_B}(X_A\otimes 1_B-U^\dagger X_{A'}\otimes1_{B'}U)+2\calF_{\rho\otimes\rho_B}(X_A\otimes 1_B-U^\dagger X_{A'}\otimes1_{B'}U,1\otimes X_B).\label{definitionofDeltaC}
}
To be concrete, for arbitrary test ensembles, the following inequality holds:
\eq{
\frac{\calC}{\sqrt{\Delta_C+\calF_{\rho_{B}}(X_{B})}}\le\sqrt{\dmul}.\label{SIQ-GiniS_im}
}
And for test ensembles whose test states $\{\rho_k\}$ are orthogonal each other, the following inequality holds:
\begin{align}
\frac{\calC}{\sqrt{\Delta_C+\calF_{\rho_{B}}(X_{B})}}\le\delta\times\sqrt{1-\min_kp_k}\label{SIQ-CiniS_im}.
\end{align}
\end{theorem}

Clearly, \eqref{SIQ-Gini} and \eqref{SIQ-Cini} are direct corollaries of \eqref{SIQ-GiniS} and \eqref{SIQ-CiniS} due to $\overline{T}\le 1$ and $1-\min_kp_k\le1$. We also remark that when the test ensemble $\{p_k,\rho_k\}_{k\in \calK}$ satisfies $\calK=\{0,1\}$ and $p_k=1/2$, then $\sqrt{1-\min_kp_k}=\sqrt{1/2}$. Therefore, for an arbitrary test ensemble in the form of $\{1/2,\rho_k\}_{k=1,2}$, the bound \eqref{SIQ-CiniS} becomes stronger than  \eqref{SIQ-Cini} by $\sqrt{2}$.

We prove \eqref{SIQ-GiniS} and \eqref{SIQ-GiniS_im}, and \eqref{SIQ-CiniS} and \eqref{SIQ-CiniS_im} separately. We first prove \eqref{SIQ-CiniS} and \eqref{SIQ-CiniS_im}.
\begin{proofof}{\eqref{SIQ-CiniS} and  \eqref{SIQ-CiniS_im}}
We show \eqref{SIQ-CiniS} and  \eqref{SIQ-CiniS_im} under the assumption that $\{\rho_k\}$ are orthogonal to each other.
We prove \eqref{SIQ-CiniS_im} first and then derive \eqref{SIQ-CiniS} as its corollary.
Note that we can take a projective measurement $\{Q_k\}$ that completely distinguishes $\{\rho_k\}$ in this case, i.e., $\Tr[Q_k\rho_k]=1$.
We define a CPTP map $\calQ(...):=\sum_k\Tr[Q_k...]\ket{k}\bra{k}_S$ where $\{\ket{k}_S\}$ is a set of orthogonal states.
Then, by the monotonicity of $D_F$, we obtain 
\begin{align}
\delta_k\ge D_F(\calQ(\rho_k),\calQ\circ\calR\circ\calE(\rho_k)).
\end{align}
The above implies
\begin{align}
\Tr[\calE^\dagger\circ\calR^\dagger(Q_k)\rho_k]&\ge1-\delta^2_k,\label{S26}\\
\Tr[\calE^\dagger\circ\calR^\dagger(Q_{k'})\rho_k]\le\Tr[(1-\calE^\dagger\circ\calR^\dagger(Q_{k}))\rho_k]&\le\delta^2_k\enskip (k\ne k').\label{S27}
\end{align}
Now, let us take a spectral decompotion of $\rho_k$ as $\rho_k=\sum_{l}q^{(k)}_{l}\psi^{(k)}_{l}$ (here we use the abbreviation $\psi^{(k)}_l:=\ket{\psi^{(k)}_l}\bra{\psi^{(k)}_l}$. we use similar abbreviations for density matrices of pure states), and define
\begin{align}
1-\delta^{2}_{(k),l}:=\Tr[\calE^\dagger\circ\calR^\dagger(Q_k)\psi^{(k)}_{l}],\\
\delta^2_{(k),[k'],l}:=\Tr[\calE^\dagger\circ\calR^\dagger(Q_{k'})\psi^{(k)}_{l}].
\end{align} 
Then, due to \eqref{S26} and \eqref{S27},
\begin{align}
\sum_{l}q^{(k)}_{l}\delta^{2}_{(k),l}&\le\delta^{2}_k,\\
\sum_{l}q^{(k)}_{l}\delta^{2}_{(k),[k'],l}&\le\delta^{2}_k,\enskip\forall k'\mbox{ s.t. } k'\ne k.
\end{align}
Let us define $\rho^{k,k',l,l',\theta}_A$ as
\begin{align}
\rho^{k,k',l,l',\theta}_A:=\left(\frac{\ket{\psi^{(k)}_{l}}+e^{i\theta}\ket{\psi^{(k')}_{l'}}}{\sqrt{2}}\right)\left(\frac{\bra{\psi^{(k)}_{l}}+e^{-i\theta}\bra{\psi^{(k')}_{l'}}}{\sqrt{2}}\right)\label{CBpre3},
\end{align}
for $\theta \in \mathbb{R}$.
Note that $\rho^{k,k',l,l',\theta}_A$ for $k\ne k'$ satisfies
\begin{align}
Q_k\rho^{k,k',l,l',\theta}_AQ_k&=\frac{\psi^{(k)}_{l}}{2},\\
(1-Q_k)\rho^{k,k',l,l',\theta}_A(1-Q_k)&=\frac{\psi^{(k')}_{l'}}{2}.
\end{align}
Therefore, $\rho^{k,k',l,l',\theta}_A$ and $P_k:=\calR^\dagger(Q_k)$ satisfy
\begin{align}
\Tr[(1-\calE^\dagger(P_k))Q_k\rho^{k,k',l,l',\theta}_AQ_k]+\Tr[\calE^\dagger(P_k)(1-Q_k)\rho^{k,k',l,l',\theta}_A(1-Q_k)]
&=\frac{\Tr[(1-\calE^\dagger\circ\calR^\dagger(Q_k))\psi^{(k)}_{l}]
+\Tr[\calE^\dagger\circ\calR^\dagger(Q_k)\psi^{(k')}_{l'}]}{2}\nonumber\\
&=\frac{\delta^2_{(k),l}+\delta^2_{(k'),[k],l'}}{2}.\label{S35}
\end{align}
Combining Lemma \ref{L1S} and \eqref{S35}, we obtain 
\begin{align}
\frac{\delta^2_{(k),l}+\delta^2_{(k'),[k],l'}}{2}&\ge\frac{|\ex{[Q_k,Y_A]}_{\rho^{k,k',l,l',\theta}_A}|^2}{\calF_{\rho^{k,k',l,l',\theta}_A\otimes\rho_B}(T_{AB})+2\calF_{\rho^{k,k',l,l',\theta}_A\otimes\rho_B}(T_{AB},1_A\otimes X_B)+\calF_{\rho_B}(X_B)}\non
&\ge\frac{|\ex{[Q_k,Y_A]}_{\rho^{k,k',l,l',\theta}_A}|^2}{\max_{\rho\in\cup_k\mathrm{supp}(\rho_k)}(\calF_{\rho\otimes\rho_B}(T_{AB})+\calF_{\rho\otimes\rho_B}(T_{AB},1_A\otimes X_B))+\calF_{\rho_B}(X_B)}
\non
&=\frac{|\ex{[Q_k,Y_A]}_{\rho^{k,k',l,l',\theta}_A}|^2}{\Delta_C+\calF_{\rho_B}(X_B)}\label{S34A}
\end{align}
Here we use an abbreviation $T_{AB}:=X_A\otimes1_{B}-U^\dagger X_{A'}\otimes1_{B'}U$.

We evaluate $\ex{[Q_k,Y_A]}_{\rho^{k,k',l,l',\theta}_A}$ as
\begin{align}
\ex{[Q_k,Y_A]}_{\rho^{k,k',l,l',\theta}_A}=\frac{e^{i\theta}\bra{\psi^{(k)}_{l}}Y_A\ket{\psi^{(k')}_{l'}}-(c.c)}{2}.\label{CQev0}
\end{align}
Therefore, by defining $e^{i\eta}:=\frac{\bra{\psi^{(k)}_{l}}Y_A\ket{\psi^{(k')}_{l'}}}{|\bra{\psi^{(k)}_{l}}Y_A\ket{\psi^{(k')}_{l'}}|}$ and taking $\theta:=\frac{\pi}{2}-\eta$, we obtain 
\begin{align}
\ex{[Q_k,Y_A]}_{\rho^{k,k',l,l',\theta}_A}=i|\bra{\psi^{(k)}_{l}}Y_A\ket{\psi^{(k')}_{l'}}|.\label{CQev0'}
\end{align}
Therefore,
\begin{align}
\frac{\delta^2_{(k),l}+\delta^2_{(k'),[k],l'}}{2}&\ge\frac{|\bra{\psi^{(k)}_{l}}Y_A\ket{\psi^{(k')}_{l'}}|^2}{\Delta_C+\calF_{\rho_B}(X_B)}\label{Cstarbb}
\end{align}
Multiplying by $p_kp_{k'}q^{(k)}_{l}q^{(k')}_{l'}$, summing for $l$ and $l'$, and summing for $k$ and $k'$ with $k\ne k'$, we obtain \eqref{SIQ-CiniS_im} as follows:
\eq{
\frac{\calC^2}{\Delta_C+\calF_{\rho_B}(X_B)}
&=\sum_{k\ne k'}\sum_{l,l'}p_kp_{k'}q^{(k)}_{l}q^{(k')}_{l'}\frac{|\bra{\psi^{(k)}_{l}}Y_A\ket{\psi^{(k')}_{l'}}|^2}{(\Delta+\sqrt{\calF_{\rho_B}(X_B)})^2}\non
&\le\sum_{k\ne k'}\sum_{l,l'}p_kp_{k'}q^{(k)}_{l}q^{(k')}_{l'}\frac{\delta^2_{(k),l}+\delta^2_{(k'),[k],l'}}{2}\non
&\le\sum_{k\ne k'}\sum_{l'}p_kp_{k'}q^{(k')}_{l'}\frac{\delta^2_{k}}{2}+
\sum_{k\ne k'}\sum_{l}p_kp_{k'}q^{(k)}_{l}\frac{\delta^2_{k'}}{2}\non
&=\sum_{k\ne k'}p_kp_{k'}\frac{\delta^2_{k}}{2}+
\sum_{k\ne k'}p_kp_{k'}\frac{\delta^2_{k'}}{2}\non
&=\sum_{k}p_k(1-p_k)\delta^2_k\le\left(1-\min_kp_k\right)\times\delta^2.\label{S41}
}

Let us derive \eqref{SIQ-CiniS} from \eqref{SIQ-CiniS_im}.
To do so, we only have to show the following inequality:
\eq{
\sqrt{\Delta_C+\calF_{\rho_B}(X_B)}\le\Delta+\sqrt{\calF_{\rho_B}(X_B)}.\label{extension_key}
}
We derive this inequality as follows:
\eq{
\sqrt{\Delta_C+\calF_{\rho_B}(X_B)}&=\sqrt{\calF_{\rho_B}(X_B)+\max_{\rho\in\cup_k\mathrm{supp}(\rho_k)}(\calF_{\rho\otimes\rho_B}(T_{AB})+2\calF_{\rho\otimes\rho_B}(T_{AB},1\otimes X_B))}\non
&=\max_{\rho\in\cup_k\mathrm{supp}(\rho_k)}\sqrt{\calF_{\rho\otimes\rho_B}(T_{AB}+X_B)}\non
&\le\max_{\rho\in\cup_k\mathrm{supp}(\rho_k)}(\sqrt{\calF_{\rho\otimes\rho_B}(T_{AB})}+\sqrt{\calF_{\rho_B}(X_B)})\non
&=\Delta+\sqrt{\calF_{\rho_B}(X_B)}.
}
\end{proofof}

To prove \eqref{SIQ-GiniS} and \eqref{SIQ-GiniS_im}, we use the following lemma:
\begin{lemma}\label{L1G}
Let us consider a quantum system $S$ and two states $\rho_0$ and $\rho_1$ on it.
We suppose that a POVM $\{P,1-P\}$ and a real positive number $\delta$ satisfy
\begin{align}
\frac{1}{2}\|p_0-p_1\|_1\approx_{\delta}\frac{1}{2}\|\rho_0-\rho_1\|_1,
\end{align}
where $x\approx_{\delta}y\stackrel{\mathrm{def.}}{\Leftrightarrow}|x-y|\le\delta$, $p_0$ and $p_1$ are probability distributions that are defined as $p_{j}(+):=\Tr[\rho_jP]$ and $p_{j}(-):=\Tr[\rho_j(1-P)]$, and  $\|p_0-p_1\|_1:=|p_0(+)-p_1(+)|+|p_0(-)-p_1(-)|$.
Suppose that $P$ is taken so that $p_0(+) \geq p_1 (+)$.
Then, the following inequality holds:
\begin{align}
\Tr[(1-P)(\rho_0-\rho_1)_+]+\Tr[P(\rho_0-\rho_1)_-]\le\delta.\label{L1-1}
\end{align} 
\end{lemma}

\begin{proofof}{Lemma \ref{L1G}}
We first note the following:
\begin{align}
|p_0(+)-p_{1}(+)|&=\frac{1}{2}\|p_0-p_1\|_1\nonumber\\
&\approx_{\delta}\frac{1}{2}\|\rho_0-\rho_1\|_1\nonumber\\
&=\Tr[(\rho_0-\rho_1)_+].
\end{align}
Therefore, because of the definition $\{p_{j}(+)\}_{j=0,1}$, we obtin
\begin{align}
\Tr[(\rho_0-\rho_1)_+]\approx_\delta|\Tr[P(\rho_0-\rho_1)]|.
\end{align}
Then, we obtain 
\begin{align}
\Tr[(\rho_0-\rho_1)_+]-\delta&\le|\Tr[P(\rho_0-\rho_1)]|\nonumber\\
&=\Tr[P(\rho_0-\rho_1)_{+}]-\Tr[P(\rho_0-\rho_1)_{-}].
\end{align}
This is equivalent to the desired inequality \eqref{L1-1}.
\end{proofof}

Now, let us prove \eqref{SIQ-GiniS} and \eqref{SIQ-GiniS_im}.
\begin{proofof}{\eqref{SIQ-GiniS} and \eqref{SIQ-GiniS_im}}
We can easily obtain \eqref{SIQ-GiniS} from \eqref{SIQ-GiniS_im} using \eqref{extension_key}.
Therefore, we only have to derive \eqref{SIQ-GiniS_im}.
Due to the definition of $\delta_k$ $\delta_{k,T}$, the triangle inequality and the monotonicity of the trace norm, we obtain
\begin{align}
\frac{1}{2}\|\rho_k-\rho_{k'}\|_1&\le\frac{1}{2}\|\rho_k-\calR\circ\calE(\rho_{k})\|_1+\frac{1}{2}\|\calR\circ\calE(\rho_k)-\calR\circ\calE(\rho_{k'})\|_1
+\frac{1}{2}\|\rho_{k'}-\calR\circ\calE(\rho_{k'})\|_1\nonumber\\
&\le\delta_{k,T}+\delta_{k',T}+\frac{1}{2}\|\calE(\rho_k)-\calE(\rho_{k'})\|_1.\label{A}
\end{align}

Let us define $P_{k,k'}$ as the projection to the support of $(\calE(\rho_k)-\calE(\rho_{k'}))_+$. Then, $\{p_k(+)\}$ which are defined as $p_k(+):=\Tr[P_{k,k'}\calE(\rho_k)]$ satisfy
\begin{align}
p_k(+)-p_{k'}(+)=\frac{1}{2}\|\calE(\rho_k)-\calE(\rho_{k'})\|_1\ge0.\label{A'}
\end{align}
Here, we note that $p_k(+)=\Tr[\calE^\dagger(P_{k,k'})\rho_k]$. Therefore, by defining $p_k(-):=\Tr[\calE^\dagger(1-P_{k,k'})\rho_k]$, \eqref{A} implies
\begin{align}
\frac{1}{2}\|p_k-p_{k'}\|_1\approx_{\delta_{k,T}+\delta_{k',T}}\frac{1}{2}\|\rho_k-\rho_{k'}\|_1,\label{A''}
\end{align}
where $\|p_k-p_{k'}\|_1:=|p_k(+)-p_{k'}(+)|+|p_{k}(-)-p_{k'}(-)|$.
(Note that $\frac{1}{2}\|p_k-p_{k'}\|_1\le\frac{1}{2}\|\rho_k-\rho_{k'}\|_1$ holds by definition.)

By applying Lemma \ref{L1G} to \eqref{A''}, we obtain
\begin{align}
\Tr[(1-\calE^\dagger(P_{k,k'}))(\rho_k-\rho_{k'})_{+}]+\Tr[\calE^\dagger(P_{k,k'})(\rho_k-\rho_{k'})_{-}]\le\delta_{k,T}+\delta_{k',T}.\label{Bpre}
\end{align}
Now, let us take the spectral decomposition $(\rho_k-\rho_{k'})_{\pm}=\sum_{j}\tilde{q}^{(\pm,k,k')}_j\phi^{(\pm,k,k')}_j$ and define 
\begin{align}
\delta^{(+,k,k')}_j&:=\Tr[(1-\calE^\dagger(P_{k,k'}))\phi^{(+,k,k')}_j]\\
\delta^{(-,k,k')}_j&:=\Tr[\calE^\dagger(P_{k,k'})\phi^{(-,k,k')}_j].
\end{align}
By substituting $(\rho_k-\rho_{k'})_{\pm}=\sum_{j}\tilde{q}^{(\pm,k,k')}_j\phi^{(\pm,k,k')}_j$ into \eqref{Bpre} and using $\sum_j\tilde{q}^{(+,k,k')}_j=\sum_j\tilde{q}^{(-,k,k')}_j=\frac{1}{2}\|\rho_k-\rho_{k'}\|_1$, we obtain
\begin{align}
\sum_{j,j'}\tilde{q}^{(+,k,k')}_j\tilde{q}^{(-,k,k')}_{j'}(\delta^{(+,k,k')}_j+\delta^{(-,k,k')}_{j'})\le(\delta_{k,T}+\delta_{k',T})\times\frac{\|\rho_k-\rho_{k'}\|_1}{2}.\label{Bpre2}
\end{align}
Now, let us define $Q_{k,k'}$ as the projection onto the support of $(\rho_k-\rho_k')_+$ and $\rho^{j,j',k,k',\theta}_S$ as
\begin{align}
\rho^{j,j',k,k',\theta}_A:=\left(\frac{\ket{\phi^{(+,k,k')}_j}+e^{i\theta}\ket{\phi^{(-,k,k')}_{j'}}}{\sqrt{2}}\right)\left(\frac{\bra{\phi^{(+,k,k')}_j}+e^{-i\theta}\bra{\phi^{(-,k,k')}_{j'}}}{\sqrt{2}}\right).\label{Bpre3}
\end{align}
Then,
\begin{align}
\Tr[(1-\calE^\dagger(P))Q_{k,k'}\rho^{j,j',k,k',\theta}_AQ_{k,k'}]+\Tr[\calE^\dagger(P)(1-Q_{k,k'})\rho^{j,j',k,k',\theta}_A(1-Q_{k,k'})]=\frac{\delta^{(+,k,k')}_j+\delta^{(-,k,k')}_{j'}}{2}\label{Bpre4}
\end{align}
Therefore, by using Lemma \ref{L1S} and \eqref{Bpre4}, we obtain 
\begin{align}
\frac{\delta^{(+,k,k')}_j+\delta^{(-,k,k')}_{j'}}{2}&\ge\frac{\left|\ex{[Q_{k,k'},Y_A]}_{\rho^{j,j',k,k',\theta}_A}\right|^2}{\calF_{\rho^{j,j'k,k',\theta}_A\otimes\rho_B}(T_{AB})+2\calF_{\rho^{j,j'k,k',\theta}_A\otimes\rho_B}(T_{AB},1_A\otimes X_B)+\calF_{\rho_{B}}(X_{B})}\non
&\ge\frac{|\ex{[Q_k,Y_A]}_{\rho^{j,j',k,k',\theta}_A}|^2}{(\max_{\rho\in\cup_{k}\mathrm{supp}(\rho_k)}(\calF_{\rho\otimes\rho_B}(T_{AB})+2\calF_{\rho\otimes\rho_B}(T_{AB},1_A\otimes X_B))+\calF_{\rho_B}(X_B)}\non
&=\frac{|\ex{[Q_k,Y_A]}_{\rho^{j,j',k,k',\theta}_A}|^2}{\Delta_C+\calF_{\rho_B}(X_B)}\label{starB}
\end{align}
Therefore, we obtain
\begin{align}
\frac{\delta^{(+,k,k')}_j+\delta^{(-,k,k')}_{j'}}{2}\ge\frac{\left|\ex{[Q_{k,k'},Y_A]}_{\rho^{j,j',k,k',\theta}_A}\right|^2}{\Delta_C+\calF_{\rho_{B}}(X_{B})}\label{star}
\end{align}

We can easily evaluate the term  $\ex{[Q_{k,k'},Y_A]}_{\rho^{j,j',k,k',\theta}_S}$ as follows:
\begin{align}
\ex{[Q_{k,k'},Y_A]}_{\rho^{j,j',k,k',\theta}_S}=\frac{e^{i\theta}\bra{\phi^{(+,k,k')}_j}Y_A\ket{\phi^{(-,k,k')}_{j'}}-(c.c)}{2}.\label{Qev}
\end{align}
Here, let us define $e^{i\eta}:=\frac{\bra{\phi^{(+,k,k')}_j}Y_A\ket{\phi^{(-,k,k')}_{j'}}}{|\bra{\phi^{(+,k,k')}_j}Y_A\ket{\phi^{(-,k,k')}_{j'}}|}$ and set $\theta:=\frac{\pi}{2}-\eta$. 
Then, we have
\begin{align}
\left|\ex{[Q,Y_A]}_{\rho^{j,j',k,k',\theta}_S}\right|=\left|\bra{\phi^{(+,k,k')}_j}Y_A\ket{\phi^{(-,k,k')}_{j'}}\right|.\label{S55}
\end{align}
Substituting the above into \eqref{star}, multiplying by $p_kp_{k'}\tilde{q}^{(+,k,k')}_j\tilde{q}^{(-,k,k')}_{j'}$, and summing for  $j$, $j'$, $k$ and $k'$ with $k\ne k'$, we obtain
\begin{align}
\frac{\calC^2}{\Delta_C+\calF_{\rho_{B}}(X_{B})}\le\sum_{k\ne k'}\sum_{j,j'}p_kp_{k'}\tilde{q}^{(+,k,k')}_j\tilde{q}^{(-,k,k')}_{j'}\frac{\delta^{(+,k,k')}_j+\delta^{(-,k,k')}_{j'}}{2}
\label{f1}
\end{align}
where $T(\rho_k,\rho_{k'}):=\|\rho_k-\rho_{k'}\|_1/2$.
To obtain \eqref{SIQ-GiniS}, we evaluate the RHS of \eqref{f1} in the following two ways:
\begin{align}
(\mbox{RHS in \eqref{f1}})&\stackrel{(a)}{\le}\sum_{k\ne k'}p_kp_{k'}(\delta_{k,T}+\delta_{k',T})\times\frac{T(\rho_k,\rho_{k'})}{2}\non
&\stackrel{(b)}{\le}\sum_{k,k'}p_kp_{k'}(\delta_{k}+\delta_{k'})\times\frac{T(\rho_k,\rho_{k'})}{2}
\non
&\le\frac{1}{2}\left(\sum_{k'}p_{k'}\sqrt{\sum_{k}p_{k}\delta^2_{k}}\sqrt{\sum_{k}p_kT(\rho_k,\rho_{k'})^2}+\sum_{k}p_{k}\sqrt{\sum_{k'}p_{k'}\delta^2_{k'}}\sqrt{\sum_{k'}p_{k'}T(\rho_k,\rho_{k'})^2}\right)\nonumber\\
&\le\delta\sqrt{\sum_{k,k'}p_{k}p_{k'}T(\rho_k,\rho_{k'})^2}=\delta\times\overline{T}\non
&=\dmo,\label{f2}\\
(\mbox{RHS in \eqref{f1}})&\le
\sum_{k\ne k'}p_kp_{k'}(\delta_{k,T}+\delta_{k',T})\times\frac{T(\rho_k,\rho_{k'})}{2}\non
&\le\frac{1}{2}\sum_{k\ne k'}p_kp_{k'}(\delta_{k,T}+\delta_{k',T})\nonumber\\
&=\sum_kp_k(1-p_k)\delta_{k,T}\le\delta_T\times\left(1-\min_kp_k\right)\non
&=\dmt .\label{f2'}
\end{align}
Here we used \eqref{Bpre2} in $(a)$, and used $T(\rho,\sigma)\le D_F(\rho,\sigma)$ and $T(\rho_k,\rho_k)=0$ in $(b)$. Therefore, we obtain 
\eq{
\sum_{k\ne k'}\sum_{j,j'}p_kp_{k'}\tilde{q}^{(+,k,k')}_j\tilde{q}^{(-,k,k')}_{j'}\frac{\delta^{(+,k,k')}_j+\delta^{(-,k,k')}_{j'}}{2}\le\dmul.\label{f2''}
}
By combining \eqref{f1} and \eqref{f2''}, we obtain \eqref{SIQ-GiniS_im}.
\end{proofof}

\subsection{Extension to the case of violated conservation law}\label{SI_violated_SIQ}
Similarly to Lemma \ref{L1S}, we can extend Theorem \ref{T1S} to the case of the violated conservation law.
When $Z=U^\dagger(X_{A'}+X_{B'})U-(X_A+X_B)$ holds, we obtain the following relation for an arbitrary "orthogonal" test ensemble $\{p_k,\rho_k\}$ that satisfies $F(\rho_k,\rho_{k'})=0$ for $k\ne k'$:
\eq{
\frac{\calC-\frac{\Delta_Z}{2}}{\Delta+\Delta_Z+\sqrt{\calF_{\rho_{B}}(X_{B})}}\le\delta\times\sqrt{1-\min_kp_k}.\label{SIQ-CiniS_ex}
}
For an arbitrary test ensemble, we obtain
\eq{
\frac{\calC-\frac{\Delta_Z}{2}}{\Delta+\Delta_Z+\sqrt{\calF_{\rho_{B}}(X_{B})}}&\le\sqrt{\delta}\times\sqrt{\overline{T}}\non
&\le\sqrt{\delta}.\label{SIQ-GiniS_ex}
}
Therefore, as we pointed out in the main text, we can extend our main results \eqref{SIQ-Cini} and \eqref{SIQ-Gini} to the case of the violated conservation law by substituting
\eq{
\calC\rightarrow\calC-\frac{\Delta_Z}{2},\enskip\Delta\rightarrow\Delta+\Delta_Z.
}

\begin{proofof}{\eqref{SIQ-CiniS_ex} and \eqref{SIQ-GiniS_ex}}
We first derive \eqref{SIQ-CiniS_ex}.
The proof of \eqref{SIQ-CiniS_ex} is almost the same as \eqref{SIQ-CiniS}, except for we use \eqref{convertS_ex} instead of \eqref{convertS}.
The proof is the same as that of \eqref{SIQ-CiniS} to the front of  \eqref{S34A}.
In \eqref{S34A}, we use \eqref{convertS_ex}, and obtain
\begin{align}
\sqrt{\frac{\delta^2_{(k),l}+\delta^2_{(k'),[k],l'}}{2}}\ge\frac{|\ex{[Q_k,Y_A]}_{\rho^{k,k',l,l',\theta}_A}|-\frac{\Delta_Z}{2}}{\Delta+\Delta_Z+\sqrt{\calF_{\rho_B}(X_B)}}.\label{S34_ex}
\end{align}
By using \eqref{CQev0'}, we obtain
\begin{align}
\left(\frac{\Delta_Z}{2}+\left(\Delta+\Delta_Z+\sqrt{\calF_{\rho_B}(X_B)}\right)\sqrt{\frac{\delta^2_{(k),l}+\delta^2_{(k'),[k],l'}}{2}}\right)^2
\ge|\bra{\psi^{(k)}_{l}}Y_A\ket{\psi^{(k')}_{l'}}|^2.\label{Cstarbb_ex}
\end{align}
Multiplying by $p_kp_{k'}q^{(k)}_{l}q^{(k')}_{l'}$, summing for $l$ and $l'$, and summing for  $k$ and $k'$ with $k\ne k'$, we obtain  \eqref{SIQ-CiniS_ex} as follows:
\begin{align}
\calC^2&\le \sum_{k\ne k'}\sum_{l,l'}p_kp_{k'}q^{(k)}_{l}q^{(k')}_{l'}\left(\frac{\Delta^2_Z}{4}+\Delta_Z\Gamma\sqrt{\frac{\delta^2_{(k),l}+\delta^2_{(k'),[k],l'}}{2}}+\Gamma^2\left(\frac{\delta^2_{(k),l}+\delta^2_{(k'),[k],l'}}{2}\right)\right)\non
&\stackrel{(a)}{\le}\frac{\Delta^2_Z}{4}+\Delta_Z\Gamma\sqrt{\sum_{k\ne k'}\sum_{l,l'}p_kp_{k'}q^{(k)}_{l}q^{(k')}_{l'}}\sqrt{\sum_{k\ne k'}\sum_{l,l'}p_kp_{k'}q^{(k)}_{l}q^{(k')}_{l'}\frac{\delta^2_{(k),l}+\delta^2_{(k'),[k],l'}}{2}}\non
&+\sum_{k\ne k'}\sum_{l,l'}p_kp_{k'}q^{(k)}_{l}q^{(k')}_{l'}\Gamma^2\left(\frac{\delta^2_{(k),l}+\delta^2_{(k'),[k],l'}}{2}\right)\non
&\stackrel{(b)}{\le}\frac{\Delta^2_Z}{4}+\Delta_Z\Gamma\delta\sqrt{1-\min_kp_k}+\Gamma^2\delta^2(1-\min_kp_k)\non
&=\left(\frac{\Delta_Z}{2}+\left(\Delta+\Delta_Z+\sqrt{\calF_{\rho_B}(X_B)}\right)\sqrt{1-\min_kp_k}\delta\right)^2,
\end{align}
where we use the abbreviation $\Gamma:=\left(\Delta+\Delta_Z+\sqrt{\calF_{\rho_B}(X_B)}\right)$. In (a), we use $\sum_jr_ja_j\le\sqrt{\sum_jr_j}\sqrt{\sum_jr_ja^2_j}$ for arbitrary non-negative numbers $\{r_j\}$ and $\{a_j\}$. In (b) we also use $\sum_{k\ne k'}\sum_{l,l'}p_kp_{k'}q^{(k)}_{l}q^{(k')}_{l'}\frac{\delta^2_{(k),l}+\delta^2_{(k'),[k],l'}}{2}\le \left(1-\min_kp_k\right)\times\delta^2$ which is shown in \eqref{S41}.

Next, we derive \eqref{SIQ-GiniS_ex}. Again, the proof of \eqref{SIQ-GiniS_ex} is completely the same as that of \eqref{SIQ-Gini} in front of \eqref{starB}.
And in \eqref{starB}, we use \eqref{convertS_ex}, and obtain
\begin{align}
\sqrt{\frac{\delta^{(+,k,k')}_j+\delta^{(-,k,k')}_{j'}}{2}}\ge\frac{\left|\ex{[Q_{k,k'},Y_A]}_{\rho^{j,j',k,k',\theta}_S}\right|-\frac{\Delta_Z}{2}}{\Delta+\Delta_Z+\sqrt{\calF_{\rho_{B}}(X_{B})}}.\label{star_ex}
\end{align}
By using \eqref{S55}, we obtain
\eq{
\sqrt{\frac{\delta^{(+,k,k')}_j+\delta^{(-,k,k')}_{j'}}{2}}\ge\frac{\left|\bra{\phi^{(+,k,k')}_j}Y_A\ket{\phi^{(-,k,k')}_{j'}}\right|-\frac{\Delta_Z}{2}}{\Delta+\Delta_Z+\sqrt{\calF_{\rho_{B}}(X_{B})}}.
}
Therefore, we again use the abbreviation $\Gamma$ and obtain
\eq{
\left(\frac{\Delta_Z}{2}+\Gamma\sqrt{\frac{\delta^{(+,k,k')}_j+\delta^{(-,k,k')}_{j'}}{2}}\right)^2\ge\left|\bra{\phi^{(+,k,k')}_j}Y_A\ket{\phi^{(-,k,k')}_{j'}}\right|^2.
}
Multiplying $p_kp_{k'}\tilde{q}^{(+,k,k')}_j\tilde{q}^{(-,k,k')}_{j'}$, and summing up for $k,k',j,j'$, we obtain \eqref{SIQ-GiniS_ex} as follows:
\eq{
\calC^2&\le\sum_{k,k',j,j'}p_kp_{k'}\tilde{q}^{(+,k,k')}_j\tilde{q}^{(-,k,k')}_{j'}\left(\frac{\Delta_Z}{2}+\Gamma\sqrt{\frac{\delta^{(+,k,k')}_j+\delta^{(-,k,k')}_{j'}}{2}}\right)^2\non
&\le\overline{T}^2\frac{\Delta_Z^2}{4}+\Delta_Z\Gamma(\sum_{k,k',j,j'}p_kp_{k'}\tilde{q}^{(+,k,k')}_j\tilde{q}^{(-,k,k')}_{j'}\sqrt{\frac{\delta^{(+,k,k')}_j+\delta^{(-,k,k')}_{j'}}{2}})+\Gamma^2\sum_{k,k',j,j'}p_kp_{k'}\tilde{q}^{(+,k,k')}_j\tilde{q}^{(-,k,k')}_{j'}\frac{\delta^{(+,k,k')}_j+\delta^{(-,k,k')}_{j'}}{2}\non
&\le
\overline{T}^2\frac{\Delta_Z^2}{4}+\Delta_Z\Gamma\sqrt{\sum_{k,k',j,j'}p_kp_{k'}\tilde{q}^{(+,k,k')}_j\tilde{q}^{(-,k,k')}_{j'}}\sqrt{\sum_{k,k',j,j'}p_kp_{k'}\tilde{q}^{(+,k,k')}_j\tilde{q}^{(-,k,k')}_{j'}\frac{\delta^{(+,k,k')}_j+\delta^{(-,k,k')}_{j'}}{2}}\non
&+\Gamma^2\sum_{k,k',j,j'}p_kp_{k'}\tilde{q}^{(+,k,k')}_j\tilde{q}^{(-,k,k')}_{j'}\frac{\delta^{(+,k,k')}_j+\delta^{(-,k,k')}_{j'}}{2}\non
&\stackrel{(a)}{\le}\overline{T}^2\frac{\Delta_Z^2}{4}+\Delta_Z\Gamma\overline{T}\sqrt{\dmul}+\Gamma^2\dmul
\non
&\le\left(\frac{\Delta_Z}{2}+\left(\Delta+\Delta_Z+\sqrt{\calF_{\rho_B}(X_B)}\right)\sqrt{\dmul}\right)^2.
}
Here we use \eqref{f2''} in $(a)$.
\end{proofof}

\subsection{Derivations of upper bounds of $\Delta$}\label{SIsubsec_boundDelta}
We derive the upper bounds of $\Delta$ corresponding to \eqref{D_1}--\eqref{D_3} in the main text. For reader's convenience, we write down them again:
\eq{
\Delta&\le\Delta_1:=\Delta_{X_A}+\Delta_{X_{A'}},\label{D_1S}\\
\Delta&\le\Delta_2:=\Delta_Y+2\sqrt{\|\calE^\dagger(X^2_{A'})-\calE^\dagger(X_{A'})^2\|_\infty}\label{D_2S}\\
\Delta&\le\Delta_3:=\max_{\rho\in\cup_k\mathrm{supp}(\rho_k)}(\sqrt{\calF_{\rho}(Y)}+\sqrt{\calF_{\rho\otimes\rho_B}(U^\dagger X_{A'}\otimes1_{B'}U-\calE^\dagger(X_{A'})\otimes1_{B'})}).\label{D_3S}
}
\begin{proof}
To show \eqref{D_1S}, we evaluate $\sqrt{\calF_{\rho_A\otimes\rho_B}(X_A\otimes1_B-U^\dagger X_{A'}\otimes1_{B'}U)}$ as follows:
\eq{
\sqrt{\calF_{\rho_A\otimes\rho_B}(X_A\otimes1_B-U^\dagger X_{A'}\otimes1_{B'}U)}&\le\sqrt{\calF_{\rho_A\otimes\rho_B}(X_A\otimes1_B)}
+
\sqrt{\calF_{\rho_A\otimes\rho_B}(U^\dagger X_{A'}\otimes1_{B'}U)}\non
&=\sqrt{\calF_{\rho_A}(X_A)}
+
\sqrt{\calF_{U\rho_{A}\otimes\rho_{B}U^\dagger}(X_{A'}\otimes1_{B'})}\non
&\le\Delta_{X_A}+2\sqrt{V_{\calE(\rho_{A})}(X_{A'})}\non
&\le\Delta_{X_A}+\Delta_{X_{A'}}=\Delta_{1},
}
where we used \eqref{S7} in the first inequality.
From this inequality and the definition of $\Delta$, we obtain \eqref{D_1}.

Similarly, to obtain \eqref{D_3}, we evaluate $\sqrt{\calF_{\rho_A\otimes\rho_B}(X_A\otimes1_B-U^\dagger X_{A'}\otimes1_{B'}U)}$ as follows:
\eq{
\sqrt{\calF_{\rho_A\otimes\rho_B}(X_A\otimes1_B-U^\dagger X_{A'}\otimes1_{B'}U)}&=
\sqrt{\calF_{\rho_A\otimes\rho_B}(X_A\otimes1_B-\calE^\dagger(X_{A'})\otimes1_B+\calE^\dagger(X_{A'})\otimes1_B-U^\dagger X_{A'}\otimes1_{B'}U)}
\non
&\le\sqrt{\calF_{\rho_A\otimes\rho_B}(Y\otimes1_B)}
+
\sqrt{\calF_{\rho_A\otimes\rho_B}(\calE^\dagger(X_{A'})\otimes1_B-U^\dagger X_{A'}\otimes1_{B'}U)}\non
&=\sqrt{\calF_{\rho_A}(Y)}
+
\sqrt{\calF_{\rho_A\otimes\rho_B}(\calE^\dagger(X_{A'})\otimes1_B-U^\dagger X_{A'}\otimes1_{B'}U)},
}
where we used \eqref{S7} in the first inequality.
From this inequality and the definition of $\Delta$, we obtain \eqref{D_3}.

Next, let us derive \eqref{D_2}. To do so, we only have to show $\Delta_3\le\Delta_2$.
We show $\Delta_3\le\Delta_2$ as follows:
\eq{
\Delta_{3}&\le\Delta_Y+\max_{\rho\in\cup_k\mathrm{supp}(\rho_k)}\sqrt{\calF_{\rho\otimes\rho_B}(\calE^\dagger(X_{A'})\otimes1_B-U^\dagger X_{A'}\otimes1_{B'}U)}\non
&\le\Delta_Y
+
\max_{\rho\in\cup_k\mathrm{supp}(\rho_k)}2\sqrt{V_{\rho\otimes\rho_B}(\calE^\dagger(X_{A'})\otimes1_B-U^\dagger X_{A'}\otimes1_{B'}U)}\non
&=\Delta_Y+\max_{\rho\in\cup_k\mathrm{supp}(\rho_k)}2(\ex{(\calE^\dagger(X_{A'})\otimes1_B-U^\dagger X_{A'}\otimes1_{B'}U)^2}_{\rho\otimes\rho_B}-\ex{\calE^\dagger(X_{A'})\otimes1_B-U^\dagger X_{A'}\otimes1_{B'}U}^2_{\rho\otimes\rho_B})^{1/2}\non
&=\Delta_Y+\max_{\rho\in\cup_k\mathrm{supp}(\rho_k)}2\sqrt{\ex{(\calE^\dagger(X^2_{A'})-\calE^\dagger(X_{A'})^2}_{\rho}}\non
&\le\Delta_Y+2\sqrt{\|\calE^\dagger(X^2_{A'})-\calE^\dagger(X_{A'})^2\|_\infty}\non
&=\Delta_2
}

\end{proof}

\subsection{Proof of properties of $\calC$}\label{SIsubsec_prop_C}
In this subsection, we show the following properties of $\calC$:
\begin{enumerate}
\item The following necessary and sufficient condition of $\calC>0$ holds:
\eq{
Y\not\propto I_A\Leftrightarrow \exists\{p_k,\rho_k\} \mbox{ s.t. } \calC>0.\label{nonshift-nonzero_S}
}
\item When the test states $\{\rho_k\}$ are orthogonal to each other, $\calC$ satisfies
\eq{
\left(\min_kp_k\lambda^{\min}_{>0}(\rho_k)\right)\frac{\mathrm{C}_\calF}{4}\le \calC^2\le\frac{\mathrm{C}_\calF}{8}.\label{skew_C_S}
}
Here, $\lambda^{\min}_{>0}(\xi)$ is the minimum of the positive eigenvalues of state $\xi$ and $\mathrm{C}_\calF$ is the convexity of the quantum Fisher information of operator $Y$:
\eq{
\mathrm{C}_\calF:=\sum_kp_k\calF_{\rho_k}(Y)-\calF_{\sum_kp_k\rho_k}(Y)
}.
\end{enumerate}

\begin{proofof}{\eqref{nonshift-nonzero_S}}: When $Y\not\propto I_A$, there are two eigenstates $\ket{\psi_0}$ and $\ket{\psi_1}$ of $Y$ with different eigenvalues, and we can easily see that $\calC$ is strictly greater than 0 for the test ensemble $\{\{1/2,1/2\},\{\psi_+,\psi_-\}\}$, where  $\ket{\psi_\pm}:=(\ket{\psi_0}\pm\ket{\psi_1})/\sqrt{2}$.
To show the converse part $Y\not\propto I_A\Leftarrow \exists\{p_k,\rho_k\} \mbox{ s.t. } \calC>0$, note that its contraposition is $Y\propto I_A\Rightarrow \forall\{p_k,\rho_k\}, \calC=0$. And when $Y\propto I_A$, the operator $Y$ commutes any operator, so $\calC_{k,k'}=\Tr[(\rho_{k}-\rho_{k'})_{+}(\rho_{k}-\rho_{k'})_{-}Y^2]=0$ holds. Therefore, the converse part of \eqref{nonshift-nonzero_S} also holds.
\end{proofof}

\begin{proofof}{\eqref{skew_C_S}}: 
We first remark that for any state $\sigma$, $\calF_\sigma(Y)$ can be written as
\eq{
\calF_\sigma(Y)&=\sum_{m,m'}\frac{2(r_m-r_{m'})^2}{r_{m}+r_{m'}}|Y_{m,m'}|^2\nonumber\\
&=2\sum_{m,m'}(r_m+r_{m'})|Y_{m,m'}|^2-\sum_{m,m'}\frac{8r_{m}r_{m'}}{r_m+r_{m'}}|Y_{m,m'}|^2\nonumber\\
&=4\ex{Y^2}_\sigma-\sum_{m,m'}\frac{8r_{m}r_{m'}}{r_m+r_{m'}}|Y_{m,m'}|^2\nonumber\\
&=4\ex{Y^2}_\sigma-4\sum_{m}r_m|Y_{m,m}|^2-\sum_{m\ne m'}\frac{8r_{m}r_{m'}}{r_m+r_{m'}}|Y_{m,m'}|^2\nonumber\\
&=4\sum_mr_mV_{\phi_m}(Y)-\sum_{m\ne m'}\frac{8r_{m}r_{m'}}{r_m+r_{m'}}|Y_{m,m'}|^2\nonumber\\
&=4\sum_mr_mV_{\phi_m}(Y)-\sum_{(m,m'):m\ne m'\land r_m\ne0\land r_{m'}\ne0}\frac{8r_{m}r_{m'}}{r_m+r_{m'}}|Y_{m,m'}|^2
}
where $Y_{m,m'}:=\bra{\phi_m}Y\ket{\phi_{m'}}$, and $\{r_m\}$ and $\{\ket{\phi_m}\}$ are the eigenvalues and eigenbasis of $\sigma$. $\sigma=\sum_mr_m\ket{\phi_m}\bra{\phi_m}$.
From the above relation, we obtain
\eq{
\sum_{m\ne m'}\frac{8r_{m}r_{m'}}{r_m+r_{m'}}|Y_{m,m'}|^2=4\sum_mr_mV_{\phi_m}(Y)-\calF_\sigma(Y).
}

Let us derive \eqref{skew_C_S}.
We take the spectral decomposition of $\rho_k$ as $\rho_k=\sum_lq^{(k)}_l\ket{\phi^{(k)}_l}\bra{\phi^{(k)}_l}$.
Now, we refer to the range of $k$ as $K$ and the range of $l$ for $\rho_k$ as $L^{(k)}$
Since we have assumed that the test states $\{\rho_k\}$ are orthogonal to each other, we can take sets of number $L'^{(k)}$ and an orthonormal basis $\{\ket{\psi_{(k,l)}}\}_{k\in K, l\in L'^{(k)}}$ of $A$ such that $\{\ket{\phi^{(k)}_{l}}\}_{k\in K, l\in  L^{(k)}}\subset\{\ket{\psi_{(k,l)}}\}_{k\in K, l\in L'^{(k)}}$ and $L^{(k)}\subset L'^{(k)}$.
Then, since the test states $\{\rho_k\}$ are orthogonal to each other, the orthonormal basis $\{\ket{\psi_{(k,l)}}\}_{k\in K, l\in L'^{(k)}}$ is an eigenbasis of each $\rho_k$.
Therefore, we obtain the following relation for $Y_{(k,l),(k',l')}:=\bra{\psi_{(k,l)}}Y\ket{\psi_{(k',l')}}$:
\eq{
\nonumber&\sum_{k\ne k'}\sum_{l,l'}\frac{8p_kq^{(k)}_{l}p_{k'}q^{(k')}_{l'}}{p_kq^{(k)}_{l}+p_{k'}q^{(k')}_{l'}}|Y_{(k,l),(k',l')}|^2
\\&=\sum_{(k,l)\ne (k',l')}\frac{8p_kq^{(k)}_{l}p_{k'}q^{(k')}_{l'}}{p_kq^{(k)}_{l}+p_{k'}q^{(k')}_{l'}}|Y_{(k,l),(k',l')}|^2
-\sum_{k}\sum_{l\ne l'}p_k\frac{8q^{(k)}_{l}q^{(k')}_{l'}}{q^{(k)}_{l}+q^{(k')}_{l'}}|Y_{(k,l),(k,l')}|^2\nonumber\\
&=4\sum_{k,l}p_kq^{(k)}_lV_{\psi_{(k,l)}}(Y)-\calF_{\sum_kp_k\rho_k}(Y)
-\left(4\sum_{k,l}p_kq^{(k)}_lV_{\psi_{(k,l)}}(Y)-\sum_kp_k\calF_{\rho_k}(Y)\right)\nonumber\\
&=\sum_kp_k\calF_{\rho_k}(Y)-\calF_{\sum_kp_k\rho_k}(Y).
}
Using this relation, we obtain the upper bound of $\calC^2$ in \eqref{skew_C_S} as follows:
\eq{
\calC^2&=\sum_{k\ne k'}p_kp_{k'}\Tr[\rho_k Y \rho_{k'} Y]\nonumber\\
&=\sum_{k\ne k'}\sum_{l,l'}p_kq^{(k)}_lp_{k'}q^{(k')}_{l'}|Y_{(k,l),(k',l')}|^2\nonumber\\
&=\sum_{k\ne k'}\sum_{l,l'}\frac{8p_kq^{(k)}_lp_{k'}q^{(k')}_{l'}}{p_kq^{(k)}_l+p_{k'}q^{(k')}_{l'}}|Y_{(k,l),(k',l')}|^2\frac{p_kq^{(k)}_l+p_{k'}q^{(k')}_{l'}}{8}\nonumber\\
&\le\sum_{k\ne k'}\sum_{l,l'}\frac{8p_kq^{(k)}_lp_{k'}q^{(k')}_{l'}}{p_kq^{(k)}_l+p_{k'}q^{(k')}_{l'}}|Y_{(k,l),(k',l')}|^2\frac{p_kq^{(k)}_l+(1-p_{k})q^{(k')}_{l'}}{8}\nonumber\\
&\le\sum_{k\ne k'}\sum_{l,l'}\frac{8p_kq^{(k)}_lp_{k'}q^{(k')}_{l'}}{p_kq^{(k)}_l+p_{k'}q^{(k')}_{l'}}|Y_{(k,l),(k',l')}|^2\frac{\max\{q^{(k)}_l,q^{(k')}_{l'}\}}{8}\nonumber\\
&\le\frac{1}{8}\sum_{k\ne k'}\sum_{l,l'}\frac{8p_kq^{(k)}_lp_{k'}q^{(k')}_{l'}}{p_kq^{(k)}_l+p_{k'}q^{(k')}_{l'}}|Y_{(k,l),(k',l')}|^2
\nonumber\\
&=\frac{1}{8}\left(\sum_kp_k\calF_{\rho_k}(Y)-\calF_{\sum_kp_k\rho_k}(Y)\right)\nonumber\\
&=\frac{1}{8}\mathrm{C}_\calF.
}
Similarly, we obtain the lower bound of $\calC^2$ in \eqref{skew_C_S} as follows:
\eq{
\calC^2&=\sum_{k\ne k'}\sum_{l,l'}\frac{8p_kq^{(k)}_lp_{k'}q^{(k')}_{l'}}{p_kq^{(k)}_l+p_{k'}q^{(k')}_{l'}}|Y_{(k,l),(k',l')}|^2\frac{p_kq^{(k)}_l+p_{k'}q^{(k')}_{l'}}{8}\nonumber\\
&\ge\frac{1}{4}\left(\min_kp_k\lambda^{\min}_{>0}(\rho_k)\right)\sum_{k\ne k'}\sum_{l,l'}\frac{8p_kq^{(k)}_lp_{k'}q^{(k')}_{l'}}{p_kq^{(k)}_l+p_{k'}q^{(k')}_{l'}}|Y_{(k,l),(k',l')}|^2
\nonumber\\
&=\frac{1}{4}\left(\min_kp_k\lambda^{\min}_{>0}(\rho_k)\right)\left(\sum_kp_k\calF_{\rho_k}(Y)-\calF_{\sum_kp_k\rho_k}(Y)\right)\nonumber\\
&=\frac{1}{4}\left(\min_kp_k\lambda^{\min}_{>0}(\rho_k)\right)\mathrm{C}_\calF.
}
\end{proofof}

\subsection{Invariance of $\calC$, $\Delta$, and $\Delta_\alpha$ with respect to the shift of $X_A$ and $X_{A'}$}
We remark that $\calC$, $\Delta$ and $\Delta_\alpha$ do not change by the shift of the conserved quantities $X_{A}$ and $X_{A'}$.
To see this concretely, we write the definitions of $\calC$, $\Delta$, $\Delta_1$, $\Delta_2$, $\Delta_3$ and $Y$ again:
\eq{
\calC&=\sqrt{\sum_{k,k'}p_kp_{k'}\Tr[(\rho_k-\rho_{k'})_+Y(\rho_k-\rho_{k'})_-Y]},\non
\Delta&=\max_{\rho\in\cup_k\mathrm{supp}(\rho_k)}(\sqrt{\calF_{\rho\otimes\rho_B}(X_A\otimes1_B-U^\dagger X_{A'}\otimes1_{B'}U)})\non
\Delta_1&=\Delta_{X_A}+\Delta_{X_{A'}},\non
\Delta_2&=\Delta_Y+2\sqrt{\|\calE^\dagger(X^2_{A'})-\calE^\dagger(X_{A'})^2\|_\infty}\non
\Delta_3&=\max_{\rho\in\cup_k\mathrm{supp}(\rho_k)}(\sqrt{\calF_{\rho}(Y)}+2\sqrt{\calF_{\rho\otimes\rho_B}(\calE^\dagger(X_{A'})\otimes1_{B'}-U^\dagger X_{A'}\otimes1_{B'}U)}),\non
Y&=X_A-\calE^\dagger(X_{A'}).
}
Now, let us define $\tilde{X}_{A}:=X_{A}+aI_A$ and $\tilde{X}_{A'}:=X_{A'}+bI_{A'}$, where $a$ and $b$ are arbitrary real numbers.
We also define $\tilde{\calC}$, $\tilde{\Delta}$, $\tilde{\Delta}_1$, $\tilde{\Delta}_2$, $\tilde{\Delta}_3$ and $\tilde{Y}$ as $\calC$, $\Delta$, $\Delta_1$, $\Delta_2$, $\Delta_3$ and $Y$ for $\tX_A$ and $\tX_{A'}$.
Then, the following relations hold:
\eq{
\tC=\calC,\enskip\tD=\Delta,\enskip\tD_\alpha=\Delta_\alpha.\label{i-s}
}
Let us show \eqref{i-s}. At first, $\Delta_1=\tD_1$ is easily obtained by their definitions. To show $\calC=\tC$, we note that $\tY=Y+(a-b)I_A$, since $\calE^\dagger$ is unital. 
Since the supports of $(\rho_k-\rho_{k'})_+$ and  $(\rho_k-\rho_{k'})_-$ are orthogonal to each other, and since $[(\rho_k-\rho_{k'})_\pm,I_A]=0$, we obtain $\calC=\tC$.
Next, we show that $\Delta=\tD$.
To show this, we note that
\eq{
\tX_A\otimes1_B-U^\dagger \tX_{A'}\otimes1_{B'}U=X_A\otimes1_B-U^\dagger X_{A'}\otimes1_{B'}U+(a-b)1_{AB}.
}
Due to $\calF_\rho(W+c1)=\calF_\rho(W)$ for an arbitrary state $\rho$, an Hermitian operator $W$, and a real number $c$, we obtain $\Delta=\tD$.

Next, let us show $\Delta_3=\tD_3$. Due to $\tY=Y+(a-b)I_A$, we obtain $\calF_{\rho}(Y)=\calF_{\rho}(\tY)$. 
We also have $\calF_{\rho\otimes\rho_B}(\calE^\dagger(X_{A'})\otimes1_{B'}-U^\dagger X_{A'}\otimes1_{B'}U))=\calF_{\rho\otimes\rho_B}(\calE^\dagger(\tX_{A'})\otimes1_{B'}-U^\dagger \tX_{A'}\otimes1_{B'}U))$ since $\calE^\dagger$ is unital.  Therefore, we obtain $\Delta_3=\tD_3$.
Next, let us show $\Delta_2=\tilde{\Delta}_2$. Due to $\tY=Y+(a-b)I_A$, $\Delta_Y=\Delta_{\tY}$ holds. Therefore, we only have to show 
\eq{
\|\calE^\dagger(X^2_{A'})-\calE^\dagger(X_{A'})^2\|_\infty=\|\calE^\dagger(\tX^2_{A'})-\calE^\dagger(\tX_{A'})^2\|_\infty.\label{S74}
}
To derive \eqref{S74}, we show $\calE^\dagger(X^2_{A'})-\calE^\dagger(X_{A'})^2=\calE^\dagger(\tX^2_{A'})-\calE^\dagger(\tX_{A'})^2$ as follows:
\eq{
\calE^\dagger(\tX^2_{A'})-\calE^\dagger(\tX_{A'})^2
&=\calE^\dagger(X^2_{A'}-2bX_{A'}+b^2I_{A'})-\calE^\dagger(X_{A'}-bI_{A'})^2\non
&=\calE^\dagger(X^2_{A'})-2b\calE^\dagger(X_{A'})
+b^2I_A
-(\calE^\dagger(X_{A'})-bI_{A})^2\non
&=\calE^\dagger(X^2_{A'})-\calE^\dagger(X_{A'})^2.
}
Therefore, we obtain \eqref{S74}, and thus we proved \eqref{i-s}.

\subsection{Relations between irreversibility $\delta$ and other concepts including thermodynamic irreversibility and entanglement fidelity errors}\label{SIsubsecRelation}

In this section, we show that the relationships between the irreversibility measure $\delta$ and other quantities in physics.
As remarked in the main text, the irreversibility measure $\delta$ is related to measures of other various concepts.
To be concrete, the irreversibility $\delta$ recovers or bounds from lower measures of various concepts, including thermodynamic irreversibility, quantum information recovery errors, and the unitary implementations error. 
Because of this property of $\delta$, our main results \eqref{SIQ-Cini} and \eqref{SIQ-Gini} treat many concepts at the same time, and provide a series of inequalities restricting these concepts.

The relation between the entropy production and $\delta$ is shown in the Materials and Methods in the main text.
The measure of classical irreversibility is just a special case of $\delta$, where the test states $\{\rho_k\}$ are pure states orthogonal to each other.
The relation between $\delta$ and measures of quantum irreversibility, i.e. entanglement fidelity-based recovery errors, is also introduced in Appendix \ref{prop_delta}, although we have given no proofs.
Therefore, below we give the proofs.
(Also, in Ref. \cite{ET2023}, it is shown that $\delta$ recovers measurement errors/disturbances and  the OTOC.)

The entanglement fidelity-based recovery errors are well-used measures in quantum error correcting codes. They are also useful to understand the errors in the implementation of unitary gates.
Three of the most commonly used recovery errors for a CPTP map $\calE$ from $A$ to $A'$ are as follows:
\eq{
\ew&:=\min_{\calR_{A'\rightarrow A}} \max_{\rho \mbox{ on } AR}D_F(\calR_{A'\rightarrow A}\circ\calE\otimes\mathrm{id}_{R}(\rho),\rho),\label{S78}\\
\me&:=\min_{\calR_{A'\rightarrow A}}D_F(\calR_{A'\rightarrow A}\circ\calE\otimes\mathrm{id}_{R}(\Psi),\Psi)\label{S79},\\
\epsilon(\psi)&:=\min_{\calR_{A'\rightarrow A}}D_F(\calR_{A'\rightarrow A}\circ\calE\otimes\mathrm{id}_{R}(\psi),\psi),\label{S80}
}
where $R$ is a reference system whose Hilbert space has the same dimension as that of $A$, and $\Psi$ is the maximally entangled state on $AR$, and $\psi$ is an arbitrary pure state on $AR$. Clearly, $\me$ is a special case of $\epsilon(\psi)$.
The irreversibility measure $\delta$ can provide lower bounds for these three errors.

First, for an arbitrary test ensemble $\{p_k,\rho_k\}$, we obtain
\eq{
\delta\le\ew.\label{BB1}
}
Second, for an arbitrary test ensemble $\{p_k,\rho_k\}$ satisfying $\sum_kp_k\rho_k=I_A/d_A$ ($d_A$ is the dimension of $A$), we obtain
\eq{
\delta\le\me.\label{BB2}
} 
Third, for an arbitary pure state $\psi$ on $AR$ and for an arbitrary test ensemble $\{p_k,\rho_k\}$ satisfying $\sum_kp_k\rho_k=\Tr_{R}[\psi]$, we obtain 
\eq{
\delta\le\epsilon(\psi).\label{BB3}
} 

Let us prove \eqref{BB1}--\eqref{BB3}. Since we can easily obtain \eqref{BB2} from \eqref{BB3}, we only prove \eqref{BB1} and \eqref{BB3}.
\begin{proofof}{\eqref{BB1}}
Due to the definition of $\ew$, the following relation holds:
\eq{
D_F(\calR_{A'\rightarrow A}\circ\calE\otimes\mathrm{id}_{R}(\sum _kp_k\rho_k\otimes\ket{k}\bra{k}),\sum_k p_k\rho_k\otimes\ket{k}\bra{k}))\le \max_{\rho \mbox{ on } AR}D_F(\calR_{A'\rightarrow A}\circ\calE\otimes\mathrm{id}_{R}(\rho),\rho).
}
Therefore, we obtain
\eq{
\min_{\calR_{A'\rightarrow A}} D_F(\calR_{A'\rightarrow A}\circ\calE\otimes\mathrm{id}_{R}(\sum_k p_k\rho_k\otimes\ket{k}\bra{k}),\sum_k p_k\rho_k\otimes\ket{k}\bra{k}))\le\ew.
}
Note that
\eq{
\calR_{A'\rightarrow A}\circ\calE\otimes\mathrm{id}_{R}(\sum_k p_k\rho_k\otimes\ket{k}\bra{k})=\sum_kp_k\rho'_k\otimes\ket{k}\bra{k},
}
where $\rho'_k:=\calR_{A'\rightarrow A}\circ\calE(\rho_k)$.
Therefore, if the following inequality holds for arbitrary $\{q_j\}$, $\{\rho_j\}$ and $\{\sigma_j\}$, we obtain \eqref{BB1}:
\begin{align}
D_F(\sum_jq_j\rho_j\otimes\ket{j}\bra{j},\sum_jq_j\sigma_j\otimes\ket{j}\bra{j})^2\ge\sum_jq_jD_F(\rho_j,\sigma_j)^2.\label{keyD}
\end{align}

Let us prove \eqref{keyD}.
\begin{align}
F(\sum_jq_j\rho_j\otimes\ket{j}\bra{j},\sum_jq_j\sigma_j\otimes\ket{j}\bra{j})&=\Tr[\sqrt{\sqrt{\sum_{j'}q_{j'}\sigma_{j'}\otimes\ket{j'}\bra{j'}}\sum_jq_j\rho_j\otimes\ket{j}\bra{j}\sqrt{\sum_{j''}q_{j''}\sigma_{j''}\otimes\ket{j''}\bra{j''}}}]\nonumber\\
&=\Tr[\sqrt{\sum_{j'}\sqrt{q_{j'}}\sqrt{\sigma_{j'}}\otimes\ket{j'}\bra{j'}\sum_jq_j\rho_j\otimes\ket{j}\bra{j}\sum_{j''}\sqrt{q_{j''}}\sqrt{\sigma_{j''}}\otimes\ket{j''}\bra{j''}}]\nonumber\\
&=\Tr[\sqrt{\sum_{j}q^2_{j}\sqrt{\sigma_{j}}\rho_j\sqrt{\sigma_{j}}\otimes\ket{j}\bra{j}}]\nonumber\\
&=\Tr[\sum_{j}q_{j}\sqrt{\sqrt{\sigma_{j}}\rho_j\sqrt{\sigma_{j}}}\otimes\ket{j}\bra{j}]\nonumber\\
&=\sum_jq_jF(\rho_j,\sigma_j).
\end{align}
Therefore, we obtain
\begin{align}
D_F(\sum_jq_j\rho_j\otimes\ket{j}\bra{j},\sum_jq_j\sigma_j\otimes\ket{j}\bra{j})^2&=1-F(\sum_jq_j\rho_j\otimes\ket{j}\bra{j},\sum_jq_j\sigma_j\otimes\ket{j}\bra{j})^2\nonumber\\
&=1-(\sum_jq_jF(\rho_j,\sigma_j))^2\nonumber\\
&\ge1-\sum_jq_jF(\rho_j,\sigma_j)^2\nonumber\\
&=\sum_jq_jD_F(\rho_j,\sigma_j)^2.
\end{align}

\end{proofof}

\begin{proofof}{\eqref{BB3}}
To obtain \eqref{BB3}, we first note that due to the assumption $\Tr_R[\psi]=\sum_kp_k\rho_k$, we can take a partial isometry $\calW$ \cite{Paulsen} from $R$ to $R'$ and a measurement $\mathcal{M}_{R^\prime}$ on $R'$ such that
\begin{align}
\id_A\otimes\calM_{R'}\circ\calW(\psi)=\sum_{k}p_k\rho_{k}\otimes\ket{k}\bra{k}_{R'}.
\end{align}
Thus, due to the monotonicity of $D_F$, we obtain 
\begin{align}
D_F(\psi,\calR\circ\calE\otimes\id_{R}(\psi))&\ge D_F(\id_{A}\otimes\calM_{R'}\circ\calW(\psi),\calR\circ\calE\otimes\calM_{R'}\circ\calW(\psi))\nonumber\\
&=D_F(\sum_{k}p_k\rho_{k}\otimes\ket{k}\bra{k}_{R'},\calR\circ\calE\otimes\id_{R}(\sum_{k}p_k\rho_{k}\otimes\ket{k}\bra{k}_{R'}).
\end{align}
Let us take recovery maps $\calR_Q$ and $\calR_C$ satisfying
\eq{
D_F(\psi,\calR_Q\circ\calE\otimes\id_{R}(\psi))&=\epsilon(\psi)\\ D_F(\sum_{k}p_k\rho_{k}\otimes\ket{k}\bra{k}_{R'},\calR_C\circ\calE\otimes\id_{R}(\sum_{k}p_k\rho_{k}\otimes\ket{k}\bra{k}_{R'}))&=\min_{\calR}D_F(\sum_{k}p_k\rho_{k}\otimes\ket{k}\bra{k}_{R'},\calR\circ\calE\otimes\id_{R}(\sum_{k}p_k\rho_{k}\otimes\ket{k}\bra{k}_{R'})).
}
Then, we obtain
\begin{align}
\epsilon(\psi)&=D_F(\psi,\calR_Q\circ\calE\otimes\id_{R}(\psi))\nonumber\\
&\ge D_F(\sum_{k}p_k\rho_{k}\otimes\ket{k}\bra{k}_{R'},\calR_Q\circ\calE\otimes\id_{R}(\sum_{k}p_k\rho_{k}\otimes\ket{k}\bra{k}_{R'}))\nonumber\\
&\ge D_F(\sum_{k}p_k\rho_{k}\otimes\ket{k}\bra{k}_{R'},\calR_C\circ\calE\otimes\id_{R}(\sum_{k}p_k\rho_{k}\otimes\ket{k}\bra{k}_{R'}))\nonumber\\
&\ge\delta.
\end{align}
Here, in the final line, we use \eqref{keyD}.
\end{proofof}

\subsection{Suppression effect on irreversibility by quantum coherence}\label{SIsubsec_suppression}
In this subsection, we discuss the effect of coherence on irreversibility.
For simplicity, we focus on the case when the test states $\{\rho_k\}$ are orthogonal to each other, but our discussion in this subsection is also valid for the general test states. 
When the test states are orthogonal to each other, the main result \eqref{SIQ-Cini} holds. Due to $\Delta\le\Delta_\alpha$ $(\alpha=1,2,3)$, we obtain
\begin{align}
\frac{\calC}{\sqrt{\calF}+\Delta_\alpha}\le\delta.\label{SIQ-setsu}
\end{align}
As we explain in the main text, we can see two messages from this inequality. First, when there is no coherence in the system $B$ in terms of the conserved charge $X_B$, i.e. when $\calF=0$, it is impossible to implement any channel that realizes $\delta<\calC/\Delta_\alpha$. 
The second message of the inequality \eqref{SIQ-setsu} is that we could implement channels that satisfy $\delta<\calC/\Delta_\alpha$ if $\calF$ is sufficiently large. The purpose of this subsection is to show concrete examples of implementations of such channels. 

\textbf{Example 1: bit-flip operations on a two-level system---}
We start with the case of a two-level system.
We consider a two-level system as $A$, and a conserved charge $X_A:=\ket{0}\bra{0}-\ket{1}\bra{1}$ on it.
We define the bit flip unitary between $\ket{0}$ and $\ket{1}$ as $U_{\mathrm{flip}}:=\ket{0}\bra{1}+\ket{1}\bra{0}$.
We also define a set $\calM_{\mathrm{flip}}$ of CPTP maps from $A$ to $A$ whose element mimics the behavior of the bit flip unitary $U_{\mathrm{flip}}$ for the states $\ket{0}$ and $\ket{1}$ as follows:
\eq{
\Lambda\in\calM_{\mathrm{flip}}\Rightarrow \Lambda(\ket{0}\bra{0})=\ket{1}\bra{1}\enskip\land\enskip\Lambda(\ket{1}\bra{1})=\ket{0}\bra{0}.\label{kitaichimimic}
}
Then, the following relation holds:
\eq{
\Lambda\in\calM_{\mathrm{flip}}\Rightarrow\Lambda^\dagger(X_A)=-X_A.\label{gapume}
}
(\textit{Proof:} To show \eqref{gapume}, we only have to show $\Lambda^\dagger(\ket{0}\bra{0})=\ket{1}\bra{1}$ and $\Lambda^\dagger(\ket{1}\bra{1})=\ket{0}\bra{0}$. Let us show $\Lambda^\dagger(\ket{0}\bra{0})=\ket{1}\bra{1}$. Due to \eqref{kitaichimimic}, we easily obtain
$\bra{1}\Lambda^\dagger(\ket{0}\bra{0})\ket{1}=1$ and $\bra{0}\Lambda^\dagger(\ket{0}\bra{0})\ket{0}=0$. Therefore, because of $\Lambda^\dagger(\ket{0}\bra{0})$ is a positive operator, $\bra{1}\Lambda^\dagger(\ket{0}\bra{0})\ket{0}=\bra{1}\Lambda^\dagger(\ket{0}\bra{0})\ket{0}=0$ must hold.  Hence, we obtain $\Lambda^\dagger(\ket{0}\bra{0})=\ket{1}\bra{1}$. We can show $\Lambda^\dagger(\ket{1}\bra{1})=\ket{0}\bra{0}$ in the same manner. $\blacksquare$)

Due to \eqref{gapume}, for any $\Lambda\in\calM_{\mathrm{flip}}$, the quantity $\calC$ for the test ensemble $\{(1/2,1/2),(\ket{+},\ket{-})\}$ $(\ket{\pm}:=(\ket{0}\pm\ket{1})/\sqrt{2})$ satisfies
\eq{
\calC^2&=\sum_{j,j'=\pm}\frac{1}{4}|\bra{j}X_A-\Lambda^\dagger(X_A)\ket{j'}|^2\non
&=\sum_{j,j'=\pm}\frac{1}{4}|\bra{j}2X_A\ket{j'}|^2\non
&=2
}
Because of the above and $\Delta_1=2\Delta_{X_A}=4$, for an arbitrary implementation $(U,\rho_B,X_B)$ of $\Lambda$ satisfying $[U,X_A+X_B]=0$, the following inequality holds for the test ensemble $\{(1/2,1/2),(\ket{+},\ket{-})\}$ $(\ket{\pm}:=(\ket{0}\pm\ket{1})/\sqrt{2})$:
\eq{
\frac{\sqrt{2}}{\sqrt{\calF}+4}\le\delta.
}
Therefore, when a channel $\Lambda\in\calM_{\mathrm{flip}}$ is implementable with an implementation $(U,\rho_B,X_B)$ satisfying $[\rho_B,X_B]=0$ and $[U,X_A+X_B]=0$, the channel must satisfy
\eq{
\frac{1}{2\sqrt{2}}\le\delta\label{ex-no-coherence0}
}
for the test ensemble $\{(1/2,1/2),(\ket{+},\ket{-})\}$ $(\ket{\pm}:=(\ket{0}\pm\ket{1})/\sqrt{2})$.

Now, we show that when we can use coherence in $B$, we can implement a channel $\Gamma\in\calM_{\mathrm{flip}}$ that breaks the inequality \eqref{ex-no-coherence0}.
To be concrete, we construct an implementation $(V,\sigma_B,X_B)$ satisfying $[V,X_A+X_B]=0$ and implementing $\Gamma\in\calM_{\mathrm{flip}}$ that satisfies $\delta<1/2\sqrt{2}$.
We define $B$ as a $d$-level system where $d\ge5$ and $X_B$ as 
\eq{
X_B:=\sum^{d}_{k=1}k\ket{k}\bra{k},
}
where $\{\ket{k}\}_{k=1,...,d}$ is an orthogonal basis on $B$.
We also define $\sigma_B$ as $\sigma_B:=\ket{\sigma_B}\bra{\sigma_B}$, where $\ket{\sigma_B}$ is 
\eq{
\ket{\sigma_B}:=\frac{1}{\sqrt{d-2}}\sum^{d-1}_{k=2}\ket{k}.
}
We also define $V$ as
\eq{
V:=\ket{0}\bra{0}_A\otimes\ket{1}\bra{1}_B+\sum^{d}_{k=2}\ket{1}\bra{0}_A\otimes\ket{k-1}\bra{k}_B+\sum^{d-1}_{k=1}\ket{0}\bra{1}_A\otimes\ket{k+1}\bra{k}_B+\ket{1}\bra{1}_A\otimes\ket{d}\bra{d}_B.
}
Clearly, $V$ is unitary and satisfies $[V,X_A+X_B]=0$.

Let us show that the tuple $(V,\sigma_B,X_B)$ implements a CPTP map $\Gamma$ satisfying $\Gamma\in\calM_{\mathrm{flip}}$ and $\delta<1/2\sqrt{2}$.
To show $\Gamma\in\calM_{\mathrm{flip}}$, note that
\eq{
V\ket{0}_A\otimes\ket{\sigma_B}_B&=\ket{1}_A\otimes\ket{\sigma^{(-1)}_B}_B,\label{hint0-ex}\\
V\ket{1}_A\otimes\ket{\sigma_B}_B&=\ket{0}_A\otimes\ket{\sigma^{(+1)}_B}_B\label{hint1-ex}
}
hold, where $\ket{\sigma^{(+1)}_B}$ and $\ket{\sigma^{(-1)}_B}$ are shifted states of $\ket{\sigma_B}$ defined as:
\eq{
\ket{\sigma^{(+1)}_B}:=\frac{1}{\sqrt{d-2}}\sum^{d}_{k=3}\ket{k},\\
\ket{\sigma^{(-1)}_B}:=\frac{1}{\sqrt{d-2}}\sum^{d-2}_{k=1}\ket{k}.
}
Therefore, we obtain $\Gamma(\ket{0}\bra{0})=\ket{1}\bra{1}$ and $\Gamma(\ket{1}\bra{1})=\ket{0}\bra{0}$, i.e. $\Gamma\in\calM_{\mathrm{flip}}$.

Next, let us show that $\Gamma$ satisfies $\delta<1/2\sqrt{2}$.
To show it, we only have to show a CPTP map $\calR$ satisfying
\eq{
D_F(\ket{+}\bra{+},\calR\circ\Gamma(\ket{+}\bra{+}))&<\frac{1}{2\sqrt{2}},\label{hint2-ex}\\
D_F(\ket{-}\bra{-},\calR\circ\Gamma(\ket{-}\bra{-}))&<\frac{1}{2\sqrt{2}}.\label{hint3-ex}
}
Below, we show that $\calR$ can be chosen as the identity map $\mathrm{id}$ to satisfy \eqref{hint2-ex} and \eqref{hint3-ex}.
Due to \eqref{hint0-ex} and \eqref{hint1-ex}, we obtain
\eq{
V(\ket{+}_A\bra{+}_A\otimes\sigma_B)V^\dagger=&\frac{1}{2}(\ket{0}\bra{0}_A\otimes\ket{\sigma^{(+1)}_B}\bra{\sigma^{(+1)}_B}+\ket{0}\bra{1}_A\otimes\ket{\sigma^{(+1)}_B}\bra{\sigma^{(-1)}_B}\non
&+\ket{1}\bra{0}_A\otimes\ket{\sigma^{(-1)}_B}\bra{\sigma^{(+1)}_B}+\ket{1}\bra{1}_A\otimes\ket{\sigma^{(-1)}_B}\bra{\sigma^{(-1)}_B}).
}
Due to $\Gamma(...)=\Tr_B[V(...\otimes\sigma_B)V^\dagger]$, we obtain
\eq{
\Gamma(\ket{+}_A\bra{+}_A)&=\frac{1}{2}\left(
\ket{0}\bra{0}_A+\ket{0}\bra{1}_A\Tr[\ket{\sigma^{(+1)}_B}\bra{\sigma^{(-1)}_B}]
+\ket{1}\bra{0}_A\Tr[\ket{\sigma^{(-1)}_B}\bra{\sigma^{(+1)}_B}]+\ket{1}\bra{1}_A
\right)\non
&=\frac{1}{2}\left(\ket{0}\bra{0}_A+\frac{d-4}{d-2}\ket{0}\bra{1}_A
+\frac{d-4}{d-2}\ket{1}\bra{0}_A+\ket{1}\bra{1}_A\right).
}
From the above, we obtain
\eq{
F(\ket{+}_A\bra{+}_A,\Gamma(\ket{+}_A\bra{+}_A))^2=\bra{+}\Gamma(\ket{+}_A\bra{+}_A)\ket{+}=\frac{d-3}{d-2}.
}
Due to $D_F(\rho_1,\rho_2)=\sqrt{1-F(\rho_1,\rho_2)^2}$, we obtain
\eq{
D_F(\ket{+}\bra{+},\Gamma(\ket{+}\bra{+}))=\frac{1}{d-2}.
}
Therefore, when $d\ge5$, 
\eq{
D_F(\ket{+}\bra{+},\Gamma(\ket{+}\bra{+}))<\frac{1}{2\sqrt{2}}.
}
We can show $D_F(\ket{-}\bra{-},\Gamma(\ket{-}\bra{-}))<\frac{1}{2\sqrt{2}}$ in the same manner. Therefore, $\Gamma$ satisfies $\delta<1/2\sqrt{2}$.

\textbf{Example 2: permutations between eigenvectors of $X_A$ on a $d_A$-level system---}
The above is an easy example, but its essence can be extended to more general situations.
To show it, we extend the above example to the case where $A$ is a $d_A$-level system.
For simplicity, we consider $X_A$ as an equidistant spectrum charge, i.e.
\eq{
X_A:=\sum_{k}k\ket{k}\bra{k}_A,
}
where $\{\ket{k}_A\}^{d_A}_{k=1}$ is an orthonormal basis on $A$. We remark that the following example is applicable to $X_A$ with a general spectrum by using techniques in Appendix C of Ref. \cite{TSS2}.

We also define a permuation unitary between the eigenvectors $\{\ket{k}\}^{d_A}_{k=1}$ as $U_{\mathrm{per}}:=\sum_k\ket{f(k)}\bra{k}$, where $f$ is a permutation function on $\{k\}^{d_A}_{k=1}$.
We also define a set $\calM_{\mathrm{per}}$ of CPTP maps from $A$ to $A$ whose element mimics the behavior of the permutation unitary $U_{\mathrm{per}}$ for the states $\{\ket{k}\}^{d_A}_{k=1}$ as follows:
\eq{
\Lambda\in\calM_{\mathrm{per}}\Rightarrow \Lambda(\ket{k}\bra{k})=\ket{f(k)}\bra{f(k)},\enskip\forall k.\label{kitaichimimic2}
}
Then, we obtain 
\eq{
\Lambda\in\calM_{\mathrm{per}}\Rightarrow\Lambda^\dagger(X_A)=\sum_k f(k)\ket{k}\bra{k}.\label{gapume2}
}
(\textit{Proof:} To show \eqref{gapume2}, we only have to show $\Lambda^\dagger(\ket{k}\bra{k})=\ket{f^{-1}(k)}\bra{f^{-1}(k)}$, where $f^{-1}$ is the inverse function of $f$.  Let us show $\Lambda^\dagger(\ket{k}\bra{k})=\ket{f^{-1}(k)}\bra{f^{-1}(k)}$. Due to \eqref{kitaichimimic2}, we easily obtain
$\bra{k'}\Lambda^\dagger(\ket{k}\bra{k})\ket{k'}=\delta_{k',f^{-1}(k)}$, where $\delta_{i,j}$ is the Kronecker delta.  And for arbitrary $k'$ and $k''$, $P_{k',k''}\Lambda^\dagger(\ket{k}\bra{k})P_{k',k''}$ is a non-negative operator, where $P_{k',k''}:=\ket{k'}\bra{k'}+\ket{k''}\bra{k''}$. Therefore, $\bra{k'}\Lambda^\dagger(\ket{k}\bra{k})\ket{k''}=\bra{k''}\Lambda^\dagger(\ket{k}\bra{k})\ket{k'}=0$ must hold.  Hence, we obtain $\Lambda^\dagger(\ket{k}\bra{k})=\ket{f^{-1}(k)}\bra{f^{-1}(k)}$. $\blacksquare$)

Let us take a test ensemble $\{(1/2,1/2),(\ket{+},\ket{-})\}$, where $(\ket{\pm}:=(\ket{k_0}\pm\ket{k_1})/\sqrt{2})$ and $(k_0,k_1):=\mathrm{argmax}_{(k,k')}|(k-f(k))-(k'-f(k'))|$.
Then, due to \eqref{gapume2}, the quantity $\calC$ for $\Lambda$ and the test ensemble satisfies
\eq{
\calC^2&=\sum_{j,j'=\pm}\frac{1}{4}|\bra{j}X_A-\Lambda^\dagger(X_A)\ket{j'}|^2\non
&=\frac{\max_{k,k'}|(k-f(k))-(k'-f(k'))|^2}{8}\non
&=\frac{\Delta^2_{Y_A}}{8}.
}
Because of the above and $\Delta_2=\Delta_{Y_A}+2\sqrt{\|\Lambda^\dagger(X^2_{A})-\Lambda^\dagger(X_{A})^2\|_\infty}=\Delta_{Y_A}$, for an arbitrary implementation $(U,\rho_B,X_B)$ of $\Lambda$ satisfying $[U,X_A+X_B]=0$, the following inequality holds for the test ensemble $\{(1/2,1/2),(\ket{+},\ket{-})\}$:
\eq{
\frac{\Delta_{Y_A}}{2\sqrt{2}(\sqrt{\calF}+\Delta_{Y_A})}\le\delta.
}
Therefore, when a channel $\Lambda\in\calM_{\mathrm{flip}}$ is implementable with an implementation $(U,\rho_B,X_B)$ satisfying $[\rho_B,X_B]=0$ and $[U,X_A+X_B]=0$, the channel must satisfy
\eq{
\frac{1}{2\sqrt{2}}\le\delta\label{ex-no-coherence}
}
for the test ensemble $\{(1/2,1/2),(\ket{+},\ket{-})\}$.

Now, we show that when we can use coherence in $B$, we can implement a channel $\Gamma\in\calM_{\mathrm{flip}}$ that breaks the inequality \eqref{ex-no-coherence}.
To be concrete, we construct an implementation $(V,\sigma_B,X_B)$ satisfying $[V,X_A+X_B]=0$ and implementing $\Gamma\in\calM_{\mathrm{per}}$ that satisfies $\delta<1/2\sqrt{2}$.
We define $B$ as a $d$-level system where $d\ge 2d_A+1$ and $X_B$ as 
\eq{
X_B:=\sum^{d}_{k=1}k\ket{k}\bra{k},
}
where $\{\ket{k}\}^d_{k=1}$ is an orthogonal basis on $B$.
We also define $\sigma_B$ as $\sigma_B:=\ket{\sigma_B}\bra{\sigma_B}$, where $\ket{\sigma_B}$ is 
\eq{
\ket{\sigma_B}:=\frac{1}{\sqrt{d-2(d_A-1)}}\sum^{d-d_A+1}_{k=d_A}\ket{k}.
}
We also define $V$ as
\eq{
V:=\sum_{(k,k')\in L}\ket{f(k)}\bra{k}_A\otimes\ket{k+k'-f(k)}\bra{k'}_B+\sum_{(k,k')\ne\in L}\ket{k}\bra{k}_A\otimes\ket{k'}\bra{k'}_B.
}
where $L:=\{(k,k')|d_A\le k+k'\le d-d_A+1\}$. Then, clearly $V$ is unitary and satisfies $[V,X_A+X_B]=0$.

Let us show that the tuple $(V,\sigma_B,X_B)$ implements a CPTP map $\Gamma$ satisfying $\Gamma\in\calM_{\mathrm{per}}$ and $\delta<1/2\sqrt{2}$.
To show $\Gamma\in\calM_{\mathrm{per}}$, note that
\eq{
V\ket{k}_A\otimes\ket{\sigma_B}_B&=\ket{f(k)}_A\otimes\ket{\sigma^{(+k-f(k))}_B}_B,\label{hint1-dA-ex}
}
hold, where $\ket{\sigma^{(+x)}_B}$ is a shifted state of $\ket{\sigma_B}$ defined as:
\eq{
\ket{\sigma^{(+x)}_B}:=\frac{1}{\sqrt{d-2(d_A-1)}}\sum^{d-d_A+1+x}_{k=d_A+x}\ket{k}.
}
Therefore, we obtain $\Gamma(\ket{k}\bra{k})=\ket{f(k)}\bra{f(k)}$, i.e. $\Gamma\in\calM_{\mathrm{per}}$.

Next, let us show that $\Gamma$ satisfies $\delta<1/2\sqrt{2}$.
To show it, we only have to show a CPTP map $\calR$ satisfying
\eq{
D_F(\ket{+}\bra{+},\calR\circ\Gamma(\ket{+}\bra{+}))&<\frac{1}{2\sqrt{2}},\label{hint2-dA-ex}\\
D_F(\ket{-}\bra{-},\calR\circ\Gamma(\ket{-}\bra{-}))&<\frac{1}{2\sqrt{2}}.\label{hint3-dA-ex}
}
Below, we show that $\calR(...):=U^\dagger_{\mathrm{per}}...U_{\mathrm{per}}$ satisfies \eqref{hint2-dA-ex} and \eqref{hint3-dA-ex}.
Due to \eqref{hint1-dA-ex}, we obtain
\eq{
U^\dagger_{\mathrm{per}}V(\ket{+}_A\bra{+}_A\otimes\sigma_B)V^\dagger U_{\mathrm{per}}=&\frac{1}{2}(\ket{k_{0}}\bra{k_{0}}_A\otimes\ket{\sigma^{(+k_{0}-f(k_0))}_B}\bra{\sigma^{(+k_{0}-f(k_0))}_B}+\ket{k_0}\bra{k_1}_A\otimes\ket{\sigma^{(+k_{0}-f(k_0))}_B}\bra{\sigma^{(+k_{1}-f(k_1))}_B}\non
&+\ket{k_1}\bra{k_0}_A\otimes\ket{\sigma^{(+k_{1}-f(k_1))}_B}\bra{\sigma^{(+k_{0}-f(k_0))}_B}+\ket{k_1}\bra{k_1}_A\otimes\ket{\sigma^{(+k_{1}-f(k_1))}_B}\bra{\sigma^{(+k_{1}-f(k_1))}_B}).
}
Due to $\Gamma(...)=\Tr_B[V(...\otimes\sigma_B)V^\dagger]$, we obtain
\eq{
U^\dagger_{\mathrm{per}}\Gamma(\ket{+}_A\bra{+}_A)U_{\mathrm{per}}&=\frac{1}{2}\left(
\ket{k_0}\bra{k_0}_A+\ket{k_0}\bra{k_1}_A\Tr[\ket{\sigma^{(+k_0-f(k_0))}_B}\bra{\sigma^{(+k_1-f(k_1))}_B}]\right.\non
&\enskip\enskip\left.+\ket{k_1}\bra{k_0}_A\Tr[\ket{\sigma^{(+k_1-f(k_1))}_B}\bra{\sigma^{(+k_0-f(k_0))}_B}]+\ket{k_1}\bra{k_1}_A
\right)\non
&=\frac{1}{2}\left(\ket{k_0}\bra{k_0}_A+\frac{d-2(d_A-1)-\Delta_{Y_A}}{d-2(d_A-1)}\ket{k_0}\bra{k_1}_A
+\frac{d-2(d_A-1)-\Delta_{Y_A}}{d-2(d_A-1)}\ket{k_1}\bra{k_0}_A+\ket{k_1}\bra{k_1}_A\right).
}
From the above, we obtain
\eq{
F(\ket{+}_A\bra{+}_A,\calR\circ\Gamma(\ket{+}_A\bra{+}_A))^2=\bra{+}U^\dagger_{\mathrm{per}}\Gamma(\ket{+}_A\bra{+}_A)U_{\mathrm{per}}\ket{+}=\frac{d-2(d_A-1)-\frac{\Delta_{Y_A}}{2}}{d-2(d_A-1)}.
}
Due to $D_F(\rho_1,\rho_2)=\sqrt{1-F(\rho_1,\rho_2)^2}$, we obtain
\eq{
D_F(\ket{+}\bra{+},\Gamma(\ket{+}\bra{+}))=\frac{\Delta_{Y_A}}{2(d-2(d_A-1))}.
}
Therefore, when $d>\sqrt{2}\Delta_{Y_A}+2(d_A-1)$, 
\eq{
D_F(\ket{+}\bra{+},\Gamma(\ket{+}\bra{+}))<\frac{1}{2\sqrt{2}}.
}
We can show $D_F(\ket{-}\bra{-},\Gamma(\ket{-}\bra{-}))<\frac{1}{2\sqrt{2}}$ in the same manner. Therefore, $\Gamma$ satisfies $\delta<1/2\sqrt{2}$.

\textbf{Example 3: approximate implementations of unitary gates---}
Other than the above two examples, there are various examples of the suppression effect on irreversibility by quantum coherence, i.e. implementations of channels satisfying $\delta<\calC/\Delta_\alpha$.
One such example is the approximate implementation of unitary channels.

In Ref. \cite{TSS2}, it is shown that for an arbitrary $U_A$ on a quantum system $A$ and arbitrary real number $\epsilon>0$, we can construct an implementation $(U,\rho_B,X_B)$ satisfying $[U,X_A+X_B]=0$ and $\calF_{\rho_B}(X_B)\gg 1$ to realize a CPTP-map $\Lambda$ which satisfies
\eq{
D_F(\Lambda(\rho_A),U_A\rho_AU^\dagger_A)<\epsilon.
}
Therefore, for an arbitrary test ensemble $\{p_k,\rho_k\}$, the realized channel $\Lambda$ satisfies
\eq{
\delta<\epsilon.\label{TSS2coro}
}
Using \eqref{TSS2coro}, we show that when $[U_A,X_A]\ne0$, the channel $\Lambda$ satisfies $\delta<\calC/\Delta_1$.
To show that, we only have to show that $\Lambda$ satisfies $\calC>0$ for a proper test ensemble.
Let us take a test ensemble $\{(1/2,1/2),(\ket{+},\ket{-})\}$, where $(\ket{\pm}:=(\ket{\psi_0}\pm\ket{\psi_1})/\sqrt{2})$ and $\ket{\psi_0}$ and $\ket{\psi_1}$ are the eigenvectors of $X_A-U^\dagger_AX_AU_A$ corresponding to the maximum and minimum eigenvalues.
Then, $\calC$ for the test ensemble satisfies 
\eq{
\calC\ge\frac{\Delta_{X_A-U^\dagger_AX_AU_A}-2\epsilon\Delta_{X_A}}{2\sqrt{2}}.\label{CandDeltaforuni}
}
(\textit{Proof of \eqref{CandDeltaforuni}:} 
Due to $\calC=\frac{|\bra{+}Y_A\ket{-}|}{\sqrt{2}}$ and $|\bra{+}X_A-U^\dagger_AX_AU_A\ket{-}|=\Delta_{X_A-U^\dagger_AX_AU_A}/2$, we only have to show
\eq{
||\bra{+}Y_A\ket{-}|-|\bra{+}X_A-U^\dagger_AX_AU_A\ket{-}||\le2\epsilon\Delta_{X_A}.\label{CandDeltaforuni2}
}
We show show \eqref{CandDeltaforuni2} as follows:
\eq{
||\bra{+}Y_A\ket{-}|-|\bra{+}X_A-U^\dagger_AX_AU_A\ket{-}||&\le|\bra{+}Y_A\ket{-}-\bra{+}X_A-U^\dagger_AX_AU_A\ket{-}|\non
&\stackrel{(a)}{\le}\Delta_{\Lambda^\dagger(X_A)-U^\dagger_AX_AU_A}\non
&=2\max_{\ket{\psi}}|\bra{\psi}\Lambda^\dagger(X_A)-U^\dagger_AX_AU_A\ket{\psi}|\non
&\le2\|X_A-\beta I\|_\infty\|\Lambda(\psi)-U^\dagger_A\psi U_A\|_1\non
&\le2\Delta_{X_A}\epsilon.
}
Here we use the fact that $|\bra{+}Z\ket{-}|=|\bra{+}Z+\alpha I\ket{-}|\le\sqrt{|\bra{+}\sqrt{Z+\alpha I}\ket{+}|}\sqrt{|\bra{-}\sqrt{Z+\alpha I}\ket{-}|}\le\Delta_{Z}$ is valid for an arbitrary Hermitian operator $Z$ in (a), where $\alpha$ is a real number such that the minimum eigenvalue of $Z+\alpha I$ is zero. The number $\beta$ is a real number such that $\|X_A-\beta I\|_\infty=\Delta_{X_A}/2$.
$\blacksquare$)

Due to \eqref{CandDeltaforuni}, by taking a sufficiently small $\epsilon$, we obtain 
\eq{
\frac{\calC}{\Delta_{1}}\ge\frac{\Delta_{X_A-U^\dagger_AX_AU_A}-2\epsilon\Delta_{X_A}}{2\sqrt{2}\Delta_{1}}\ge\frac{\Delta_{X_A-U^\dagger_AX_AU_A}}{4\Delta_1}.
}
Because of \eqref{TSS2coro}, there is a proper implementation  $(U,\rho_B,X_B)$ for $\Lambda$ that satisifies
\eq{
\delta<\frac{\Delta_{X_A-U^\dagger_AX_AU_A}}{4\Delta_{1}}
}
for the test ensemble $\{(1/2,1/2),(\ket{+},\ket{-})\}$. Therefore, the implementation $(U,\rho_B,X_B)$ satisfies $\delta<\calC/\Delta_1$.

This example is intimately related to the application of the main result \eqref{SIQ-Cini} to the unitary gate implementation. In fact, the above implementation 
 $(U,\rho_B,X_B)$ that satisfies $\delta<\calC/\Delta_1$ does not break \eqref{SIQ-Cini}, because the coherence $\calF=\calF_{\rho_B}(X_B)$ of the implementation is also very large.
In other words, the restriction 
on the unitary-gate implementation provided by \eqref{SIQ-Cini} covers the above implementation.
We discuss the restriction in Section \ref{SIsecQIP}.

\subsection{Formulation and results for the case of general Lie groups}\label{subsecLie-extenxtion}
In this article, we mainly treat the case of the simplest continuous symmetry, i.e. the case where there is a single conserved charge.
However, we can easily extend the setup and the main result to the case of a general Lie group symmetry. 
In this subsection, we describe the extension.

\begin{figure}[tb]
		\centering
		\includegraphics[width=.45\textwidth]{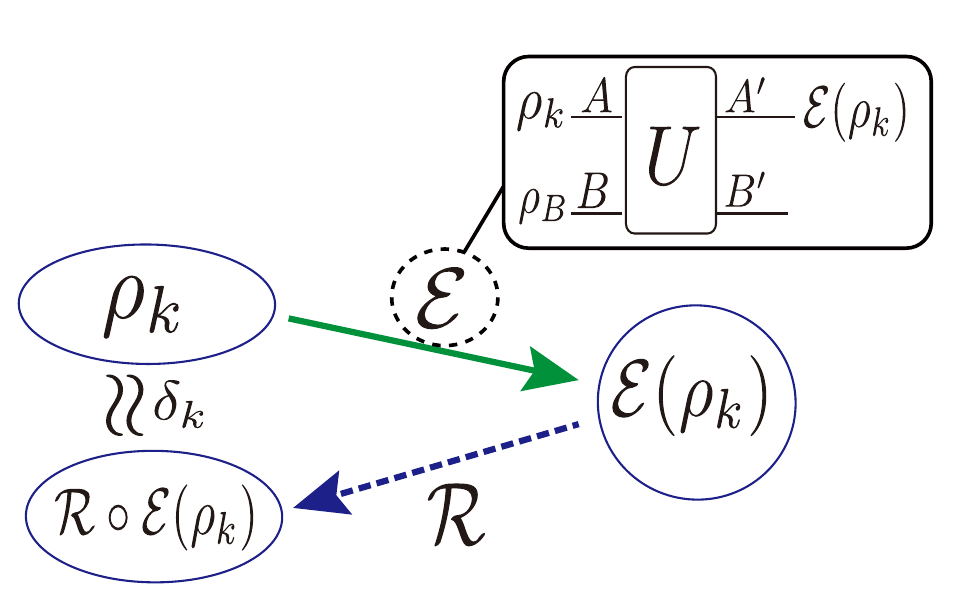}
		\caption{Schematic diagram of the framework. The figure is the same as the figure of setup in the main text. We prepare the test states $\{\rho_k\}$ with probability $\{p_k\}$ and perform a CPTP map $\calE$ caused by a unitary interaction $U$. We try to recover the test states with a recovery CPTP map $\calR$ independent of $k$, and define the irreversibility of $\calE$ for the test ensemble $\{p_k,\rho_k\}$ as the average of recovery error for the optimal recovery map: $\delta:=\sqrt{\sum_kp_k\delta^2_k}$. We investigate the restriction on the irreversibility under the assumption that $U$ satisfies the global symmetry \eqref{multi-X-con-0}.}
		\label{setupagain}
\end{figure}

We first introduce the setup for the general Lie group symmetry. 
We consider the almost same setup as in the main text (Figure \ref{setupagain}). 
We consider two systems, $A$ and $B$. We consider $A$ as the system of interest, and $B$ as another quantum system that works as an environment whose initial state is fixed to a quantum state $\rho_B$.
We perform a unitary operation $U$ on $AB$ and divide $AB$ into two systems, $A'$ and $B'$.
Then, the quantum process from $A$ to $A'$ is described as a completely positive trace preserving (CPTP) map $\calE(...):=\Tr_{B'}[U...\otimes\rho_BU^\dagger]$. 
We assume that $U$ has a global symmetry described by a Lie group $G$. 
We consider a unitary representation (or a projective unitary representation) $\{U_{g,A}\}_{g\in G}$, $\{U_{g,B}\}_{g\in G}$, $\{U_{g,A'}\}_{g\in G}$, $\{U_{g,B'}\}_{g\in G}$ on the systems $A$, $B$, $A'$ and $B'$, and assume that the dynamics $U$ satisfies
\begin{align}
U^{\dagger}(U_{g,A'}\otimes U_{g,B'})U=U_{g,A}\otimes U_{g,B}, \enskip \forall g\in G.\label{multi-X-con-0}
\end{align}
We can easily see that the above assumption reduces to conservation laws of conserved quantities.
In fact, when $g$ is near the origin of $G$, there are a proper parametrization $t(g):=(t^{(1)}(g),...,t^{(m)}(g))$ and a set of Hermitian operators $\vect{X}_\alpha:=(X^{(1)}_\alpha,...,X^{(m)}_\alpha)$  ($\alpha=A,B,A',B'$) such that $U_{g,\alpha}$ can be written as 
\eq{
U_{g,\alpha}=e^{i\sum_a t^{(a)}(g)X^{(a)}_{\alpha}}.\label{multi-X-con-0.5}
}
Using \eqref{multi-X-con-0.5}, \eqref{multi-X-con-0} reduces to
\begin{align}
U^\dagger (X^{(a)}_{A'}+X^{(a)}_{B'})U=X^{(a)}_{A}+X^{(a)}_{B},\label{multi-X-con-1}
\end{align}
where $X^{(a)}_{\alpha}$ is the local operator of the conserved quantity on the system $\alpha$ ($\alpha=A,B,A',B'$).

Next, we introduce several quantities to describe the extension of the main result to the above setup.
The first quantity is the quantum SLD-Fisher information matrix for the state family $\{e^{i\sum_a t^{(a)}(g)X^{(a)}_{\alpha}}\rho e^{-i\sum_a t^{(a)}(g)X^{(a)}_{\alpha}}\}$:
\eq{
(\widehat{\calF}_{\rho}(\vect{X}))_{ab}:=\sum_{i,j} \frac{2(p_i - p_j)^2}{p_i+p_j}  \bra{\psi_i} X^{(a)} \ket{\psi_j} \bra{\psi_j} X^{(b)} \ket{\psi_i}.
}
Here $\rho=\sum_j p_j \dm{\psi_j}$ is a spectral decomposition of a state $\rho$, and $(\widehat{\calF}_{\rho}(\vect{X}))_{ab}$ is the $(a,b)$-component of the Fisher matrix $\widehat{\calF}_{\rho}(\vect{X})$. 
The quantum SLD-Fisher information matrix is a resource measure of the resource theory of asymmetry for a general connected Lie group \cite{Kudo_Tajima}. It is also conjectured that the ratio of the quantum SLD-Fisher information matrix determines the iid conversion rate in the resource theory of asymmetry for a general connected Lie group \cite{Shitara_Tajima}.
Hereafter, we use the abbreviation $\widehat{\calF}$ of $\widehat{\calF}_{\rho_B}(\vect{X}_B)$.

The second quantity is an extension of $\calC$:
\eq{
(\widehat{\calC})_{ab}&:=\sum_{k\ne k'}p_kp_{k'}\Tr[(\rho_{k}-\rho_{k'})_{+}Y^{(a)}(\rho_{k}-\rho_{k'})_{-}Y^{(b)}],\\
Y^{(a)}&:=X^{(a)}_A-\calE^\dagger(X^{(a)}_{A'})
}
Again $(\widehat{\calC})_{ab}$ is the $(a,b)$-component of the matrix $\widehat{\calC}$. 

The third quantity is $\widehat{\Delta_1}$, an extension of $\Delta_1$:
\eq{
(\widehat{\Delta_1})_{ab}&:=\delta_{a,b}\sum_{a}(\Delta^{(a)}_1)^2
}
where $\Delta^{(a)}_1:=(\Delta_{X^{(a)}_A}+\Delta_{X^{(a)}_{A'}})/2$.

For the above formulation, the following inequality holds for an arbitrary test ensemble:
\eq{
\widehat{\calC}\le2\left(\widehat{\calF}+\widehat{\Delta_1}\right)\delta_{\mathrm{multi}}\label{Lie-SIQ1}
}
where $\delta_{\mathrm{multi}}$ is the irreversibility defined for the setting with multiple conserved quantities. 
For a test ensemble whose test states are orthogonal to each other, the following inequality holds:
\eq{
\widehat{\calC}\le2\left(1-\min_k p_k\right)(\widehat{\calF}+\widehat{\Delta_1})\delta^2.\label{Lie-SIQ2}
}

\begin{proofof}{\eqref{Lie-SIQ1} and \eqref{Lie-SIQ2}}
First, we show \eqref{Lie-SIQ1}.
Since \eqref{multi-X-con-1} holds for arbitrary $a=1,...,m$, the following relation holds for an arbitrary real number vector $\vect{\lambda}:=(\lambda_1,...,\lambda_m)$:
\begin{align}
U^\dagger (X^{\vect{\lambda}}_{A'}+X^{\vect{\lambda}}_{B'})U=X^{\vect{\lambda}}_{A}+X^{\vect{\lambda}}_{B},\label{multi-X-con-vec}
\end{align}
where $X^{\vect{\lambda}}_{\alpha}:=\sum_{a}\lambda_a X^{(a)}_{\alpha}$ is the local operator of the conserved quantity on the system $\alpha$ ($\alpha=A,B,A',B'$).

Due to \eqref{multi-X-con-vec}, the following inequality holds:
\eq{
(\calC^{\vect{\lambda}})^2&\le \left(\sqrt{\calF^{\vect{\lambda}}}+\Delta^{\vect{\lambda}}_1\right)^2\delta_{\mathrm{multi}}\nonumber\\
&\le2\left(\calF^{\vect{\lambda}}+(\Delta^{\vect{\lambda}}_1)^2\right)\delta_{\mathrm{multi}}\label{Lie-1}
}
Here $\calC^{\vect{\lambda}}$, $\calF^{\vect{\lambda}}$ and $\Delta^{\vect{\lambda}}_1$ are $\calC$, $\calF$ and $\Delta_1$ whose $X_\alpha$ is $X^{\vect{\lambda}}_\alpha$. 

We note that \eqref{Lie-1} holds even if we substitute the following $X'^{\vect{\lambda}}_\alpha$ for $X^{\vect{\lambda}}_\alpha$ since $X'^{\vect{\lambda}}_\alpha$ also satisfies \eqref{multi-X-con-1}:
\eq{
X'^{\vect{\lambda}}_\alpha:=\sum_a \lambda_a (X^{(a)}_{\alpha}-\beta_{a,\alpha}I_{\alpha})
}
Here $\beta_{a,\alpha}$ is a real number satisfying $\|X^{(a)}_{A}-\beta_{a,A}I_{A}\|=\Delta_{X^{(a)}_{A}}/2$, $\|X^{(a)}_{A'}-\beta_{a,A'}I_{A'}\|=\Delta_{X^{(a)}_{A'}}/2$, $\beta_{a,A}=-\beta_{a,B}$ and $\beta_{a,A'}=-\beta_{a,B'}$.
Since $\calC$ or $\calF$ does not change by the shift of the origin of $X$, we obtain
\eq{
(\calC^{\vect{\lambda}})^2\le2\left(\calF^{\vect{\lambda}}+(\|X'^{\vect{\lambda}}_{A}\|+\|X'^{\vect{\lambda}}_{A'}\|)^2\right)\delta_{\mathrm{multi}}\label{Lie-1'}
}
Because of 
\eq{
\|X'^{\vect{\lambda}}_{\alpha}\|&=\|\sum_a\lambda_aX'^{(a)}_{\alpha}\|\nonumber\\
&\le\sum_a|\lambda_a|\|X'^{(a)}_{\alpha}\|\nonumber\\
&=\sum_a|\lambda_a|\frac{\Delta_{X^{(a)}_{A}}+\Delta_{X^{(a)}_{A'}}}{2}\nonumber\\
&\le\sqrt{\vect{\lambda}^T\vect{\lambda}\sum_a(\Delta^{(a)}_1)^2},
}
We obtain
\eq{
(\calC^{\vect{\lambda}})^2\le2\left(\calF^{\vect{\lambda}}+\vect{\lambda}^T\vect{\lambda}\sum_a(\Delta^{(a)}_1)^2\right)\delta_{\mathrm{multi}}\label{Lie-2}
}
Note that
\eq{
(\calC^{\vect{\lambda}})^2&=\vect{\lambda}^T \widehat{\calC}\vect{\lambda}\\
\calF^{\vect{\lambda}}&=\vect{\lambda}^T \widehat{\calF}\vect{\lambda}\\
\vect{\lambda}^T\vect{\lambda}\sum_a(\Delta^{(a)}_1)^2&=\vect{\lambda}^T \widehat{\Delta_1}\vect{\lambda}.
}
Therefore, for arbitrary real vector $\vect{\lambda}$,
\eq{
\vect{\lambda}^T \widehat{\calC}\vect{\lambda}\le2\vect{\lambda}^T \left(\widehat{\calF}+\widehat{\Delta_1}\right)\vect{\lambda}\delta_{\mathrm{multi}}.
}
Since the matrices $\widehat{\calC}$, $\widehat{\calF}$ and $\widehat{\Delta_1}$ are real symmetric matrices (Note that $\Tr[(\rho_{k'}-\rho_{k})_{+}Y^{(a)}(\rho_{k'}-\rho_{k})_{-}Y^{(b)}]=\Tr[(\rho_{k}-\rho_{k'})_{+}Y^{(a)}(\rho_{k}-\rho_{k'})_{-}Y^{(b)}]^*$), we obtain \eqref{Lie-SIQ1}.

We can prove \eqref{Lie-SIQ2} by substituting $(1-\min_k p_k)\delta^2$ for $\delta_{\mathrm{multi}}$ in the above discussion.
\end{proofof}

\section{ Applications to quantum information processing}\label{SIsecQIP}
In this section, we apply the result \eqref{cost-C} in the main text to quantum information processing. For readers' convenience, we write \eqref{cost-C} here again:
\eq{
\sqrt{\calF^{\mathrm{cost}}_{\calN}}\ge\frac{\calC}{\delta}-\Delta.\label{cost-CS}
}
This inequality holds whenever the test states $\{p_k,\rho_k\}$ satisfies $F(\rho_k,\rho_{k'})=0$ for $k\ne k'$.
Here 
\eq{
\calF^{\mathrm{cost}}_{\calN}:=\min\{\calF_{\rho_B}(X_B)\enskip|\enskip(\rho_B,X_B,X_{B'},U)\mbox{ realizes }\calN,\mbox{and satisfies }U^\dagger (X_{A'}+X_{B'})U=X_A+X_B.\}
}

\subsection{Measurement: a quantitative Wigner-Araki-Yanase theorem for fidelity error}\label{SIsubsec_WAY}
We first apply \eqref{cost-CS} to measurements.
We can derive the following theorem from \eqref{cost-CS}:
\begin{theorem}
Let $\calQ$ and $\calP$ be measurement channels from $A$ to $A'$ defined as
\eq{
\calQ(...):=\sum_{k\in\calK}\Tr[Q_k...]\ket{k}\bra{k},\\
\calP(...):=\sum_{k\in\calK}\Tr[P_k...]\ket{k}\bra{k},
}
where $\{Q_k\}$ and $\{P_k\}$ are PVM (projection valued measure) and POVM (positive operator valued measure) operators on $A$, respectively. 
We assume that each $\ket{k}\bra{k}$ commutes with the conserved quantity $X_{A'}$ on $A'$.  
We remark that in natural settings (e.g. $A'$ is a memory system for classical data), we can assume that $X_{A'}\propto I_{A'}$, and then the assumption $[X_{A'},\ket{k}\bra{k}]=0$ holds automatically. 
We also assume that the measurement channel $\calQ$ is approximated by $\calP$, i.e., the following inequality holds for a real positive number $\epsilon$:
\eq{
D_F(\calP(\rho),\calQ(\rho))\le\epsilon,\enskip\forall\rho\mbox{ on }A.\label{app_meas}
}
Then, the implementation cost of $\calP$ under conservation law of $X$ as follows:
\eq{
\sqrt{\calF^{\mathrm{cost}}_{\calP}}\ge\max_k\frac{\sqrt{2}\|[X_A,Q_k]\|_{\infty}}{\epsilon}-\Delta'.
}
Here $\Delta':=\Delta_{X_A}+2\Delta_{X_{A'}}$.
We remark that when $X_{A'}\propto I_{A'}$ holds, $\Delta'=\Delta_{X_A}$ also holds.
\end{theorem}

\begin{proof}
We first take a value $k$ in $\calK$, and define the following CPTP map from $A'$ to $A'$:
\eq{
\calD_k(..):=\bra{k}...\ket{k}\ket{0}\bra{0}+\sum_{k':k\ne k'}\bra{k'}...\ket{k'}\ket{1}\bra{1},
}
where $\ket{0}$ and $\ket{1}$ are eigenstates of $X_{A'}$.
Using $\calD_k$, we define
\eq{
\calQ'_k(...)&:=\calD_k\circ\calQ(...),\\
\calP'_k(...)&:=\calD_k\circ\calP(...).
}
Clearly, the channel $\calD_k$ is covariant with respect to $X_{A'}$.
Therefore, $\calF^{\mathrm{cost}}_{\calP'_k}\le\calF^{\mathrm{cost}}_{\calP}$, and thus the following inequality holds:
\eq{
\max_k\calF^{\mathrm{cost}}_{\calP'_k}\le\calF^{\mathrm{cost}}_{\calP}.
}
Therefore, we first give a lower bound for $\calF^{\mathrm{cost}}_{\calP'_k}$.
Note that
\eq{
\calQ'_k(...)&:=\Tr[Q_k...]\ket{0}\bra{0}+\Tr[(1-Q_k)...]\ket{1}\bra{1},\\
\calP'_k(...)&:=\Tr[P_k...]\ket{0}\bra{0}+\Tr[(1-P_k)...]\ket{1}\bra{1}.
}
Let us take arbitrary pure states $\ket{\psi_k}$ and $\ket{\psi^{\perp}_k}$ satisfying
\eq{
\bra{\psi_k}Q_k\ket{\psi_k}&=1,\label{maru1A}\\
\bra{\psikn}Q_k\ket{\psikn}&=0.\label{maru1B}
}
Then, the following relation holds
\eq{
\calQ'_k(\psik)&=\ket{0}\bra{0},\\
\calQ'_k(\psikn)&=\ket{1}\bra{1}.
}
Therefore, due to the definition of the fidelity, we obtain
\eq{
F(\calQ'_k(\psik),\calP'_k(\psik))&=\sqrt{\Tr[Q_k\psik]\Tr[P_k\psik]}\non
&=\sqrt{\bra{\psik}P_k\ket{\psik}},\\
F(\calQ'_k(\psikn),\calP'_k(\psikn))&=\sqrt{\Tr[Q_k\psikn]\Tr[P_k\psikn]}\non
&=\sqrt{\bra{\psikn}(1-P_k)\ket{\psikn}}.
}
Due to \eqref{app_meas}, $D_F=\sqrt{1-F^2}$ and the monotonicity of $D_F$, we obtain
\eq{
\bra{\psik}P_k\ket{\psik}\ge1-\epsilon^2,\label{small1}\\
\bra{\psikn}(1-P_k)\ket{\psikn}\ge1-\epsilon^2.\label{small2}
}
Let us define a recovery CPTP map $\calR_k$ as
\eq{
\calR_k(...):=\bra{0}...\ket{0}\psik+\bra{1}...\ket{1}\psikn.
}
Then, we obtain
\eq{
D_F(\calR_k\circ\calP'_k(\psi_k),\psi_k)&=\sqrt{1-\bra{\psi_k}\calR_k\circ\calP'_k(\psi_k)\ket{\psi_k}}\non
&=\sqrt{1-\bra{\psi_k}P_k\ket{\psi_k}}\non
&\le\epsilon.
}
In the same way, we obtain
\eq{
D_F(\calR_k\circ\calP'_k(\psikn),\psikn)\le\epsilon.
}
Therefore, when we take a test ensemble $\{(1/2,\psi_k),(1/2,\psikn)\}$, the irreversibility $\delta$ for them satisfies
\eq{
\delta\le\epsilon.
}
Therefore, for arbitrary $\psi_k$ and $\psikn$ satisfying \eqref{maru1A} and \eqref{maru1B}, 
\eq{
\sqrt{\calF^{\mathrm{cost}}_{\calP'_k}}\ge\frac{\calC_k}{\epsilon}-\Delta,\label{mid_meas}
}
where $\calC_k=\frac{|\bra{\psik}(X_A-\calP'^\dagger_k(X_{A'}))\ket{\psikn}|}{\sqrt{2}}$.

Since \eqref{mid_meas} holds for arbitrary $\psi_k$ and $\psikn$ satisfying \eqref{maru1A} and \eqref{maru1B}, we obtain
\eq{
\sqrt{\calF^{\mathrm{cost}}_{\calP}}\ge\max_k\frac{\max\{\calC_k|\mbox{$\psi_k$ and $\psikn$ satisfying \eqref{maru1A} and \eqref{maru1B}}\}}{\epsilon}-\Delta
}
To evaluate the RHS, we first give a lower bound for $\calC_k$. Since  $|\alpha-\beta|\ge|\alpha|-|\beta|$ holds for arbitrary complex numbers $\alpha$ and $\beta$, we obtain
\eq{
\calC_k\ge\frac{|\bra{\psik}X_A\ket{\psikn}|-|\bra{\psik}\calP'^\dagger_k(X_{A'})\ket{\psikn}|}{\sqrt{2}}.
}

Let us evaluate $|\bra{\psik}\calP'^\dagger_k(X_{A'})\ket{\psikn}|$ in the above.
Due to the definition of $\calP'_k$, 
\eq{
\calP'^\dagger_k(...)=\bra{0}...\ket{0}P_k+\bra{1}...\ket{1}(1-P_k).
}
Clearly, $\calP'^\dagger_k(I_{A'})=I_{A'}$ holds. Therefore, due to $\bra{\psi_k}I_{A}\ket{\psikn}=0$, for an arbitrary real number $x$,
\eq{
|\bra{\psi_k}\calP'^\dagger(X_{A'})\ket{\psikn}|=|\bra{\psi_k}\calP'^\dagger(X_{A'}-xI_{A'})\ket{\psikn}|
} 
Now, let us take $x_*$ such as $\|X_{A'}-x_*I_{A'}\|_{\infty}=\frac{\Delta_{X_{A'}}}{2}$.
Then, we can evaluate $|\bra{\psik}\calP'^\dagger_k(X_{A'})\ket{\psikn}|$ as follows:
\eq{
|\bra{\psik}\calP'^\dagger_k(X_{A'})\ket{\psikn}|&=
|\bra{\psi_k}\calP'^\dagger(X_{A'}-x_*I_{A'})\ket{\psikn}|\non
&\le|\bra{0}(X_{A'}-x_*I_{A'})\ket{0}\bra{\psik}P_k\ket{\psikn}|
+
|\bra{1}(X_{A'}-x_*I_{A'})\ket{1}\bra{\psik}(1-P_k)\ket{\psikn}|\non
&\le
\le|\bra{0}(X_{A'}-x_*I_{A'})\ket{0}||\bra{\psik}P_k\ket{\psikn}|
+
|\bra{1}(X_{A'}-x_*I_{A'})\ket{1}||\bra{\psik}(1-P_k)\ket{\psikn}|\non
&\le\frac{\Delta_{X_{A'}}}{2}(|\bra{\psik}P_k\ket{\psikn}|
+|\bra{\psik}(1-P_k)\ket{\psikn}|)\non
&\stackrel{(a)}{\le}\frac{\Delta_{X_{A'}}}{2}(\sqrt{|\bra{\psik}P_k\ket{\psik}||\bra{\psikn}P_k\ket{\psikn}|}
+\sqrt{|\bra{\psik}(1-P_k)\ket{\psik}||\bra{\psikn}(1-P_k)\ket{\psikn}|})\non
&\stackrel{(b)}{\le}\frac{\Delta_{X_{A'}}}{2}\times2\epsilon\non
&\le\Delta_{X_{A'}}\epsilon.
}
Here we use the Cauchy-Schwartz inequality in (a) and \eqref{small1} and \eqref{small2} in (b).
Therefore, we obtain
\eq{
\calC_k\ge\frac{|\bra{\psik}X_A\ket{\psikn}|-\epsilon\Delta_{X_{A'}}}{\sqrt{2}}.
}
Now, let us take an arbitrary pure state $\ket{\psi}$ on $A$, then, there exist $\ket{\psi_{k}}$ and $\ket{\psikn}$ satisfying \eqref{maru1A} and \eqref{maru1B} and a phase $\theta$ such that
\eq{
\ket{\psi}=\sqrt{r}\ket{\psik}+\sqrt{1-r}e^{i\theta}\ket{\psikn}.
}
Then,
\eq{
|\ex{[X_{A},Q_k]}_{\psi}|&=\sqrt{r(1-r)}|\mathrm{Im}(\bra{\psik}X_{A}\ket{\psikn})|\non
&\le\frac{|\bra{\psik}X_{A}\ket{\psikn}|}{2}.
}
Therefore,
\eq{
\max\{\calC_k|\mbox{$\psi_k$ and $\psikn$ satisfying \eqref{maru1A} and \eqref{maru1B}}\}&\ge\frac{\max_{\psi}2|\ex{[X_{A},Q_k]}_{\psi}|}{\sqrt{2}}-\frac{\epsilon\Delta_{X_{A'}}}{\sqrt{2}}\non
&=\sqrt{2}\|[X_A,Q_k]\|_\infty-\frac{\epsilon\Delta_{X_{A'}}}{\sqrt{2}}.
}
By combining the above, we obtain
\eq{
\sqrt{\calF^{\mathrm{cost}}_{\calP}}\ge\max_k\frac{\sqrt{2}\|[X_A,Q_k]\|_{\infty}}{\epsilon}-(\Delta_{X_{A}}+(1+\frac{1}{\sqrt{2}})\Delta_{X_{A'}}).
}

\end{proof}

\subsection{Unitary gates: a quantitative Wigner-Araki-Yanase type theorem for fidelity error}\label{SIsubsec_unitaryWAY}

Next, we apply \eqref{cost-CS} to unitary gates.
We can derive the following theorem from \eqref{cost-CS}:
\begin{theorem}
Let $\calE$ be a CPTP map from $A$ to $A$.
We assume that $\calE$ approximates a unitary gate $U_A$ on $A$, i.e. for a positive number $\epsilon$,
\eq{
D_F(\calE(\rho),\calU_A(\rho))\le\epsilon,\enskip\forall\rho\mbox{ on }A.
}
Then, the implementation cost of $\calE$ under conservation law of $X$ is bounded as follows:
\eq{
\sqrt{\calF^{\mathrm{cost}}_{\calE}}\ge\frac{\calA_{U_A}}{\sqrt{2}\epsilon}-3\Delta_{X_A}\label{cost_uni}
}
Here $\calA_{U_A}:=\frac{\max_{\rho}\ex{X_A-U^\dagger_AX_AU_A}_{\rho}-\min_{\rho}\ex{X_A-U^\dagger_AX_AU_A}_{\rho}}{2}$.
\end{theorem}

\textbf{Remark:} Due to $2\|[U_A,X_A]\|_{\infty}\ge\calA_{U_A}\ge\|[U_A,X_A]\|_{\infty}$, we can also obtain the following inequality:
\eq{
\sqrt{\calF^{\mathrm{cost}}_{\calE}}\ge\frac{\|[X_A,U_A]\|_\infty}{\sqrt{2}\epsilon}-3\Delta_{X_A}
}

\begin{proof}
We take a recovery map $\calR_{U_A}$ as $\calR_{U_A}(...):=U^\dagger_A(...)U_A$.
Then, clearly,
\eq{
D_F(\calR_{U_A}\circ\calE(\rho),\rho)\le\epsilon,\enskip\forall\rho\mbox{ on }A.
}
Therefore, for an arbitrary test ensemble $\{p_k,\rho_k\}$ satisfying $F(\rho_k,\rho_{k'})=0$, $\delta\le\epsilon$ holds, and thus
\eq{
\sqrt{\calF^{\mathrm{cost}}_{\calE}}\ge\frac{\calC}{\epsilon}-2\Delta_{X_A}.
}
Therefore, we only have to show $\calC\ge\frac{\calA_{U_A}}{\sqrt{2}}-\epsilon\Delta_{X_A}$ for a proper test ensemble.

Now, let us define two states $\ket{\psi_{\max}}$ and $\ket{\psi_{\min}}$ as the eigenvectors of $X_A-U^\dagger X_A U_A$ with the maximum and minimum eigenvalues, respectively.
We also define
\eq{
\ket{\psi_+}:=\frac{\ket{\psi_{\max}}+\ket{\psi_{\min}}}{\sqrt{2}},\\
\ket{\psi_-}:=\frac{\ket{\psi_{\max}}-\ket{\psi_{\min}}}{\sqrt{2}}.
}
Let us take a test ensemble $\{(1/2,\psi_+),(1/2,\psi_-)\}$. Then, the corresponding $\calC$  satisfies
\eq{
\calC&=\frac{|\bra{\psi_+}(X_A-\calE^\dagger(X_A))\ket{\psi_-}|}{\sqrt{2}}\non
&=\frac{|\bra{\psi_+}(X_A-U^\dagger_AX_AU_A-(\calE^\dagger(X_A)-U^\dagger_AX_AU_A))\ket{\psi_-}|}{\sqrt{2}}\non
&\ge\frac{|\bra{\psi_+}(X_A-U^\dagger_AX_AU_A)\ket{\psi_-}|
-|\bra{\psi_+}(\calE^\dagger(X_A)-U^\dagger_AX_AU_A))\ket{\psi_-}|
}{\sqrt{2}}
}
We can evaluate $|\bra{\psi_+}(X_A-U^\dagger_AX_AU_A)\ket{\psi_-}|$ as follows:
\eq{
|\bra{\psi_+}(X_A-U^\dagger_AX_AU_A)\ket{\psi_-}|&=\frac{|\ex{X_A-U^\dagger_AX_AU_A}_{\psi_{\max}}-\ex{X_A-U^\dagger_AX_AU_A}_{\psi_{\min}}|}{2}\non
&=\calA_{U_A}
}
To evaluate $|\bra{\psi_+}(\calE^\dagger(X_A)-U^\dagger_AX_AU_A))\ket{\psi_-}|$, note that the following relation holds for arbitrary real number $x$
\eq{
\max_{\rho}|\ex{\calE^\dagger(X_A)-U^\dagger_AX_AU_A}_{\rho}|
&=\max_{\rho}|\Tr[(X_A)(\calE(\rho)-U\rho U^\dagger)]|
\non
&=\min_{x\in\mathbb{R}}\max_{\rho}|\Tr[(X_A-xI_A)(\calE(\rho)-U\rho U^\dagger)]|
\non
&\le\frac{\Delta_{X_A}}{2}\times\epsilon.
}
Therefore, we obtain
\eq{
\calC\ge\frac{\calA_{U_A}}{\sqrt{2}}-\frac{\epsilon\Delta_{X_A}}{2\sqrt{2}}.
}
By combining the above, we obtain
\eq{
\sqrt{\calF^{\mathrm{cost}}_{\calE}}\ge\frac{\calA_{U_A}}{\sqrt{2}\epsilon}-3\Delta_{X_A}.
}

\end{proof}

\subsection{No-go theorems for the channel implementation beyond unitary channels}\label{SIsubsec_non-unitaryWAY}

\begin{corollary}\label{cor:no-go channel implementation}
Let $U$ be a unitary and $\calN$ be a channel. If there exist two orthogonal eigenstates $\ket{x_{1,2}}$ of $X$ such that $\braket{x_1|U^\dagger X U|x_2}\neq 0$ and $\calN(\dm{x_{1,2}})=\dm{x_{1,2}}$, then $\calE=\calN\circ\calU$ cannot be exactly implemented by a finite coherence resource state.
\end{corollary}

\begin{proof}
Let $\ket{\psi_{1,2}}:= U^\dagger\ket{x_{1,2}}$.
Since $\calE(\psi_{1,2})=\calN(\dm{x_{1,2}})=\dm{x_{1,2}}$, the two states $\calE(\psi_{1,2})$ can be brought back to $\psi_{1,2}$ by applying $U^\dagger$ exactly, leading to $\delta=0$. 

On the other hand, by choosing a test-state ensemble as $\{\{1/2,\psi_1\},\{1/2,\psi_2\}\}$ we have
\eq{
 \calC&=\frac{1}{\sqrt{2}}|\braket{\psi_1|Y|\psi_2}|^2 \\
 &= \frac{1}{\sqrt{2}}|\braket{\psi_1|X|\psi_2}-\braket{\psi_1|\calE^\dagger(X)|\psi_2}|^2\\
 &= \frac{1}{\sqrt{2}}|\braket{\psi_1|X|\psi_2}-\braket{x_1|\calN^\dagger(X)|x_2}|^2
}
Let $\{K_\mu\}_\mu$ be a set of Kraus operators for $\calN$. For $i=0,1$, $\calN(\dm{x_i})=\dm{x_i}$ implies $\sum_{\mu} K_{\mu} \dm{x_i} K_{\mu}^{\dagger} = \dm{x_i}$ and thus $K_{\mu}\ket{x_i}=c_{\mu,i} \ket{x_i}$ for some $c_{\mu,i}\in\mathbb{C}$. 
This gives 
\eq{
 \braket{x_1|\calN^\dagger(X)|x_2} &= \sum_\mu \braket{x_1|K_\mu^\dagger X K_\mu|x_2}\non
 & = \sum_\mu c_{\mu,1}^* c_{\mu,2} \braket{x_1|X|x_2} \non
 &= 0,
}
where in the last equality, we used the assumption that $\ket{x_1}$ and $\ket{x_2}$ are orthogonal eigenstates of $X$.
Therefore, 
\eq{
 \calC = \frac{1}{\sqrt{2}}|\braket{\psi_1|X|\psi_2}|^2 = \frac{1}{\sqrt{2}}|\braket{x_1|U^\dagger X U |x_2}|^2 >0
}
where the last inequality is due to the assumption that $\braket{x_1|U^\dagger X U | x_2}\neq 0$.

Therefore, if $\calE$ was exactly implementable by a finite $\calF_\calE^{\rm cost}$, it would contradict with \eqref{cost-C}. 

\end{proof}

\textbf{Remark}
Corollary~\ref{cor:no-go channel implementation} is NOT a direct consequence of the no-go theorem for the implementation of coherent unitary. This is because the implementation of $\calE=\calN\circ\calU$ is not unique, and thus there are many other ways of realizing $\calE$ other than sequentially implementing $\calU$ and $\calN$. 
The above result prohibits any such implementation of $\calE$---the no-go theorem for the implementation of coherent unitary is rather a special case of Corollary~\ref{cor:no-go channel implementation}.
Thus, this result extends the class of operations that do not allow for ``resource state + free operation'' implementation to that of non-unitary channels. 

For instance, a non-unitary example can be constructed by taking a coherent unitary $U$ and a dephasing channel $\calN(\cdot)=\sum_i \Pi_i \cdot \Pi_i$, where $\Pi_i$ is the projection onto the subspace of charge $i$. 
The corresponding channel $\mathcal{E}=\mathcal{N}\circ\mathcal{U}$ is then a dephasing with respect to a rotated basis, and the above result ensures that such a dephasing cannot be implemented by any means with a finite coherent resource.  

This observation can be extended to obtain the following corollary.

\begin{corollary}
Let $\calN$ be a channel with a decoherence-free subspace $\calH_{\rm DFS}$ with a dimension greater than or equal to 2. 
If two orthogonal states $\ket{x_1},\ket{x_2}\in\calH_{\rm DFS}$ satisfy $\braket{x_1|U^\dagger X U|x_2}\neq 0$ for some unitary $U$, then $\calE=\calN\circ\calU$ cannot be implemented exactly with a finite coherent resource. 
\end{corollary}

\subsection{Quantum error correction: A refinement of Eastin-Knill theorem}\label{SIsubsec_EK}
Next, we apply \eqref{cost-CS} (\eqref{cost-C} in the main text) to quantum error correction.
To be concrete, we derive an extended version of the approximate Eastin-Knill theorem in Ref. \cite{e-EKFaist} from \eqref{cost-CS}.
We follow the setup for Theorem 1 in Ref. \cite{e-EKFaist}, and assume the following three conditions:
\begin{itemize}
\item We consider a code channel $\calE_{\mathrm{code}}$ from the ``logical system'' $L$ to the ``physical system'' $P$. We assume that the code $\Code$ is isometry and covariant with respect to $\{U^{L}_{\theta}\}_{\theta\in\mathbb{R}}$ and $\{U^{P}_{\theta}\}_{\theta\in\mathbb{R}}$, where $U^{L}_{\theta}:=e^{i\theta X_L}$ and  $U^{P}_{\theta}:=e^{i\theta X_P}$. 
\item The physical system $P$ is assumed to be a composite system of $N$ subsystems $\{P_{i}\}^{N}_{i=1}$, and the operator $X_{P}$ in $U^{P}_{\theta}$ is assumed to be written as $X_{P}=\sum_iX_{P_i}$. 
\item The noise $\calN$ that occurs after the code channel $\Code$ is assumed to be the erasure noise, and the location of the noise is assumed to be known. Concretely, the noise $\calN$ is defined as a CPTP map from $P$ to $P':=PM$ written as follows:
\begin{align}
\calN(...):=\sum_{i}\frac{1}{N}\ket{i_M}\bra{i_M}\otimes\ket{\tau_i}\bra{\tau_i}_{P_i}\otimes\Tr_{P_i}[...],\label{noise-def}
\end{align}
where the subsystem $M$ is a memory that remembers the location of the error, and $\{\ket{i_M}\}$ is an orthonormal basis of $M$. 
Each state $\ket{\tau_i}_{P_i}$ is a given fixed state in $P_i$. 
\end{itemize}
After the noise, we perform a recovery CPTP map $\calR$ and try to recover the initial state.
Now, let us take an arbitrary test ensemble $\{p_k,\rho_k\}$ and consider $\delta$ for the test ensemble and the channel $\calR\circ\calN\circ\Code$.
Then, we can interpret $\delta$ as the recovery error of the code $\Code$.
We define the error of the channel $\Code$ for the noise $\calN$ for the initial states $\{p_k,\rho_k\}$ as follows:
\eq{
\epsilon(\calC,\calN,\{p_k,\rho_k\}):=\delta \mbox{ for the channel $\calR\circ\calN\circ\calC$ and the test ensemble $\{p_k,\rho_k\}$}.
}
We remark that $\epsilon(\Code,\calN,\{p_k,\rho_k\})$ is not the worst-case entanglement fidelity. It is defined as the fidelity error and it can describe the recovery error for specific initial states $\{p_k,\rho_k\}$ on $L$.
Our inequalities \eqref{cost-C} (\eqref{cost-CS} in the supplementary) and \eqref{SIQ-Gini} show that even for $\delta$, an approximate Eastin-Knill type bound holds. We stress that our Eastin-Knill type bound gives the approximate Eastin-Knill theorem itself since $\delta$ is lower than the worst-case entanglement fidelity error.
 
Let us derive an Eastin-Knill type bound from \eqref{cost-CS}.
To begin with, we define the following channel $\tilde{\calN}$:
\begin{align}
\tilde{\calN}(...):=\sum_{i}\frac{1}{N}\ket{i_M}\bra{i_M}\otimes\ket{0_i}\bra{0_i}_{P_i}\otimes\Tr_{P_i}[...]
\end{align}
where $\ket{0_i}$ is the ground eigenvector of $X_{P_i}$.
Then, the channel $\tilde{\calN}$ is covariant and satisfies the following equality:
\eq{
\epsilon(\Code,\calN,\{p_k,\rho_k\})=\epsilon(\Code,\tilde{\calN},\{p_k,\rho_k\})\label{conNN}
}
To show that the above equality holds, we only have to note that we can transform the final state of $\tilde{\calN}\circ\Code$ to that of $\calN\circ\Code$ by the following unitary:
\begin{align}
W:=\sum_{i}\ket{i_M}\bra{i_M}\otimes U_{P_i}\otimes_{j:j\ne i} I_{P_j},
\end{align}
where $U_{P_i}$ is a unitary on $P_i$ converting $\ket{0_i}$ to $\ket{\tau_i}$.

Due to \eqref{conNN}, we can use \eqref{cost-CS} to derive a general lower bound for $\epsilon(\Code,\calN,\{p_k,\rho_k\})$.
And, since $\Code$ and $\tilde{\calN}$ are covariant, $\calF^{\mathrm{cost}}_{\tilde{\calN}\circ\Code}=0$.
Concretely, the following relation is directly derived from \eqref{cost-CS} and \eqref{conNN} for arbitrary $\{p_k,\rho_k\}$ satisfying $F(\rho_k,\rho_{k'})=0$:
\eq{
\epsilon(\Code,\calN,\{p_k,\rho_k\})\ge\frac{\calC}{\sqrt{\calF^{\mathrm{cost}}_{\tilde{\calN}\circ\Code}}+\Delta}
}
Here $\calC$ and $\Delta$ defined for $\{p_k,\rho_k\}$ and $Y:=X_L-\Code^\dagger\circ\tilde{\calN}^\dagger(X_{P}\otimes I_C)$.
And, since $\Code$ and $\tilde{\calN}$ are covariant, $\calF^{\mathrm{cost}}_{\tilde{\calN}\circ\Code}=0$. Therefore, we obtain
\eq{
\epsilon(\Code,\calN,\{p_k,\rho_k\})\ge\frac{\calC}{\Delta}\label{QECB1}
}
Similarly, from \eqref{SIQ-Gini}, we can derive the following relation for arbitrary $\{p_k,\rho_k\}$:
\eq{
\epsilon(\Code,\calN,\{p_k,\rho_k\})\ge\frac{\calC^2}{\Delta^2}.\label{QECB2}
}
We remark that \eqref{QECB1} and \eqref{QECB2} hold for not only the erasure noise but other arbitrary covariant noise.

Finally, we show that we can derive the approximate Eastin-Knill theorem (Theorem 1 in Ref. \cite{e-EKFaist}) from \eqref{QECB1}.
To be concrete, we show that the following inequality holds for a specific $\{p_k,\psi_k\}$:
\begin{align}
\frac{\Delta_{X_L}}{\Delta_{X_L}+4\sqrt{2}N\max_{i}\Delta_{X_{P_i}}}\le \epsilon(\Code,\calN,\{p_k,\psi_k\}).\label{classicalFaist}
\end{align}
This inequality holds whenever $\Delta_{X_L}>0$.
This inequality gives a corollary $\frac{\Delta_{X_L}}{\Delta_{X_L}+4\sqrt{2}N\max_{i}\Delta_{X_{P_i}}}\le\ew$ where $\ew$ is the worst-case entanglement purified distance defined in \eqref{S78}. The corollary is almost the same as the inequality in Theorem 1 of \cite{e-EKFaist}. (In Theorem 1 of \cite{e-EKFaist}, $\frac{\Delta_{X_L}}{2N\max_{i}\Delta_{X_{P_i}}}\le \ew$ is given.)
We stress that \eqref{classicalFaist} and \eqref{QECB1} are qualitatively different from Theorem 1 of \cite{e-EKFaist}, since \eqref{classicalFaist} and \eqref{QECB1} are universal bounds for the fidelity error for ensembles only on the logical system $L$ without the reference system $R$.

\begin{proofof}{\eqref{classicalFaist}}
To derive \eqref{classicalFaist}, we only have to find proper $\{p_k,\psi_k\}$ and calculate $\calC$ and $\Delta_2$ for it.
Le us take the spectral decomposition $X_L=\sum_{j}x_j\ket{j}\bra{j}$, and refer to the maximum and minimum eigenvalues as $x_{j_*}$ and $x_{j'_{*}}$, respectively. We take the test ensemble that we seek as follows:
\eq{
\ket{\psi_\pm}:=\frac{\ket{j_*}\pm\ket{j'_*}}{\sqrt{2}},\enskip
p_k:=\frac{1}{2}.
}
Let us calculate $\calC$ and $\Delta_2$.
Before the calculation, we remark that due to \eqref{i-s}, $\calC$ and $\Delta_2$ do not change if we shift $X_L$ to $X_L-aI_L$ and $X_P\otimes I_C$ to $X_P\otimes I_C-b I_{PC}$ where $a$ and $b$ are arbitrary real numbers.
Therefore, without loss of generality, we can assume that the eigenvalue of $\ket{0_i}$ is zero and that $\Delta_{X_L}=\|X_L\|$. 
Let us calculate $\Code^\dagger\circ\tilde{\calN}^\dagger(X_P\otimes I_C)$ first. We remark that since $\Code$ is an isometry channel, there exists an isometry $V$ satisfying $\Code(...)=V...V^\dagger$ and $V^\dagger V=I$. (Note that $VV^\dagger$ is just a projection). And since $\Code$ is covariant with respect to $\{e^{i\theta X_L}\}$ and $\{e^{i\theta X_P}\}$, $VX_L= X_PV$ holds.
By definition of $\tilde{\calN}$, we can see that
\eq{
\tilde{\calN}^\dagger(X_P\otimes I_C)&=\frac{1}{N}\sum^{N}_{i=1}\left(\bra{0_i}X_{P_i}\ket{0_i}I_{P_i}+\sum_{i':i\ne i'}X_{P_{i'}}\right)\non
&=\frac{N-1}{N}X_{P}\\
\tilde{\calN}^\dagger(X^2_P\otimes I_C)&=\tilde{\calN}^\dagger\left(\sum_{i',i''}X_{P_{i'}}X_{P_{i''}}\otimes I_C\right)\non
&=\frac{1}{N}\sum^{N}_{i=1}\left(\sum_{i':i\ne i'}\sum_{i'':i\ne i''}X_{P_{i'}}X_{P_{i''}}\right)\non
&=\frac{1}{N}\sum^{N}_{i=1}\left(X_{P}-X_{P_i}\right)^2\non
&=\left(1-\frac{2}{N}\right)X^2_{P}+\frac{1}{N}\sum^N_{i=1}X^2_{P_i}
}
Here, the terms proportional to $I$ are omitted because they won't contribute to $\Delta_Y$ or $\mathcal{C}$.
Therefore, we obtain
\eq{
\Code^\dagger\circ\tilde{\calN}^\dagger(X_P\otimes I_C)&=V^\dagger \left(\frac{N-1}{N}X_{P} \right)V\non
&=\frac{N-1}{N}X_{L}\\
\Code^\dagger\circ\tilde{\calN}^\dagger(X^2_P\otimes I_C)&=V^\dagger \left(\left(1-\frac{2}{N}\right)X^2_{P}+\frac{1}{N}\sum^N_{i=1}X^2_{P_i} \right)V\non
&=\frac{N-2}{N}X^2_{L}+V^\dagger \left(\frac{1}{N}\sum^N_{i=1}X^2_{P_i} \right)V
}
Hence, we obtain 
\eq{
Y=\frac{1}{N}X_L
}
Therefore, we obtain
\eq{
\Delta_Y&=\frac{1}{N}\Delta_{X_L}
\\
\calC&=\sqrt{\frac{|\bra{\psi_+}Y\ket{\psi_-}|^2}{2}}\non
&=\sqrt{\frac{|(\bra{j_*}+\bra{j'_*})X_L(\ket{j_*}-\ket{j'_*})|^2}{8N^2}}\non
&=\frac{\Delta_{X_L}}{2\sqrt{2}N}.
}
and 
\eq{
2\sqrt{\|\Code^\dagger\circ\tilde{\calN}^\dagger(X^2_P\otimes I_C)-\Code^\dagger\circ\tilde{\calN}^\dagger(X_P\otimes I_C)^2\|_\infty}
&=2\sqrt{\|V^\dagger \left(\frac{1}{N}\sum^N_{i=1}X^2_{P_i} \right)V-\frac{1}{N^2}X^2_L\|_\infty}\non
&\stackrel{(a)}{\le}2\sqrt{\|V^\dagger \left(\frac{1}{N}\sum^N_{i=1}X^2_{P_i}\right)V\|_\infty}\non
&=\max_{\ket{\psi}\mbox{ on }L}2\sqrt{\bra{\psi}V^\dagger \left(\frac{1}{N}\sum^N_{i=1}X^2_{P_i}\right)V\ket{\psi}}\non
&\le\max_{\ket{\phi}\mbox{ on }P}2\sqrt{\bra{\phi} \left(\frac{1}{N}\sum^N_{i=1}X^2_{P_i}\right)\ket{\phi}}\non
&\le2\max_{i}\Delta_{X_{P_i}}
}
Here, we use $0\le A\le B\Rightarrow \|A\|_\infty\le\|B\|_\infty$ and $\Code^\dagger\circ\tilde{\calN}^\dagger(X^2_P\otimes I_C)-\Code^\dagger\circ\tilde{\calN}^\dagger(X_P\otimes I_C)^2\ge0$ in $(a)$.
Combining the above, we obtain \eqref{classicalFaist}.
\end{proofof}

\section{Application to black hole physics and information scrambling}\label{SIsecBH}

Our results also provide helpful insights into how the symmetry of black hole dynamics affects the recovery of information from black holes. 
 To be concrete, we present a rigorous lower bound on how many of the $m$ bits of classical information string cannot be recovered in an information recovery protocol from a black hole with the energy conservation law.
 
 We first overview the background.
 In black hole physics, black holes and Hawking radiation from the black holes are often regarded as quantum many-body systems, and how much information thrown into a black hole can be recovered from Hawking radiation has been analyzed.
 One of the pioneering studies is the Hayden-Preskill thought experiment \cite{HP}.
 In the thought experiment, one considers the situation in which Alice throws a quantum system $A$ (her ``diary" in the original paper) into a quantum black hole $B$ (Figure \ref{HPmodel_main_S}). And another person, Bob, tries to recover the diary's contents from the Hawking radiation from the black hole.
Then, we assume the following three basic assumptions.
First, the black hole is old enough, and thus there is a quantum system $R_B$ corresponding to the early Hawking radiation that is maximally entangled with the black hole. To decode Alice's diary contents, Bob can use not only the Hawking radiation $A'$ after Alice throw her diary but also the early radiation $R_B$. Second, each system is described as qubits. We refer to the numbers of qubits of $A$, $A'$, and $B$ as $k$, $l$, and $N$, respectively. 
Third, the dynamics of the black hole is the Haar random unitary $U$. These assumptions, especially the second and third, are pretty strong but widely accepted today. (For details of the justification of the model, see the review \cite{BH_review} for example.)

\begin{figure}[tb]
		\centering
		\includegraphics[width=.5\textwidth]{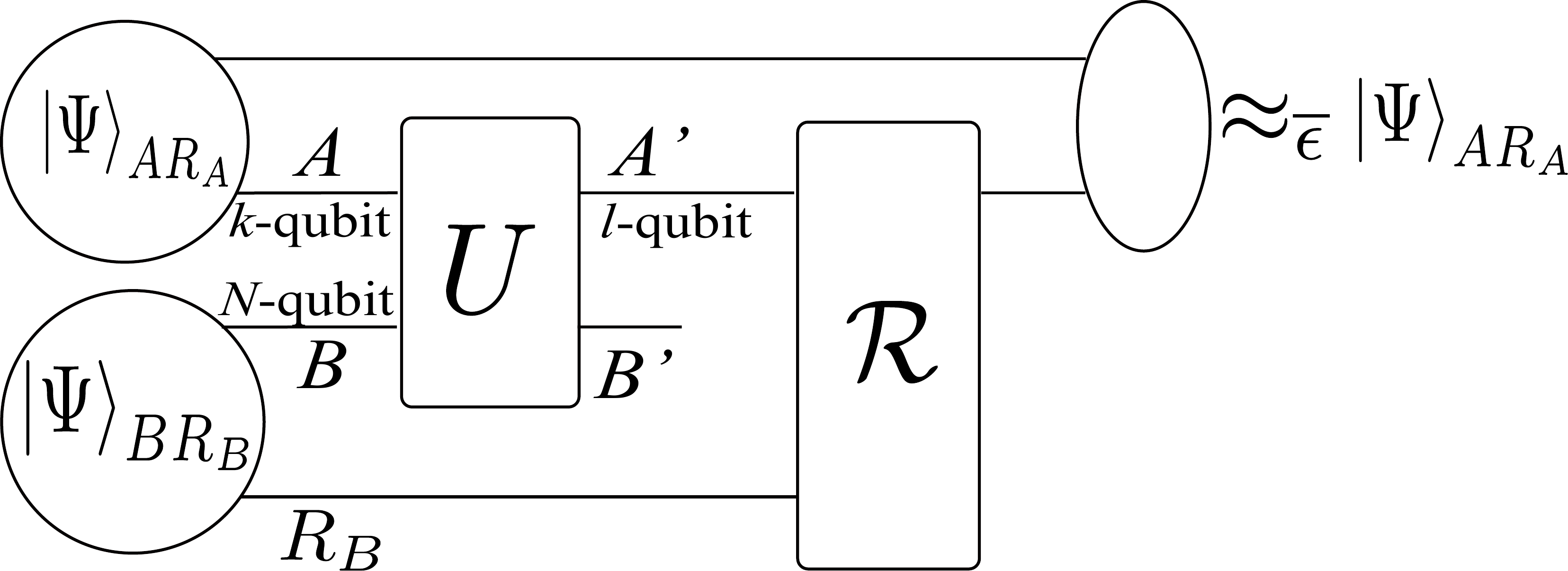}
		\caption{Schematic diagram of the Hayden-Preskill black hole model. This is the schematic for the quantum information recovery. The classical information recovery is described in the next figure.}
		\label{HPmodel_main_S}
	\end{figure}

Under the above settings, Hayden and Preskill considered how long Bob should wait to see the contents of Alice's diary. For the analysis, they considered a entanglement-fidelity based recovery error $\me$ defined as $\me:=\min_{\calR_{A'\rightarrow A}}D_F(\calR_{A'\rightarrow A}\circ\calE\otimes\mathrm{id}_{R}(\Psi),\Psi)$, where $\Psi$ is the maximally entangled state between $A$ and an external reference system $R_A$.
And for the decoding error $\me$, they derived the following inequality:
\eq{
\me \le 2^{-(l-k)}.
}  
The implication of this inequality was surprising: Bob hardly has to wait and can get the almost complete contents when the number of qubits in Hawking radiation $A'$ was just a little more than the number of qubits in $A$.

The above result is derived via a rigorous argument once the setup is accepted. However, the above setup does not take conservation laws into account. Since the conservation law of energy for the whole system should be satisfied even for a black hole, it is necessary to consider the energy conservation law for a more accurate analysis.
In recent years, analyses based on this idea have progressed, and it has been shown that taking energy conservation into account delays the escape of information from a black hole \cite{Yoshida-soft,JLiu,Nakata,TS}. 
These developments suggest that when the unitary $U$ is a Haar random unitary satisfying the energy conservation law, Bob may not read Alice's diary to some extent under the energy conservation law.
However, the question of how many classical bits in Alice's diary become unreadable for Bob has not been analyzed.
To evaluate the number of unreadable classical bits, we cannot use the entanglement-fidelity-based errors that were well used in the research related to Hayden-Preskill thought experiments, due to the following two problems.
The first problem is in the fact that the fidelity between two states becomes 0 even when only the states of a single qubit are orthogonal. Thus the fidelity-based analysis does not allow us to determine how many bits of Alice's diary are unreadable for Bob. More precisely, if the fidelity error is small enough, no letter in the diary can be unreadable, and thus the contents of the diary cannot be hidden. However, when making a prediction that the fidelity error will be large, the above point becomes a problem. In this case, even if $\overline{\epsilon}$ takes the maximum value 1, we cannot know whether the diary will be unreadable or not, since the fidelity error will be large even if only one letter becomes unreadable.

The second problem is that the entanglement-fidelity-based analysis cannot assess the errors in the classical information encoded into specific quantum states.
The quantum information recovery error $\overline{\epsilon}$ is approximately equal to "the average of the fidelity error that would be produced if one took a state from the Hilbert space of $A$ according to the Haar measure and executed the above Hayden-Preskill protocol with the state as the initial state. On the other hand, when encoding a classical $m$-bit string $\vect{a}=(a_1,...,a_m)$ into quantum states in a given way, the encoded quantum states become just $2^m$ orthogonal pure states $\{\ket{\psi_{\vect{a}}}\}$. Since this set is a 0-measure subspace in the Hilbert space of $A$, the relation between the recovery error for the classical data $\vect{a}$ from this set $\{\ket{\psi_{\vect{a}}}\}$ and the quantum recovery error $\overline{\epsilon}$ is not clear. This point might not cause trouble in the case of black holes with no conservation laws. When there are no conservation laws, the dynamics of the black hole are completely Haar random, and thus there is no reason to consider any particular $2^m$ pure states as special ones, and it becomes reasonable to infer that if $\overline{\epsilon}$ is small, recovery is possible with almost any set of pure states including $\{\ket{\psi_{\vect{a}}}\}$. However, if there is a conservation law e.g. the energy conservation, the basis of the conserved charge will be special, and it is likely that $2^m$ orthogonal pure states $\{\ket{\psi_{\vect{a}}}\}$ well defined using this basis will exhibit completely different behavior than $2^m$ orthogonal pure states taken out at random.

Our result \eqref{SIQ-Cini} allows us to overcome both of the above two problems. As shown in the main text, we can derive a lower bound on how many classical bits in Alice's diary will be lost as a corollary of \eqref{SIQ-Cini}.  The purpose of this section is to show the details of the derivation of \eqref{BH_Hamming_main} in the main text.
For readers' convenience, we introduce our setup and result again (Fig. \ref{HPmodel_classical}).
Following the Hayden-Preskill model, we consider the situation that Alice throws a quantum system $A$ (her diary in the original paper \cite{HP}) into a quantum black hole $B$ (Figure \ref{HPmodel_classical}). And another person, Bob, tries to recover the diary's contents from the Hawking radiation from the black hole.
Then, we assume the following three basic assumptions.
First, the black hole is old enough, and thus there is a quantum system $R_B$ corresponding to the early Hawking radiation that is maximally entangled with the black hole. To decode Alice's diary contents, Bob can use not only the Hawking radiation $A'$ after Alice throws her diary but also the early radiation $R_B$. Second, each system is described as qubits. We refer to the numbers of qubits of $A$, $A'$ and $B$ as $k$, $l$, and $N$, respectively. 
Third, the dynamics of the black hole $U$ satisfy the following three conditions. We stress that the second and third conditions are valid when $U$ is a typical Haar random unitary with the conservation law, as shown in Ref. \cite{TS}. Therefore, our results are valid for the Hayden-Preskill model with Haar random unitary dynamics with the energy conservation law.
\begin{itemize}
\item{The dynamics $U$ satisfies $U^\dagger (X_{A'}+X_{B'})U=X_A+X_B$.}
\item{Let $\ket{i,a}_A$ and $\ket{j,b}_B$ be energy eigenstates of $X_A$ and $X_B$ with the eigenvalues $x_{i,A}$ and $x_{j,B}$, respectively. Here $a$ and $b$ are the reference for degeneracies. Let $\rho'_{A'|i,a,j,b}$ and $\rho'_{B'|i,a,j,b}$ be the following states:
\begin{align}
\rho'_{A'|i,a,j,b,U}:=\Tr_{B'}[U(\ket{i,a}\bra{i,a}\otimes\ket{j,b}\bra{j,b})U^\dagger]\label{S107} ,\\
\rho'_{B'|i,a,j,b,U}:=\Tr_{A'}[U(\ket{i,a}\bra{i,a}\otimes\ket{j,b}\bra{j,b})U^\dagger].
\end{align}
Then,  the following relation holds:
\begin{align}
V_{\rho'_{\alpha'|i,a,j,b,U}}(X_{\alpha'})\le\frac{1+\epsilon}{4}\min\{l,\gamma(N+k)\},\label{exa}
\end{align}
where $\alpha'$ is $A'$ or $B'$, $\gamma:=1-l/(N+k)$, and $\epsilon$ is a negligible small positive number that is smaller than $1/(N+k)^2$. }
\item{The expectation values of the conserved quantity $X$ are approximately divided among $A'$ and $B'$ in proportional to the corresponding number of qubits.
In other words, the final state on $A'B'$ is thermalized in the sense of the expectation value.
To be concrete, when $N\ge 10^3$ and $\rho_B$ is the maximally entangled state, for any $\rho$ on $A$, the following two relation holds:
\begin{align}
\ex{X_{A'}}_{\Tr_{B'}[U\rho\otimes\rho_BU^\dagger]}
\approx_{\epsilon}(\ex{X_A}_{\rho}+\ex{X_B}_{\rho_B})\times(1-\gamma).\label{ex-sc-eq_TTKS}
\end{align}
Here, $x\approx_{\epsilon}y\Leftrightarrow_{\mbox{def}}|x-y|\le\epsilon$, and $\epsilon$ is a negligible small number which satisfies $1/(N+k)^3\le \epsilon\le 1/(N+k)^2$.

Furthermore, when $N\ge 10^3$ and $15<i+j<N-15$, the following relation holds:
\eq{
\ex{X_{A'}}_{\rho'_{A'|i,a,j,b,U}}
\approx_{\epsilon}(\ex{X_A}_{\ket{i,a}\bra{i,a}_A}+\ex{X_B}_{\ket{j,b}\bra{j,b}_B})\times(1-\gamma).\label{ex-sc-eq_TTKS_B}
}
}
\end{itemize}

Under the above assumptions, we define the error using the Hamming distance.
We first introduce a classical $m$-bit string $\vect{a}:=(a_1,...,a_m)$. Here each $a_j$ takes values 0 or 1.
To encode the classical string $\vect{a}$, we prepare the diary $A$ as a composite system of $m$ subsystems $A=A_1...A_m$, where each $A_j$ consists of $n$ qubits. Namely, $k=mn$ holds.
We assume that each qubit in $A$ has the same conserved observable (e.g. energy) $X:=\ket{1}\bra{1}$. 
We also prepare two pure states $\ket{\psi^{(A_j)}_{a_j}}$ ($a_j=0,1$) on each subsystem $A_j$ which are orthogonal to each other. 
Using the pure states, we encode the string $\vect{a}$ into a pure state $\ket{\psi_{\vect{a}}}:=\otimes^{m}_{j=1}\ket{\psi^{(A_j)}_{a_j}}$ on $A$.
After the preparation, we throw the pure state $\ket{\psi_{\vect{a}}}$ into the black hole $B$. In other words, we perform the energy-preserving Haar random unitary $U$ on $AB$.
After the unitary dynamics $U$, we try to recover the classical information $\vect{a}$. We perform a general measurement $\calM$ on $A'R_B$, and obtain a classical $m$-bit string $\vect{a}'$ with probability $p'_{\vect{a}}(\vect{a}')$. We define the recovery error $\delta_H$ by averaging the Hamming distance between $\vect{a}$ and $\vect{a}'$ for all possible input $\vect{a}$ as follows:
\eq{
\delta_H:=\sum_{\vect{a},\vect{a}'}\frac{p'_{\vect{a}}(\vect{a}')}{2^m}h(\vect{a},\vect{a}').
}
Here $h(\vect{a},\vect{a}')$ is the Hamming distance, which represents the number of different bits between $\vect{a}$ and $\vect{a}'$.

Under the above setup, using proper states $\{\ket{\psi^{(A_j)}_{a_j}}\}$, we can make $\delta_H$ proportional to $m$.
Remark that the eigenvalues of the conserved quantity $X^{(A_j)}$ on $A_j$ become integer from $0$ to $n$.
We refer to the eigenvectors of $H^{(A_j)}$ with the eigenvalues $0$ and $n$ as $\ket{0}_{A_j}$ and $\ket{n}_{A_j}$, respectively, and define $\ket{\psi^{(A_j)}_{0}}:=(\ket{0}_{A_j}+\ket{n}_{A_j})/\sqrt{2}$ and $\ket{\psi^{(A_j)}_{1}}:=(\ket{0}_{A_j}-\ket{n}_{A_j})/\sqrt{2}$, respectively.
Let us take $n:=a\sqrt{N}$, where $a$ is some positive constant satisfying $a\ge2$. When $N\ge 10^3$ and $k\le N$ holds, we obtain the following inequality from \eqref{SIQ-CiniS}
\eq{
\delta_H\ge m\times\frac{1}{4\left(1+\frac{3}{a\gamma}\right)^2},\label{BH_Hamming}
}
where $\gamma:=1-\frac{l}{N+k}$ represents the ratio between the number of qubits in the remained black hole $B'$ and the total number of qubits $A'B'$.

\begin{figure}[tb]
		\centering
		\includegraphics[width=.6\textwidth]{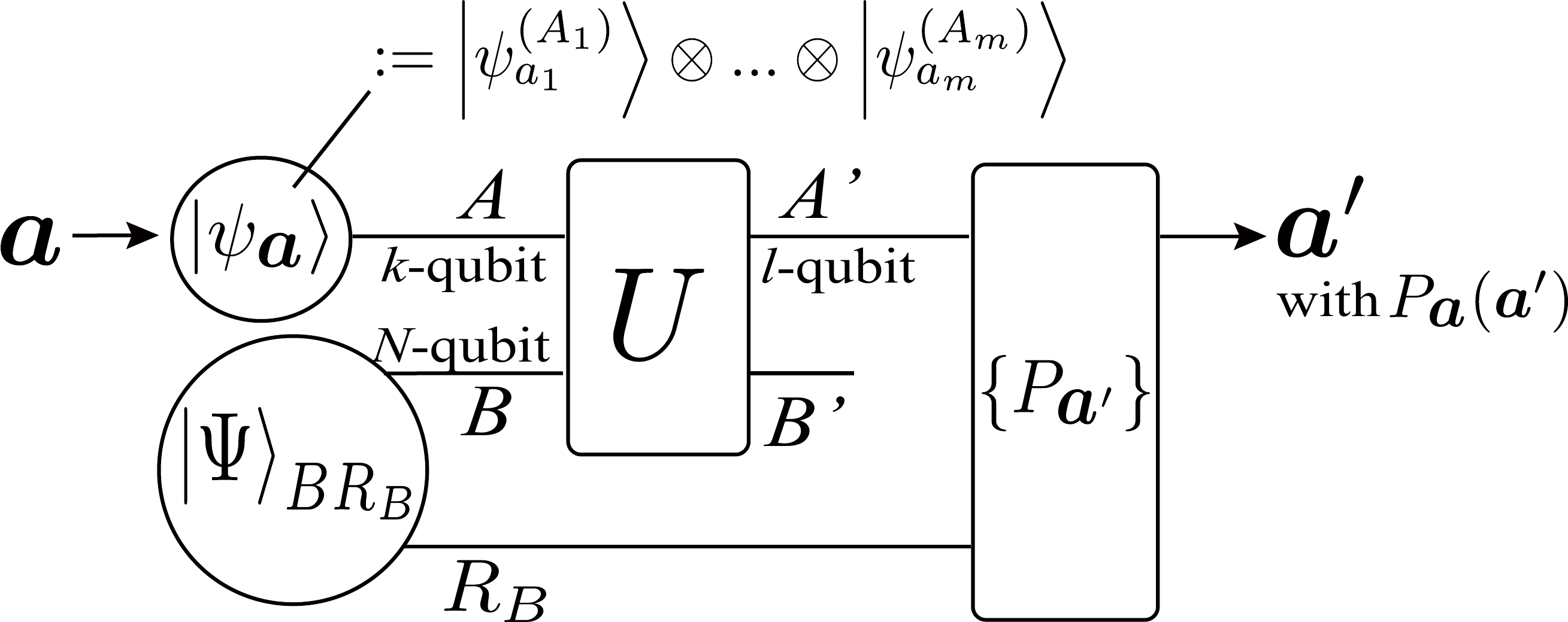}
		\caption{Schematic diagram of the classical information recovery in the Hayden-Preskill black hole model. We remark that $\vect{a}:=(a_1,...,a_m)$ and $k=m\times n$. }
		\label{HPmodel_classical}
	\end{figure}

\begin{proofof}{\eqref{BH_Hamming}}
We first remark that we can construct a recovery map $\calR_\calM:A'R_B\rightarrow A$ from a measurement $\calM$. We define the POVM of $\calM$ as $\{P_{\vect{a}'}\}$, and define $\calR_\calM$ as
\eq{
\calR_\calM(...):=\sum_{\vect{a}'}\Tr[...P_{\vect{a}'}]\psi_{\vect{a}'}.
}
Namely, when we obtain a classical bit string $\vect{a}'$ from $\calM$, $\calR_\calM$ gives $\psi_{\vect{a}'}$.

We evaluate $\delta_H$ as follows:
\eq{
\delta_H&=\sum_{\vect{a},\vect{a}'}\frac{p'_{\vect{a}}(\vect{a}')}{2^m}h(\vect{a},\vect{a}')\non
&=\sum_{\vect{a},\vect{a}'}\frac{p'_{\vect{a}}(\vect{a}')}{2^{m+1}}\sum^{m}_{j=1}\|\psi^{(A_j)}_{\vect{a}}-\psi^{(A_j)}_{\vect{a}'}\|_1
\non
&\ge\sum_{\vect{a}}\frac{1}{2^{m+1}}\sum^{m}_{j=1}\|\psi^{(A_j)}_{\vect{a}}-\sum_{\vect{a}'}p'_{\vect{a}}(\vect{a}')\psi^{(A_j)}_{\vect{a}'}\|_1\nonumber\\
&=\sum^{m}_{j=1}\sum_{\vect{a}}\frac{1}{2^{m+1}} \|\psi^{(A_j)}_{a_j}-\rho''^{(A_j)}_{\vect{a}}\|_1\nonumber\\
&\ge\sum^{m}_{j=1}\sum_{a_j:0,1}\frac{1}{4}\|\psi^{(A_j)}_{a_j}-\sum_{a_1,...,a_{j-1},a_{j+1},...,a_m}\frac{1}{2^{m-1}}\rho''^{(A_j)}_{\vect{a}}\|_1\nonumber\\
&=\sum^{m}_{j=1}\sum_{a_j:0,1}\frac{1}{4}\|\psi^{(A_j)}_{a_j}-
\Tr_{\neg A_j}[\calR_\calM\circ\calE_{A\rightarrow A'R_B}(\sum_{a_1,...,a_{j-1},a_{j+1},...,a_m}\frac{1}{2^{m-1}}\psi_{\vect{a}})]\|_1\nonumber\\
&=\sum^{m}_{j=1}\sum_{a_j:0,1}\frac{1}{4}\|\psi^{(A_j)}_{a_j}-
\Tr_{\neg A_j}[\calR_\calM\circ\calE_{A\rightarrow A'R_B}(\psi^{(A_j)}_{a_j}\otimes_{i:i\ne j}\frac{\psi^{(A_i)}_0+\psi^{(A_i)}_1}{2})]\|_1\non
&=\sum^{m}_{j=1}\sum_{a_j:0,1}\frac{1}{4}\|\psi^{(A_j)}_{a_j}-
\calR_j\circ\calR_\calM\circ\calE_{A\rightarrow A'R_B}\circ\calE_j(\psi^{(A_j)}_{a_j})\|_1\label{BHH1}
}
where $p_{\vect{a}}(\vect{a}'):=\Tr[P_{\vect{a}'}\rho'^{A'R_B}_{\vect{a}}]$, $\rho'^{A'R_B}_{\vect{a}}:=\calE_{A\rightarrow A'R_B}(\psi_{\vect{a}})$, $\calE_{A\rightarrow A'R_B}(...):=\Tr_{B'}[U\otimes 1_{R_B}(...\otimes \Phi_{BR_B})U^\dagger\otimes 1_{R_B}]$, $\rho''^{(A)}_{\vect{a}}:=\sum_{\vect{a}'}p'_{\vect{a}}(\vect{a}')\psi_{\vect{a}'}=\calR_\calM\circ\calE_{A\rightarrow A'R_B}(\psi_{\vect{a}})$, 
$\psi^{(A_j)}_{\vect{a}}:=\Tr_{\neg A_j}[\psi_{\vect{a}}]$, $\psi^{(A_j)}_{\vect{a}'}:=\Tr_{\neg A_j}[\psi_{\vect{a}'}]$, $\rho''^{(A_j)}_{\vect{a}}:=\Tr_{\neg A_j}[\rho''^{(A)}_{\vect{a}}]$, and $\calE_j$ and $\calR_j$ are CPTP maps from $A_j$ to $A$ and $A$ to $A_j$ such that $\calE_j(...):=...\otimes_{i:i\ne j}\frac{\psi^{(A_i)}_0+\psi^{(A_i)}_1}{2}$ and  $\calR_j(...):=\Tr_{\neg A_j}[...]$, respectively.

Let us evaluate $\sum_{a_j:0,1}\frac{1}{4}\|\psi^{(A_j)}_{a_j}-
\calR_j\circ\calR_\calM\circ\calE_{A\rightarrow A'R_B}\circ\calE_j(\psi^{(A_j)}_{a_j})\|_1$.
We first remark that, using $2D_F^2(\psi,\sigma)\le\|\psi-\sigma\|_1$ which holds for an arbitrary pure state $\psi$ and a mixed state $\sigma$~\cite{Nielsen2002quantum_computation}, the following relation holds for all $1\le j\le m$ :
\eq{
\sum_{a_j:0,1}\frac{1}{4}\|\psi^{(A_j)}_{a_j}-
\calR_j\circ\calR_\calM\circ\calE_{A\rightarrow A'R_B}\circ\calE_j(\psi^{(A_j)}_{a_j})\|_1\ge\sum_{a_j:0,1}\frac{1}{2}D^2_F(\psi^{(A_j)}_{a_j},
\calR_j\circ\calR_\calM\circ\calE_{A\rightarrow A'R_B}\circ\calE_j(\psi^{(A_j)}_{a_j})).
}
Therefore, when we substitute $\{1/2,\psi^{(A_j)}_{a_j}\}_{a_j=0,1}$ and $\calE_{A\rightarrow A'R_B}\circ\calE_j$ for the test ensemble $\{p_k,\rho_k\}$ and $\calE$ in the \eqref{SIQ-Cini}, respectively, we obtain
\eq{
\sum_{a_j:0,1}\frac{1}{4}\|\psi^{(A_j)}_{a_j}-
\calR_j\circ\calR_\calM\circ\calE_{A\rightarrow A'R_B}\circ\calE_j(\psi^{(A_j)}_{a_j})\|_1&\ge\delta^2\non
&\ge\frac{2\calC^2}{(\Delta_3+\sqrt{\calF_{\Phi_{BR_B}\otimes\rho_{\neg A_j}}(X_B\otimes1_{R_B}\otimes1_{\neg A_j}+X_{\neg A_j}\otimes1_{R_B}\otimes1_{B})})^2}.\label{BHH2}
}
Here $\delta$, $\calC$ and $\Delta_3$ are defined for the test ensembles $\{1/2,\psi^{(A_j)}_{a_j}\}_{a_j=0,1}$ and the channel $\calE_{A\rightarrow A'R_B}\circ\calE_j$. We also defined $\rho_{\neg A_j}:=\otimes_{i:i\ne j}\frac{\psi^{(A_i)}_0+\psi^{(A_i)}_1}{2}$.
Therefore, we only have to evaluate $\calC$, $\Delta_3$ and $\sqrt{\calF_{\Phi_{BR_B}\otimes\rho_{\neg A_j}}(X_B\otimes1_{R_B}\otimes1_{\neg A_j}+X_{\neg A_j}\otimes1_{R_B}\otimes1_{B})}$.
To conclude first, these three quantities are bounded as follows:
\eq{
2\calC^2&\ge\frac{\gamma^2 (n-1)^2}{4}\label{C_B}\\
\Delta_3&\le\gamma n+1.5\sqrt{N}\label{D_B}\\
\sqrt{\calF_{\Phi_{BR_B}\otimes\rho_{\neg A_j}}(X_B\otimes1_{R_B}\otimes1_{\neg A_j}+X_{\neg A_j}\otimes1_{R_B}\otimes1_{B})}
&=
\sqrt{\calF_{\Phi_{BR_B}}(X_B\otimes1_{R_B})}
\non
&\le\sqrt{N}\label{F_B}
}
Combining these three bounds, \eqref{BHH1} and \eqref{BHH2}, we obtain \eqref{BH_Hamming} as follows:
\eq{
\delta_H&\ge\sum^{m}_{j=1}\sum_{a_j:0,1}\frac{1}{4}\|\psi^{(A_j)}_{a_j}-
\calR_j\circ\calR_\calM\circ\calE_{A\rightarrow A'R_B}\circ\calE_j(\psi^{(A_j)}_{a_j})\|_1\non
&\ge m\times\frac{2\calC^2}{(\Delta_3+\sqrt{\calF_{\Phi_{BR_B}\otimes\rho_{\neg A_j}}(X_B\otimes1_{R_B}\otimes1_{\neg A_j}+X_{\neg A_j}\otimes1_{R_B}\otimes1_{B})})^2}\non
&\ge\frac{m}{4}\frac{1}{\left(\frac{n}{n-1}+\frac{2.5\sqrt{N}}{\gamma(n-1)}\right)^2}\non
&=\frac{m}{4}\frac{1}{\left(1+\frac{\gamma+2.5\sqrt{N}}{\gamma(n-1)}\right)^2}\non
&=\frac{m}{4}\frac{1}{\left(1+\frac{1}{a\gamma}\frac{\frac{\gamma}{\sqrt{N}}+2.5}{1-\frac{1}{n}}\right)^2}\non
&\ge m\times\frac{1}{4\left(1+\frac{3}{a\gamma}\right)^2}.
}
Here, in the final line we used $\gamma\le1$, $N\ge 1000$, $n=a\sqrt{N}$ and $a\ge2$.

Finally, let us derive \eqref{C_B}--\eqref{F_B}.
We first show \eqref{F_B}. Note that $\calF_{\Phi_{BR_B}}(X_B\otimes1_{R_B})=4V_{\rho_B}(X_B)$. Since $\Phi_{BR_B}$ is a maximally entangled state, $V_{\rho_B}(X_B)\le N/4$. Therefore, \eqref{F_B} clearly holds.
Next, we derive \eqref{C_B} and \eqref{D_B}.
We first note that
\eq{
\calE^\dagger_{A\rightarrow A'R_B}(X_{A'}\otimes1_{R_B})&=\Tr_{BR_B}[(U^\dagger X_{A'}\otimes1_{B'}U\otimes1_{R_B})1_{A}\otimes\Phi_{BR_B}]\non
&=\Tr_{B}[(U^\dagger X_{A'}\otimes1_{B'}U)1_{A}\otimes\rho_{B}]\non
&=\calE^\dagger_{A\rightarrow A'}(X_{A'}),\label{subs1}
}
where $\calE_{A\rightarrow A'}(...):=\Tr_{B'R_B}[U\otimes1_{R_B}(...\otimes\Phi_{BR_B})U^\dagger\otimes1_{R_B}]$.
In the same way, we obtain
\eq{
\calE^\dagger_{A\rightarrow A'R_B}(X^2_{A'}\otimes1_{R_B})=\calE^\dagger_{A\rightarrow A'}(X^2_{A'}).\label{subs2}
}
We also remark that
\eq{
\calE_{j}^{\dagger}(...)=\Tr_{\neg A_j}[(...)1_{A_j}\otimes_{i:i\ne j}\frac{\psi^{(A_i)}_0+\psi^{(A_i)}_1}{2}].
}

Now, let us evaluate $Y=X_{A_j}-\calE_{j}^{\dagger}\circ\calE^\dagger_{A\rightarrow A'R_B}(X_{A'}\otimes1_{R_B})$ to derive \eqref{C_B}.
Although $\calE^\dagger_{A\rightarrow A'R_B}$ is not covariant, $Y$ can be written as $Y=X_{A_j}-\calE_{j}^{\dagger}\circ\calE^\dagger_{A\rightarrow A'}(X_{A'})$ due to \eqref{subs1}.
Since $\calE_{A\rightarrow A'}\circ\calE_j$ is covariant, the operator $\calE^\dagger_{A\rightarrow A'}(X_{A'})$ commutes with $X_{A_j}$. Therefore, we can describe $Y$ as follows:
\eq{
Y&=\sum_{i,a}z_{i,a|A_j}\ket{i,a}\bra{i,a}_{A_j}.
}
where $\ket{i,a}_{A_j}$ is the eivenvector of $X_{A_j}$.

Let us evaluate $z_{i,a|A_j}$.
First, we can evaluate the $z_{i,a}$ as follows:
\eq{
z_{i,a|A_j}&=\bra{i,a}X_{A_j}-\calE^\dagger_{j}\circ\calE^\dagger_{A\rightarrow A'}(X_{A'})\ket{i,a}_{A_j}\non
&=x_{A_j}(i)-\ex{X_{A'}}_{\calE_{A\rightarrow A'}\circ\calE_j(\ket{i,a}\bra{i,a}_{A_j})}\non
&\stackrel{(a)}{\approx_{\epsilon}}x_{A_j}(i)-(1-\gamma)(x_{A_j}(i)+\ex{X_B}_{\rho_B}+\ex{X_{\neg A_j}}_{\rho_{\neg A_j}})\non
&=\gamma x_{A_j}(i)-(1-\gamma)(\ex{X_B}_{\rho_B}+\ex{X_{\neg A_j}}_{\rho_{\neg A_j}}).\label{Ydes}
}
Here we used \eqref{ex-sc-eq_TTKS_B} in (a).
Therefore, 
\eq{
Y=\gamma X_{A_j}+(1-\gamma)(\ex{X_B}_{\rho_B}+\ex{X_{\neg A_j}}_{\rho_{\neg A_j}})I_{A_j}+\hat{\epsilon}.
\label{Yform}
}
Here $\hat{\epsilon}$ is an Hermitian operator satisfying $\|\hat{\epsilon}\|_\infty\le\epsilon$ and  $[\hat{\epsilon},X_{A_j}]=0$.
We derive \eqref{D_B} as follows:
\eq{
2\calC^2&=\Tr\left[Y\psi^{(A_j)}_{0}Y\psi^{(A_j)}_{1}\right]\non
&=|\bra{\psi^{(A_j)}_{0}}Y\ket{\psi^{(A_j)}_{1}}|^2\non
&=|\bra{\psi^{(A_j)}_{0}}(\gamma X_{A_j}+\hat{\epsilon})\ket{\psi^{(A_j)}_{1}}|^2\non
&\stackrel{(a)}{\ge}\frac{\left(\gamma n-\epsilon\right)^2}{4}\non
&\stackrel{(b)}{\ge}\frac{\gamma^2(n-1)^2}{4},
}
where we use $\|\hat{\epsilon}\|_\infty\le\epsilon$ and $[\hat{\epsilon},X_{A_j}]=0$ in $(a)$, and $\epsilon<1/(N+k)^2$ and $\gamma\ge1/(N+k)$ in $(b)$.

Next, we evaluate $\Delta_3$.
By definition, we can easily obtain
\eq{
\Delta_3\le\Delta_{Y}+\max_{\rho\in\mathrm{span}\{\psi^{(A_j)}_{a_j}\}_{a_j=0,1}}\sqrt{\calF_{\rho\otimes\rho_{\neg A_j}\otimes\Phi_{BR_B}}(\calE^\dagger_{A_j\rightarrow A'R_B}(X_{A'}\otimes1_{B'R_B})\otimes1_{BR_B\neg A_j}-U^\dagger X_{A'}\otimes 1_{B'}U\otimes1_{R_B})}\non
\le\Delta_{Y}+\max_{\ket{\psi_{A_j}}\in\mathrm{span}\{\ket{\psi^{(A_j)}_{a_j}}\}_{a_j=0,1}}\sqrt{\calF_{\psi_{A_j}\otimes\rho_{\neg A_j}\otimes\Phi_{BR_B}}(\calE^\dagger_{A_j\rightarrow A'R_B}(X_{A'}\otimes1_{B'R_B})\otimes1_{BR_B\neg A_j}-U^\dagger X_{A'}\otimes 1_{B'}U\otimes1_{R_B})}
}
Here we use the abbreviation $\calE_{A_j\rightarrow A'R_B}:=\calE_{A\rightarrow A_jR_B}\circ\calE_j$.
Due to $\Delta_Y\le \gamma n+\epsilon$ because of \eqref{Yform}, we only have to evaluate the second term in the right-hand side.
We can bound it as follows:
\eq{
&\sqrt{\calF_{\psi_{A_j}\otimes\rho_{\neg A_j}\otimes\Phi_{BR_B}}(\calE^\dagger_{A_j\rightarrow A'R_B}(X_{A'}\otimes1_{B'R_B})\otimes1_{\neg A_j}\otimes1_{BR_B}-U^\dagger X_{A'}\otimes 1_{B'}U\otimes1_{R_B})}\nonumber\\
&\le\sqrt{\calF_{\psi_{A_j}\otimes\rho_{\neg A_j}\otimes\Phi_{BR_B}}((1-\gamma)X_{A_j}\otimes1_{\neg A_j}\otimes1_{BR_B}-U^\dagger X_{A'}\otimes 1_{B'}U\otimes1_{R_B})}\non
&+\sqrt{\calF_{\psi_{A_j}\otimes\rho_{\neg A_j}\otimes\Phi_{BR_B}}(\calE^\dagger_{A_j\rightarrow A'R_B}(X_{A'}\otimes1_{B'R_B})\otimes1_{\neg A_j}\otimes1_{BR_B}-(1-\gamma)X_{A_j}\otimes1_{\neg A_j}\otimes1_{BR_B})}\non
&\stackrel{(a)}{=}\sqrt{\calF_{\psi_{A_j}\otimes\rho_{\neg A_j}\otimes\Phi_{BR_B}}((1-\gamma)X_{A_j}\otimes1_{\neg A_j}\otimes1_{BR_B}-U^\dagger X_{A'}\otimes 1_{B'}U\otimes1_{R_B})}\non
&+\sqrt{\calF_{\psi_{A_j}}(\calE^\dagger_{A_j\rightarrow A'R_B}(X_{A'}\otimes1_{B'R_B})-(1-\gamma)X_{A_j})}
\non
&\stackrel{(b)}{=}\sqrt{\calF_{\psi_{A_j}\otimes\rho_{\neg A_j}\otimes\Phi_{BR_B}}((1-\gamma)X_{A_j}\otimes1_{\neg A_j}\otimes1_{BR_B}-U^\dagger X_{A'}\otimes 1_{B'}U\otimes1_{R_B})}\non
&+\sqrt{\calF_{\psi_{A_j}}((1-\gamma)(\ex{X_B}_{\rho_B}+\ex{X_{\neg A_j}}_{\rho_{\neg A_j}})I_{A_j}+\hat{\epsilon})}\non
&\stackrel{(c)}{=}\sqrt{\calF_{\psi_{A_j}\otimes\rho_{\neg A_j}\otimes\Phi_{BR_B}}((1-\gamma)X_{A_j}\otimes1_{\neg A_j}\otimes1_{BR_B}-U^\dagger X_{A'}\otimes 1_{B'}U\otimes1_{R_B})}
+2\sqrt{V_{\psi_{A_j}}(\hat{\epsilon})}\non
&\le\sqrt{\calF_{\psi_{A_j}\otimes\rho_{\neg A_j}\otimes\Phi_{BR_B}}((1-\gamma)X_{A_j}\otimes1_{\neg A_j}\otimes1_{BR_B}-U^\dagger X_{A'}\otimes 1_{B'}U\otimes1_{R_B})}
+2\epsilon
.\label{16}
}
Here we use $\calF_{\rho_A\otimes\rho_B}(X_A+X_B)=\calF_{\rho_A}(X_A)+\calF_{\rho_B}(X_B)$ \cite{Hansen} in (a), \eqref{Yform} in (b), $\calF_{\psi}(W)=4V_\psi(W)$ for an arbitrary pure state $\psi$ and an arbitrary Hermitian operator $W$ in (c).

Let us evaluate $\calF_{\psi_{A_j}\otimes\rho_{\neg A_j}\otimes\Phi_{BR_B}}((1-\gamma)X_{A_j}\otimes1_{\neg A_j}\otimes1_{BR_B}-U^\dagger X_{A'}\otimes 1_{B'}U\otimes1_{R_B})$. 
We take a decomposition $\rho_{\neg A_j}=\sum_iq_i\ket{\phi_i}\bra{\phi_i}$ satisfying $\calF_{\rho_{\neg A_j}}(X_{\neg A_j})=4\sum_iq_iV_{\phi_i}(X_{\neg A_j})$.
Using the decomposition, we obtain
\eq{
&\calF_{\psi_{A_j}\otimes\rho_{\neg A_j}\otimes\Phi_{BR_B}}((1-\gamma)X_{A_j}\otimes1_{\neg A_j}\otimes1_{BR_B}-U^\dagger X_{A'}\otimes 1_{B'}U\otimes1_{R_B})\nonumber\\
&\le 4\sum_{i}q_iV_{\psi_{A_j}\otimes\phi_{i}\otimes\Phi_{BR_B}}(Z_1-Z_2)\nonumber\\
&\le 4\sum_{i}q_i(
\bra{\psi_{A_j}}\bra{\phi_i}\bra{\Phi_{BR_B}}(Z_1-Z_2)^2\ket{\psi_{A_j}}\ket{\phi_i}\ket{\Phi_{BR_B}}
-\bra{\psi_{A_j}}\bra{\phi_i}\bra{\Phi_{BR_B}}(Z_1-Z_2)\ket{\psi_{A_j}}\ket{\phi_i}\ket{\Phi_{BR_B}}^2)\nonumber\\
&=4\sum_{i}q_i(
\bra{\psi_{A_j}}\bra{\phi_i}\calE^\dagger_{A\rightarrow A'}(X^2_{A'})-(1-\gamma)X_{A_j}\calE^\dagger_{A\rightarrow A'}(X_{A'})-\calE^\dagger_{A\rightarrow A'}(X_{A'})(1-\gamma)X_{A_j}+(1-\gamma)^2X^2_{A_j}\ket{\psi_{A_j}}\ket{\phi_i}
\nonumber\\
&-\bra{\psi_{A_j}}\bra{\phi_i}(\calE^\dagger_{A\rightarrow A'}(X_{A'})-(1-\gamma)X_{A_j})\ket{\psi_{A_j}}\ket{\phi_i}^2
\nonumber\\
&=4\sum_{i}q_i(\bra{\psi_{A_j}}\bra{\phi_i}\calE^\dagger_{A\rightarrow A'}(X^2_{A'})-\calE^\dagger_{A\rightarrow A'}(X_{A'})^2\ket{\psi_{A_j}}\ket{\phi_i}
+V_{\ket{\psi_{A_j}}\ket{\phi_i}}(\calE^\dagger_{A\rightarrow A'}(X_{A'})-(1-\gamma)X_{A_j})
)\nonumber\\
&=4\ex{\calE^\dagger_{A\rightarrow A'}(X^2_{A'})-\calE^\dagger_{A\rightarrow A'}(X_{A'})^2}_{\psi_{A_j}\otimes\rho_{\neg A_j}}
+4\sum_{i}q_iV_{\ket{\psi_{A_j}}\ket{\phi_i}}(\calE^\dagger_{A\rightarrow A'}(X_{A'})-(1-\gamma)X_{A_j}),\label{S131}
}
where $Z_1:=(1-\gamma)X_{A_j}\otimes1_{\neg A_j}\otimes1_{BR_B}$ and $Z_2:=U^\dagger X_{A'}\otimes 1_{B'}U\otimes1_{R_B}$.

To evaluate the second term in the RHS of \eqref{S131}, note that $\calE^\dagger_{A\rightarrow A'}(X_{A'})$ commutes with $X_A$, since $\calE_{A\rightarrow A'}$ is covariant.
Therefore, we can write $\calE^\dagger_{A\rightarrow A'}(X_{A'})$ as
\eq{
X_A-\calE^\dagger_{A\rightarrow A'}(X_{A'})=\sum_{i,a}z'_{i,a|A}\ket{i,a}\bra{i,a}_{A},
}
where $\ket{i,a}_A$ is an eigenvector of $X_{A}$ whose eigenvalue is $x_{A}(i)$.
We evaluate $z'_{i,a}$ as follows:
\eq{
z'_{i,a|A}&:=\bra{i,a}_{A}X_{A}-\calE^\dagger_{A\rightarrow A'}(X_{A'})\ket{i,a}_{A}\non
&=x_{A}(i)-\ex{X_{A'}}_{\calE_{A\rightarrow A'}(\ket{i,a}\bra{i,a}_{A})}\non
&\stackrel{(a)}{\approx_{\epsilon}}x_{A}(i)-(1-\gamma)(x_{A}(i)+\ex{X_B}_{\rho_B})\non
&=\gamma x_{A_j}(i)-(1-\gamma)\ex{X_B}_{\rho_B}.
}
Here we used \eqref{ex-sc-eq_TTKS} in (a).
Therefore, we can write $\calE^\dagger_{A\rightarrow A'}(X_{A'})$ as 
\eq{
\calE^\dagger_{A\rightarrow A'}(X_{A'})=(1-\gamma)X_A+(1-\gamma)\ex{X_B}_{\rho_B}I_A+\hat{\epsilon}',
}
where $\epsilon'$ is an Hermitian operator on $A$ satisfying $\|\epsilon'\|\le\epsilon$.

Now, let us evaluate the second term in the RHS of \eqref{S131}:
\eq{
&4\sum_{i}q_iV_{\ket{\psi_{A_j}}\ket{\phi_i}}(\calE^\dagger_{A\rightarrow A'}(X_{A'})-(1-\gamma)X_{A_j})\non
&=4\sum_{i}q_iV_{\ket{\psi_{A_j}}\ket{\phi_i}}((1-\gamma)X_A+(1-\gamma)\ex{X_B}_{\rho_B}I_A+\hat{\epsilon}'-(1-\gamma)X_{A_j})\non
&=4\sum_{i}q_iV_{\ket{\psi_{A_j}}\ket{\phi_i}}((1-\gamma)X_{\neg A_j}+\hat{\epsilon}')\non
&\le4\sum_{i}q_i\left(\sqrt{V_{\ket{\psi_{A_j}}\ket{\phi_i}}((1-\gamma)X_{\neg A_j})}+\sqrt{V_{\ket{\psi_{A_j}}\ket{\phi_i}}(\hat{\epsilon}')}\right)^2\non
&\le4\sum_{i}q_i\left(\sqrt{V_{\ket{\phi_i}}((1-\gamma)X_{\neg A_j})}+\epsilon\right)^2
\non
&\le8\sum_{i}q_i\left(V_{\ket{\phi_i}}((1-\gamma)X_{\neg A_j})+\epsilon^2\right)
\non
&=2(1-\gamma)^2\calF_{\rho_{\neg A_j}}(X_{\neg A_j})+8\epsilon^2.
\label{17}
} 
Here in the last line we use $\calF_{\rho_{\neg A_j}}(X_{\neg A_j})=4\sum_iq_iV_{\phi_i}(X_{\neg A_j})$.
Therefore, we obtain
\eq{
&\calF_{\psi_{A_j}\otimes\rho_{\neg A_j}\otimes\Phi_{BR_B}}((1-\gamma)X_{A_j}\otimes1_{\neg A_j}\otimes1_{BR_B}-U^\dagger X_{A'}\otimes 1_{B'}U\otimes1_{R_B})\nonumber\\
&\le
4\ex{\calE^\dagger_{A\rightarrow A'}(X^2_{A'})-\calE^\dagger_{A\rightarrow A'}(X_{A'})^2}_{\psi_{A_j}\otimes\rho_{\neg A_j}}
+2(1-\gamma)^2\calF_{\rho_{\neg A_j}}(X_{\neg A_j})+8\epsilon^2.
}
We remark this inequality holds even if $[\rho_{\neg A_j},X_{\neg A_j}]\ne0$.
Since $\rho_{\neg A_j}=\otimes_{i: i\ne j}\frac{\psi^{(A_i)}_0+\psi^{(A_i)}_1}{2}$, $\calF_{\rho_{\neg A_j}}(X_{\neg A_j})=0$ holds. Therefore, we obtain 
\eq{
&\calF_{\psi_{A_j}\otimes\rho_{\neg A_j}\otimes\Phi_{BR_B}}((1-\gamma)X_{A_j}\otimes1_{\neg A_j}\otimes1_{BR_B}-U^\dagger X_{A'}\otimes 1_{B'}U\otimes1_{R_B})\nonumber\\
&\le
4\ex{\calE^\dagger_{A\rightarrow A'}(X^2_{A'})-\calE^\dagger_{A\rightarrow A'}(X_{A'})^2}_{\psi_{A_j}\otimes\rho_{\neg A_j}}
+8\epsilon^2.
}

Let us give an upper bound of  $\ex{\calE^\dagger_{A\rightarrow A'}(X^2_{A'})-\calE^\dagger_{A\rightarrow A'}(X_{A'})^2}_{\psi_{A_j}\otimes\rho_{\neg A_j}}$.
We remark that $\calE^\dagger_{A\rightarrow A'}(X^2_{A'})-\calE^\dagger_{A\rightarrow A'}(X_{A'})^2$ commutes with $X_A$, since $\calE_{A\rightarrow A'}$ is a covariant operation.
Therefore, we can write $\calE^\dagger_{A\rightarrow A'}(X^2_{A'})-\calE^\dagger_{A\rightarrow A'}(X_{A'})^2$ as 
\eq{
\calE^\dagger_{A\rightarrow A'}(X^2_{A'})-\calE^\dagger_{A\rightarrow A'}(X_{A'})^2=\sum_{i,a}z''_{i,a}\ket{i,a}\bra{i,a}_{A}.
}
We evaluate $z''_{i,a}$ as follows:
\eq{
\sqrt{z''_{i,a}}&=\sqrt{\bra{i,a}\calE^\dagger_{A\rightarrow A'}(X^2_{A'})-\calE^\dagger_{A\rightarrow A'}(X_{A'})^2\ket{i,a}}\nonumber\\
&\stackrel{(a)}{\le}\sqrt{\ex{X^2_{A'}}_{\calE_{A\rightarrow A'}(\ket{i,a}\bra{i,a})}-\ex{X_{A'}}_{\calE_{A\rightarrow A'}(\ket{i,a}\bra{i,a})}^2}\nonumber\\
&=\sqrt{V_{\calE_{A\rightarrow A'}(\ket{i,a}\bra{i,a})}(X_{A'})}\nonumber\\
&\stackrel{(b)}{=}\sqrt{V_{\sum_{j,b}r_{j,b}\rho'_{A'|i,a,j,b,U}}(X_{A'})}\nonumber\\
&\stackrel{(c)}{=}\sqrt{\sum_{j,b}r_{j,b}V_{\rho'_{A'|i,a,j,b,U}}(X_{A'})+V_{\{r_{j,b}\}}(\ex{X_{A'}}_{\rho'_{A'|i,a,j,b,U}})}\nonumber\\
&=\sqrt{\sum_{j,b}r_{j,b}V_{\rho'_{A'|i,a,j,b,U}}(X_{A'})+V_{\{r_{j,b}\}}((1-\gamma)(x_A(i)+x_B(j))+\epsilon_{i,a,j,b,U})}\nonumber\\
&\le\sqrt{\sum_{j,b}r_{j,b}V_{\rho'_{A'|i,a,j,b,U}}(X_{A'})}+\sqrt{V_{\{r_{j,b}\}}((1-\gamma)(x_A(i)+x_B(j)))}+\sqrt{V_{\{r_{j,b}\}}(\epsilon_{i,a,j,b,U})}\non
&\stackrel{(d)}{\le}\frac{1}{2}\sqrt{(1+\epsilon)\gamma(N+k)}+\frac{(1-\gamma)}{2}\sqrt{N}+\sqrt{V_{\{r_{j,b}\}}(\epsilon_{i,a,j,b,U})}
\label{com1_pre}
}
where in (a), we used $\Tr[\rho X]^2\leq \Tr[\rho X^2]$ for an arbitrary state $\rho$ and an observable $X$ obtained by applying Cauchy-Schwartz inequality to $\sqrt{\rho}$ and $\sqrt{\rho}X$ with the Hilbert-Schmidt inner product.
In (b), we defined $r_{j,b}$ as $\rho_B=\sum_{j,b}r_{j,b}\ket{j,b}\bra{j,b}$ and $\rho'_{A'|i,a,j,b,U}$ is defined in \eqref{S107}. We also used $\epsilon_{i,a,j,b,U}:=\ex{X_{A'}}_{\rho'_{A'|i,a,j,b,U}}-(1-\gamma)(x_A(i)+x_B(j))$, and $x_A(i):=\bra{i,a}X_A\ket{i,a}_A$ and $x_B(j):=\bra{j,b}X_B\ket{j,b}_B$. 
In (c), we defined $V_{\{r_{j,b}\}}(\ex{X_{A'}}_{\rho'_{A'|i,a,j,b,U}})$ as the variance of the values $\{\ex{X_{A'}}_{\rho'_{A'|i,a,j,b,U}}\}$ with the probability $\{r_{j,b}\}$.
In (d), we used \eqref{exa}.

Let us evaluate the RHS of \eqref{com1_pre}.
Note that $0\le\epsilon_{i,a,j,b,U}\le\epsilon\le1/(N+k)^2$ for $15\le i+j\le N+k-15$, and that since $\rho_B$ is the maximally mixed state, the probability distribution $\{r_{j,b}\}$ satisfies the large deviation property and thus $\sum_{(j,b):j\le 15, j\ge N-15}r_{j,b}=O(e^{-\alpha_B N})$ holds for some positive constant $\alpha_B>0$.
More specifically, due to $N\ge 10^3$ and $N\ge k$, the following inequality holds.
\eq{
(N+k)^8\sum_{(j,b):j\le 15, j\ge N-15}r_{j,b}
&=(N+k)^8\sum_{(j,b):j\le 15, j\ge N-15}\frac{\binom{N}{j}}{2^N}
\non
&\le (2N)^8\sum_{(j,b):j\le 15, j\ge N-15}\frac{\binom{N}{j}}{2^N}\non
&\le \left.(2N)^8\sum_{(j,b):j\le 15, j\ge N-15}\frac{\binom{N}{j}}{2^N}\right|_{N=1000}\non
&=3.33889\times10^{-242}\non
&\le1.
}
Therefore, we obtain
\eq{
\sum_{(j,b):j\le 15, j\ge N-15}r_{j,b}\le\frac{1}{(N+k)^8}\le\frac{\epsilon^2}{(N+k)^2}.
}
Combining the above and $\epsilon_{i,a,j,b,U}\le\|X_A+X_B\|_{\infty}\le (N+k)$ for all $i,a,j,b,U$, we obtain
\eq{
\sqrt{V_{\{r_{j,b}\}}(\epsilon_{i,a,j,b,U})}\le2\epsilon.
}
Using this relation, we obtain
\eq{
\frac{1}{2}\sqrt{(1+\epsilon)\gamma(N+k)}+\frac{(1-\gamma)}{2}\sqrt{N}+\sqrt{V_{\{r_{j,b}\}}(\epsilon_{i,a,j,b,U})}&\le\frac{1}{2}\sqrt{(1+\epsilon)\gamma(N+k)}+\frac{(1-\gamma)}{2}\sqrt{N}+2\epsilon
\non
&\le1.45\sqrt{N}.\label{com1}
}
Here we used $k\le N$ and $\sqrt{N}\le N/30$ (note that now we are showing that \eqref{BH_Hamming} holds when $N\ge 10^3$ and $k\le N$).

Therefore, we obtain
\eq{
&\calF_{\psi_{A_j}\otimes\rho_{\neg A_j}\otimes\Phi_{BR_B}}((1-\gamma)X_{A_j}\otimes1_{\neg A_j}\otimes1_{BR_B}-U^\dagger X_{A'}\otimes 1_{B'}U\otimes1_{R_B})\nonumber\\
&\le
4\ex{\calE^\dagger_{A\rightarrow A'}(X^2_{A'})-\calE^\dagger_{A\rightarrow A'}(X_{A'})^2}_{\psi_{A_j}\otimes\rho_{\neg A_j}}
+8\epsilon^2\non
&\le 1.5\sqrt{N}.
}
Therefore, we obtain \eqref{D_B}.
\end{proofof}

\end{widetext}

\end{document}